\documentclass{article}
\usepackage{fullpage}

\usepackage[T1]{fontenc}

\usepackage{booktabs}
\usepackage{float} %

\usepackage{amsmath,amsfonts,amssymb,xcolor,subcaption}

\allowdisplaybreaks

\usepackage{hyperref}
\hypersetup{
     colorlinks   = true,
     urlcolor    = teal,
	 citecolor = teal,
	 linkcolor = teal
}

\usepackage{amsthm}
\usepackage{mathtools}
\usepackage{thmtools}
\usepackage{tikz}
\usetikzlibrary{arrows,shapes,calc,positioning}

\newcommand{\e}{\varepsilon}
\newcommand{\Prob}{\mathbb{P}}
\newcommand{\Exp}[1]{\mathbb{E}\left[#1\right]}

\newcommand{\imax}{i_{\mathit{max}}}

\usepackage{thm-restate}

\usepackage[capitalize]{cleveref}
\newtheorem{theorem}{Theorem}
\newtheorem{lem}[theorem]{Lemma}
\crefname{lem}{Lemma}{Lemmas}

\begin{document}

\title{The Fluid Mechanics of Liquid Democracy}

\author{Paul G\"olz \and Anson Kahng \and Simon Mackenzie \and Ariel D.\ Procaccia}
\date{}

\maketitle
\begin{abstract}
\emph{Liquid democracy} is the principle of making collective decisions by letting agents transitively delegate their votes. Despite its significant appeal, it has become apparent that a weakness of liquid democracy is that a small subset of agents may gain massive influence. To address this, we propose to change the current practice by allowing agents to specify multiple delegation options instead of just one. Much like in nature, where --- fluid mechanics teaches us --- liquid maintains an equal level in connected vessels, so do we seek to control the flow of votes in a way that balances influence as much as possible. Specifically, we analyze the problem of choosing delegations to approximately minimize the maximum number of votes entrusted to any agent, by drawing connections to the literature on confluent flow. We also introduce a random graph model for liquid democracy, and use it to demonstrate the benefits of our approach both theoretically and empirically.
\end{abstract}

\section{Introduction}\label{sec:intro}
\emph{Liquid democracy} is a potentially disruptive approach to democratic decision making. As in direct democracy, agents can vote on every issue by themselves. Alternatively, however, agents may delegate their vote, i.e., entrust it to any other agent who then votes on their behalf. Delegations are transitive; for example, if agents $2$ and $3$ delegate their votes to $1$, and agent $4$ delegates her vote to $3$, then agent $1$ would vote with the weight of all four agents, including herself. Just like representative democracy, this system allows for separation of labor, but provides for stronger accountability: Each delegator is connected to her transitive delegate by a path of personal trust relationships, and each delegator on this path can withdraw her delegation at any time if she disagrees with her delegate's choices.

Although the roots of liquid democracy can be traced back to the work of Miller~\cite{Mill69}, it is only in recent years that it has gained recognition among practitioners. Most prominently, the German Pirate Party adopted the platform \emph{LiquidFeedback} for internal decision-making in 2010. At the highest point, their installation counted more than 10\,000 active users~\cite{KKHS+15}. More recently, two parties --- the Net Party in Argentina, and Flux in Australia --- have run in national elections on the promise that their elected representatives would vote according to decisions made via their respective liquid-democracy-based systems. Although neither party was able to win any seats in parliament, their bids enhanced the promise and appeal of liquid democracy. 

However, these real-world implementations also exposed a weakness in the liquid democracy approach: Certain individuals, the so-called super-voters, seem to amass enormous weight, whereas most agents do not receive any delegations. In the case of the Pirate Party, this phenomenon is illustrated by an article in \href{http://www.spiegel.de/international/germany/liquid-democracy-web-platform-makes-professor-most-powerful-pirate-a-818683.html}{Der Spiegel}, according to which one particular super-voter's ``vote was like a decree,'' even though he held no office in the party. As Kling et al.~\cite{KKHS+15} describe, super-voters were so controversial that ``the democratic nature of the system was questioned, and many users became inactive.''
Besides the negative impact of super-voters on perceived legitimacy, super-voters might also be more exposed to bribing. Although delegators can retract their delegations as soon as they become aware of suspicious voting behavior, serious damage might be done in the meantime.
Furthermore, if super-voters jointly have sufficient power, they might find it more efficient to organize majorities through deals between super-voters behind closed doors, rather than to try to win a broad majority through public discourse.
Finally, recent work by Kahng et al.~\cite{KMP18} indicates that, even if delegations go only to more competent agents, a high concentration of power might still be harmful for social welfare, by neutralizing benefits corresponding to the Condorcet Jury Theorem.

While all these concerns suggest that the weight of super-voters should be limited, the exact metric to optimize for varies between them and is often not even clearly defined. For the purposes of this paper, we choose to minimize the weight of the heaviest voter. As is evident in the Spiegel article, the weight of individual voters plays a direct role in the perception of super-voters. But even beyond that, we are confident that minimizing this measure will lead to substantial improvements across all presented concerns.

Just how can the maximum weight be reduced? One approach might be to restrict the power of delegation by imposing caps on the weight. However, as argued by Behrens et al.~\cite{BKNS14}, delegation is always possible by coordinating outside of the system and copying the desired delegate's ballot. Pushing delegations outside of the system would not alleviate the problem of super-voters, just reduce transparency.
Therefore, we instead adopt a voluntary approach: If agents are considering multiple potential delegates, all of whom they trust, they are encouraged to leave the decision for one of them to a centralized mechanism.
With the goal of avoiding high-weight agents in mind, our research challenge is twofold:
\begin{quote}
\emph{First, investigate the algorithmic problem of selecting delegations to minimize the maximum weight of any agent, and, second, show that allowing multiple delegation options does indeed provide a significant reduction in the maximum weight compared to the status quo.}
\end{quote} 
Put another (more whimsical) way, we wish to design liquid democracy systems that emulate the \emph{law of communicating vessels}, which asserts that liquid will find an equal level in connected containers.

\subsection{Our Approach and Results}
We formally define our problem in \cref{sec:algorithm}.
In addition to minimizing the maximum weight of any voter, we specify how to deal with delegators whose vote cannot possibly reach any voter.
In general, our problem is closely related to minimizing congestion for confluent flow as studied by Chen et al.~\cite{chen}.
Not only does this connection suggest an optimal algorithm based on mixed integer linear programming, but we also get a polynomial-time $(1 + \log |V|)$-approximation algorithm, where $V$ is the set of voters.\footnote{Throughout this paper, let $\log$ denote the natural logarithm.}
In addition, we show that approximating our problem to within a factor of $\frac{1}{2} \, \log_2 |V|$ is NP-hard.

In \cref{sec:probmodel}, to evaluate the benefits of allowing multiple delegations, we propose a probabilistic model for delegation behavior --- inspired by the well-known \emph{preferential attachment} model~\cite{BA99} --- in which we add agents successively.
With a certain probability $d$, a new agent delegates; otherwise, she votes herself.
If she delegates, she chooses $k$ many delegation options among the previously inserted agents.
A third parameter $\gamma$ controls the bias of this selection towards agents who already receive many delegations.
Assuming $\gamma = 0$, i.e., that the choice of delegates is unbiased, we prove that allowing two choices per delegator ($k = 2$) asymptotically leads to dramatically lower maximum weight than classical liquid democracy ($k = 1$).
In the latter case, with high probability, the maximum weight is at least $\Omega(t^\beta)$ for some $\beta > 0$, whereas the maximum weight in the former case is only $\mathcal{O}(\log \log t)$ with high probability, where $t$ denotes simultaneously the time step of the process and the number of agents.
Our analysis draws on a phenomenon called the \emph{power of choice} that can be observed in many different load balancing models. In fact, even a greedy mechanism that selects a delegation option to locally minimize the maximum weight as agents arrive exhibits this asymptotic behavior, which upper-bounds the maximum weight for optimal resolution.

In \cref{sec:simulations}, we complement our theoretical findings with empirical results.
Our simulations demonstrate that our approach continues to outperform classical preferential attachment for higher values of $\gamma$.
We also show that the most substantial improvements come from increasing $k$ from one to two, i.e., that increasing $k$ even further only slightly reduces the maximum weight.
We continue to see these improvements in terms of maximum weight even if just some fraction of delegators gives two options while the others specify a single delegate.
Finally, we compare the optimal maximum weight with the maximum weight produced by the approximation algorithm and greedy heuristics.

\subsection{Related Work}\label{sec:related}
Kling et al.~\cite{KKHS+15} conduct an empirical investigation of the existence and influence of super-voters.  The analysis is based on daily data dumps, from 2010 until 2013, of the German Pirate Party installation of LiquidFeedback. As noted above, Kling et al.\ find that super-voters exist, and have considerable power. The results do suggest that super-voters behave responsibly, as they ``do not
fully act on their power to change the outcome of votes, and
they vote in favour of proposals with the majority of voters
in many cases.'' Of course, this does not contradict the idea that a balanced distribution of power would be desirable.

There are only a few papers that provide theoretical analyses of liquid democracy~\cite{Green15,CG17,KMP18}. 
We would like to stress the differences between our approach and the one adopted by Kahng et al.~\cite{KMP18}.
They consider binary issues in a setting with an objective ground truth, i.e., there is one ``correct'' outcome and one ``incorrect'' outcome. In this setting, voters are modeled as biased coins that each choose the correct outcome with an individually assigned probability, or \textit{competence level}.
The authors examine whether liquid democracy can increase the probability of making the right decision over direct democracy by having less competent agents delegate to more competent ones.
By contrast, our work is completely independent of the (strong) assumptions underlying the results of Kahng et al. In particular, our approach is agnostic to the final outcome of the voting process, does not assume access to information that would be inaccessible in practice, and is compatible with any number of alternatives and choice of voting rule used to aggregate votes. In other words, the goal is not to use liquid democracy to promote a particular outcome, but rather to adapt the process of liquid democracy such that more voices will be heard.

\section{Algorithmic Model and Results}\label{sec:alg}
\label{sec:algorithm}
Let us consider a delegative voting process where agents may specify multiple potential delegations. This gives rise to a directed graph, whose nodes represent agents and whose edges represent potential delegations. In the following, we will conflate nodes and the agents they represent.
A distinguished subset of nodes corresponds to agents who have voted directly, the \emph{voters}.
Since voters forfeit the right to delegate, the voters are a subset of the sinks of the graph.
We call all non-voter agents \emph{delegators}.

Each agent has an inherent voting weight of 1.
When the delegations will have been resolved, the weight of every agent will be the sum of weights of her delegators plus her inherent weight.
We aim to choose a delegation for every delegator in such a way that the maximum weight of any voter is minimized. 

This task closely mirrors the problem of congestion minimization for confluent flow (with infinite edge capacity):
There, a flow network is also a finite directed graph with a distinguished set of graph sinks, the \emph{flow sinks}.
Every node has a non-negative \emph{demand}.
If we assume unit demand, this demand is 1 for every node.
Since the flow is confluent, for every non-sink node, the algorithm must pick exactly one outgoing edge, along which the flow is sent.
Then, the \emph{congestion} at a node $n$ is the sum of congestions at all nodes who direct their flow to $n$ plus the demand of $n$.
The goal in congestion minimization is to minimize the maximum congestion at any flow sink. (We remark that the close connection between our problem and confluent flow immediately suggests a variant corresponding to splittable flow; we discuss this variant at length in Section~\ref{sec:disc}.) 

In spite of the similarity between confluent flow and resolving potential delegations, the two problems differ when a node has no path to a voter / flow sink.
In confluent flow, the result would simply be that no flow exists.
In our setting however, this situation can hardly be avoided.
If, for example, several friends assign all of their potential delegations to each other, and if all of them rely on the others to vote, their weight cannot be delegated to any voter.
Our mechanism cannot simply report failure as soon as a small group of voters behaves in an unexpected way. Thus, it must be allowed to leave these votes unused. 
At the same time, of course, our algorithm should not exploit this power to decrease the maximum weight, but must primarily maximize the number of utilized votes.
We formalize these issues in the following section.

\subsection{Problem Statement}
All graphs $G = (N, E)$ mentioned in this section will be finite and directed.
Furthermore, they will be equipped with a subset $V \subseteq \mathit{sinks}(G)$.
For the sake of brevity, these assumptions will be implicit in the notion ``graph $G$ with $V$''.

Some of these graphs represent situations in which all delegations have already been resolved and in which each vote reaches a voter:
We call a graph $(N, E)$ with $V$ a \emph{delegation graph} if it is acyclic, its sinks are exactly the set $V$, and every other vertex has outdegree one.
In such a graph, define the \emph{weight} $w(n)$ of a node $n \in N$ as
\[ w(n) \coloneqq 1 + \sum_{(m,n) \in E} w(m). \]
This is well-defined because $E$ is a well-founded relation on $N$.

Resolving the delegations of a graph $G$ with $V$ can now be described as the \textsc{MinMaxWeight} problem: Among all delegation subgraphs $(N', E')$ of $G$ with voting vertices $V$ of maximum $|N'|$, find one that minimizes the maximum weight of the voting vertices.

\subsection{Connections to Confluent Flow}
\label{subsec:confflow}
We recall definitions from the flow literature as used by Chen et al.~\cite{chen}.
We slightly simplify the exposition by assuming unit demand at every node.

Given a graph $(N, E)$ with $V$, a \emph{flow} is a function $f : E \to \mathbb{R}_{\geq 0}$.
For any node $n$, set $\mathit{in}(n) \coloneqq \sum_{(m, n) \in E} f(m, n)$ and $\mathit{out}(n) \coloneqq \sum_{(n, m) \in E} f(n, m)$.
At every node $n \in N \setminus V$, a flow must satisfy \emph{flow conservation}:
\[ \mathit{out}(n) = 1 + \mathit{in}(n). \]
The congestion at any node $n$ is defined as $1 + \mathit{in}(n)$.
A flow is \emph{confluent} if every node has at most one outgoing edge with positive flow. We define \textsc{MinMaxCongestion} as the problem of finding a confluent flow on a given graph such that the maximum congestion is minimized. 

To relate the two presented problems, we need to refer to the parts of a graph $(N, E)$ with $V$ from which $V$ is reachable:
The \emph{active} nodes $\mathit{active}_V (N, E)$ are all $n \in N$ such that $n \mathrel{E^*} v$ for some $v \in V$.
The active subgraph is the restriction of $(N, E)$ to $\mathit{active}_V (N, E)$.
In particular, $V$ is part of this subgraph.

\begin{lem}
    \label{lem:restricttoactive}
    Let $G = (N, E)$ with $V$ be a graph.
    Its delegation subgraphs $(N', E')$ that maximize $|N'|$ are exactly the delegation subgraphs with $N' = \mathit{active}_V (N,E)$.
    At least one such subgraph exists.
\end{lem}
\begin{proof}
	First, we show that all nodes of a delegation subgraph are active. Indeed, consider any node $n_1$ in the subgraph.
	By following outgoing edges, we obtain a sequence of nodes $n_1 \, n_2 \dots$ such that $n_i$ delegates to $n_{i+1}$.
	Since the graph is finite and acyclic, this sequence must end with a vertex $n_j$ without outgoing edges.
	This must be a voter; thus, $n_1$ is active.

	Furthermore, there exists a delegation subgraph of $(N, E)$ with nodes exactly $\mathit{active}_V (N,E)$.
	Indeed, the shortest-paths-to-set-$V$ forest (with edges pointed in the direction of the paths) on the active subgraph is a delegation graph.

	By the first argument, all delegation subgraphs must be subgraphs of the active subgraph.
	By the second argument, to have the maximum number of nodes, they must include all nodes of this subgraph.
\end{proof}
\begin{lem}
    \label{lem:flowtodel}
    Let $(N, E)$ with $V$ be a graph %
    and let $f : E \to \mathbb{R}_{\geq 0}$ be a confluent flow (for unit demand).
    By eliminating all zero-flow edges from the graph, we obtain a delegation graph.
\end{lem}
\begin{proof}
	We first claim that the resulting graph is acyclic.
	Indeed, for the sake of contradiction, suppose that there is a cycle including some node $n$.
	Consider the flow out of $n$, through the cycle and back into $n$.
	Since the flow is confluent, and thus the flow cannot split up, the demand can only increase from one node to the next.
	As a result, $\mathit{in}(n) \geq \mathit{out}(n)$.
	However, by flow conservation and unit demand, $\mathit{out}(n) = \mathit{in}(n) + 1$, which contradicts the previous statement.
	
	Furthermore, the sinks of the graph are exactly $V$:
	By assumption, the nodes of $V$ are sinks in the original graph, and thus in the resulting graph.
	For any other node, flow conservation dictates that its outflow be at least its demand 1, thus every other node must have outgoing edges.
	
	Finally, every node not in $V$ must have outdegree 1.
	As detailed above, the outdegree must be at least 1.
	Because the flow was confluent, the outdegree cannot be greater.
	
	As a result of these three properties, we have a delegation graph.
\end{proof}

\begin{lem}
    \label{lem:deltoflow}
    Let $(N, E)$ with $V$ be a graph in which all vertices are active, and let $(N,E')$ be a delegation subgraph.
    Let $f : E \to \mathbb{R}_{\geq 0}$ be defined such that, for every node $n \in N \setminus V$ with (unique) outgoing edge $e\in E'$, $f(e) \coloneqq w(n)$.
    On all other edges $e\in E\setminus E'$, set $f(e) \coloneqq 0$. 
    Then, $f$ is a confluent flow.
\end{lem}
\begin{proof}
    For every non-sink, flow conservation holds by the definition of weight and flow.
    By construction, the flow must be confluent.
\end{proof}

\subsection{Algorithms}
The observations made above allow us to apply algorithms --- even approximation algorithms --- for \textsc{MinMaxCongestion} to our \textsc{MinMaxWeight} problem, that is, we can reduce the latter problem to the former.

\begin{theorem}
\label{thm:alg}
Let $\mathcal{A}$ be an algorithm for \emph{\textsc{MinMaxCongestion}} with approximation ratio $c\geq 1$. Let $\mathcal{A'}$ be an algorithm that, given $(N, E)$ with $V$, runs $\mathcal{A}$ on the active subgraph, and translates the result into a delegation subgraph by eliminating all zero-flow edges. Then $\mathcal{A'}$ is a $c$-approximation algorithm for \emph{\textsc{MinMaxWeight}}.
\end{theorem}

\begin{proof}
By \cref{lem:restricttoactive}, removing inactive parts of the graph does not change the solutions to \textsc{MinMaxWeight}, so we can assume without loss of generality that all vertices in the given graph are active. 

Suppose that the optimal solution for \textsc{MinMaxCongestion} on the given instance has maximum congestion $\alpha$. By \cref{lem:flowtodel}, it can be translated into a solution for \textsc{MinMaxWeight} with maximum weight $\alpha$. By \cref{lem:deltoflow}, the latter instance has no solution with maximum weight less than $\alpha$, otherwise it could be used to construct a confluent flow with the same maximum congestion. It follows that the optimal solution to the given \textsc{MinMaxWeight} instance has maximum weight $\alpha$. 

Now, $\mathcal{A}$ returns a confluent flow with maximum congestion at most $c\cdot \alpha$. Using Lemma~\ref{lem:flowtodel}, $\mathcal{A'}$ constructs a solution to \textsc{MinMaxWeight} with maximum weight at most $c\cdot \alpha$. Therefore, $\mathcal{A'}$ is a $c$-approximation algorithm. 
\end{proof}

Note that Theorem~\ref{thm:alg} works for $c=1$, i.e., even for exact algorithms. Therefore, it is possible to solve \textsc{MinMaxWeight} by adapting any exact algorithm for \textsc{MinMaxFlow}. For completeness we provide a mixed integer linear programming (MILP) formulation of the latter problem in Appendix~\ref{app:milp}.

Since the foregoing algorithm is based on solving an NP-hard problem, it might be too inefficient for typical use cases of liquid democracy with many participating agents.
Fortunately, it might be acceptable to settle for a slightly non-optimal maximum weight if this decreases computational cost.
To our knowledge, the best polynomial approximation algorithm for \textsc{MinMaxCongestion} is due to Chen et al.~\cite{chen} and achieves an approximation ratio of $1 + \log |V|$. Their algorithm starts by computing the optimal solution to the splittable-flow version of the problem, by solving a linear program.
The heart of their algorithm is a non-trivial, deterministic rounding mechanism.
This scheme drastically outperforms the natural, randomized rounding scheme, which leads to an approximation ratio of $\Omega(|N|^{1/4})$ with arbitrarily high probability \cite{chen06}.

\subsection{Hardness of Approximation}
In this section, we demonstrate the NP-hardness of approximating the MinMaxWeight problem to within a factor of $\frac{1}{2} \, \log_2 |V|$.
On the one hand, this justifies the absence of an exact polynomial-time algorithm.
On the other hand, this shows that the approximation algorithm is optimal up to a multiplicative constant.

\begin{theorem}
\label{thm:nphard}
It is NP-hard to approximate the \emph{\textsc{MinMaxWeight}} problem to a factor of $\frac{1}{2} \log_2 |V|$, even when each node has outdegree at most $2$.
\end{theorem}

Not surprisingly, we derive hardness via a reduction from \textsc{MinMaxCongestion}, i.e., a reduction in the opposite direction from the one given in Theorem~\ref{thm:alg}.
As shown by Chen et al.~\cite{chen}, approximating \textsc{MinMaxCongestion} to within a factor of $\frac{1}{2} \, \log_2 |V|$ is NP-hard. However, in our case, nodes have unit demands. Moreover, we are specifically interested in the case where each node has outdegree at most $2$, as in practice we expect outdegrees to be very small, and this case plays a special role in Section~\ref{sec:probmodel}. 

\begin{restatable}{lem}{restnphard}
\label{lem:nphard}
It is NP-hard to approximate the \emph{\textsc{MinMaxCongestion}} problem to a factor of $\frac{1}{2} \log_2 k$, where $k$ is the number of sinks, even when each node has unit demand and outdegree at most $2$.
\end{restatable}

The proof of Lemma~\ref{lem:nphard} is relegated to Appendix~\ref{app:nphard}. We believe the lemma is of independent interest, as it shows a surprising separation between the case of outdegree $1$ (where the problem is moot) and outdegree $2$, and that the asymptotically optimal approximation ratio is independent of degree. But it also allows us to prove Theorem~\ref{thm:nphard} almost directly. 

\begin{proof}[Proof of Theorem~\ref{thm:nphard}]
We reduce (gap) \textsc{MinMaxCongestion} with unit demand and outdegree at most $2$ to (gap) \textsc{MinMaxWeight} with outdegree at most $2$. First, we claim that if there are inactive nodes, there is no confluent flow. Indeed, let $n_1$ be an inactive node. 
For the sake of contradiction, suppose that there exists a flow $f$.
Follow the positive flow to obtain a sequence $n_1 \, n_2 \dots$.
By definition, none of the nodes reachable from $n_1$ can be a voter.
Since, by flow conservation and unit demand, each node must delegate, the sequence must be infinite.
As detailed in the proof of \cref{lem:flowtodel}, a confluent flow with unit demand cannot contain cycles.
Thus, the sequence contains infinitely many different nodes, which contradicts the finiteness of $G$.

Therefore, we can assume without loss of generality that in the given instance of \textsc{MinMaxCongestion}, all nodes are active (as the problem is still NP-hard). The reduction creates an instance of \textsc{MinMaxWeight} that has the same graph as the given instance of \textsc{MinMaxCongestion}. Using an analogous argument to Theorem~\ref{thm:alg} (reversing the roles of \cref{lem:flowtodel} and \cref{lem:deltoflow} in its proof), we see that this is a strict approximation-preserving reduction. 
\end{proof}
 
\section{Probabilistic Model and Results}
\label{sec:probmodel}

Our generalization of liquid democracy to multiple potential delegations aims to decrease the concentration of weight.
Accordingly, the success of our approach should be measured by its effect on the maximum weight in real elections.
Since, at this time, we do not know of any available datasets,\footnote{There is one relevant dataset that we know of, which was analyzed by Kling et al.~\cite{KKHS+15}. However, due to stringent privacy constraints, the data privacy officer of the German Pirate Party was unable to share this dataset with us.} we instead propose a probabilistic model for delegation behavior, which can serve as a credible proxy.
Our model builds on the well-known preferential attachment model, which generates graphs possessing typical properties of social networks.

The evaluation of our approach will be twofold:
In \cref{sec:lowersingle,sec:upperchoice}, for a certain choice of parameters in our model, we establish a striking separation between traditional liquid democracy and our system.
In the former case, the maximum weight at time $t$ is $\Omega(t^\beta)$ for a constant $\beta$ with high probability, whereas in the latter case, it is in $\mathcal{O}(\log \log t)$ with high probability, even if each delegator only suggests two options. For other parameter settings, we empirically corroborate the benefits of our approach in \cref{sec:simulations}.

\subsection{The Preferential Delegation Model}
Many real-world social networks have degree distributions that follow a power law~\cite{KNT10,New01}. Additionally, in their empirical study, Kling et al.~\cite{KKHS+15} observed that the weight of voters in the German Pirate Party was ``power law-like'' and that the graph had a very unequal indegree distribution. In order to meld the previous two observations in our liquid democracy delegation graphs, we adapt a standard preferential attachment model~\cite{BA99} for this specific setting. On a high level, our \emph{preferential delegation} model is characterized by three parameters: $0 < d < 1$, the probability of delegation; $k \geq 1$, the number of delegation options from each delegator; and $\gamma \geq 0$, an exponent that governs the probability of delegating to nodes based on current weight.

At time $t = 1$, we have a single node representing a single voter.
In each subsequent time step, we add a node for agent $i$ and flip a biased coin to determine her delegation behavior.
With probability $d$, she delegates to other agents.
Else, she votes independently.
If $i$ does not delegate, her node has no outgoing edges.
Otherwise, add edges to $k$ many i.i.d.\ selected, previously inserted nodes, where the probability of choosing node $j$ is proportional to $(\mathit{indegree}(j) + 1)^{\gamma}$.
Note that this model might generate multiple edges between the same pair of nodes, and that all sinks are voters.
\Cref{fig:graphs} shows example graphs for different settings of $\gamma$.

\begin{figure}
    \centering
    \begin{subfigure}[h]{.49\textwidth}
        \centering
        \includegraphics[height=.23\textwidth]{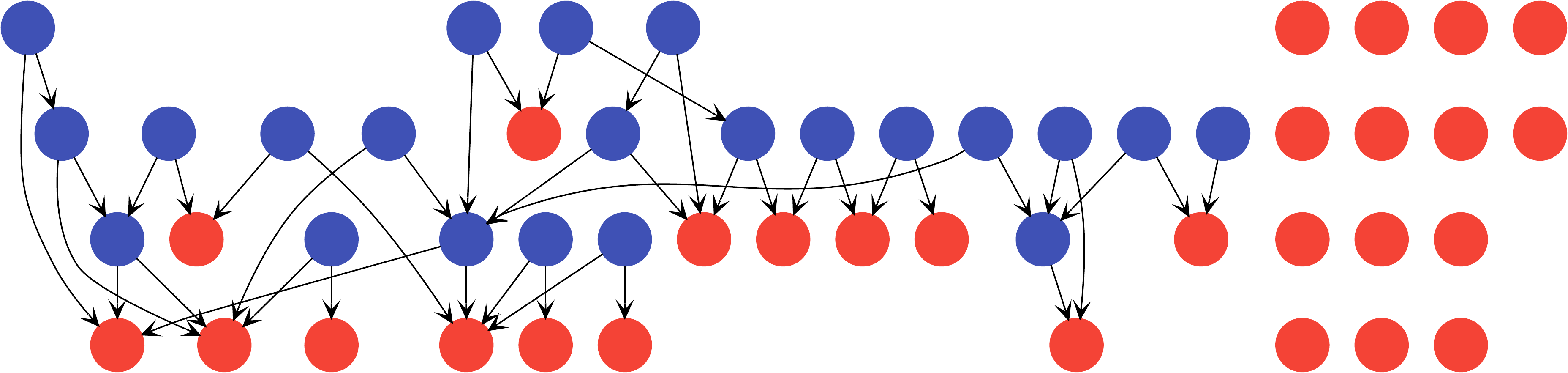}
        \caption{$\gamma = 0$}
    \end{subfigure}
    \begin{subfigure}[h]{.49\textwidth}
        \centering
        \includegraphics[height=.23\textwidth]{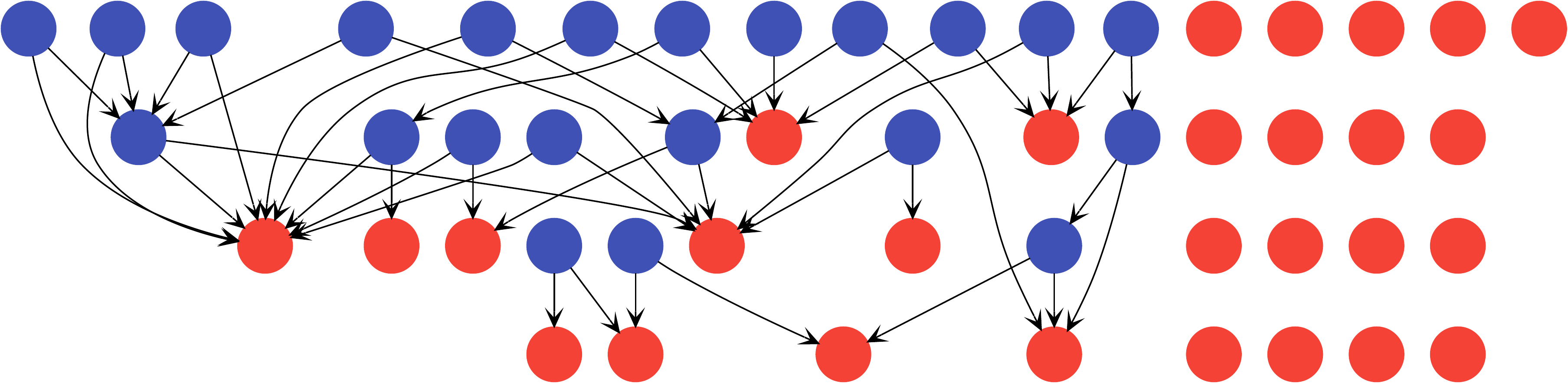}
        \caption{$\gamma = 1$}
    \end{subfigure}
    \caption{Example graphs generated by the preferential delegation model for $k=2$ and $d=0.5$.}
    \label{fig:graphs}
\end{figure}
In the case of $\gamma = 0$, which we term \emph{uniform delegation}, a delegator is equally likely to attach to any previously inserted node.
Already in this case, a ``rich-get-richer'' phenomenon can be observed, i.e., voters at the end of large networks of potential delegations will likely see their network grow even more.
Indeed, a larger network of delegations is more likely to attract new delegators.
In traditional liquid democracy, where $k = 1$ and all potential delegations will be realized, this explains the emergence of super-voters with excessive weight observed by Kling et al.~\cite{KKHS+15}.
We aim to show that for $k \geq 2$, the resolution of potential delegations can strongly outweigh these effects.
In this, we profit from an effect known as the ``power of two choices'' in load balancing described by Azar et al.~\cite{azar}.

For $\gamma > 0$, the ``rich-get-richer'' phenomenon additionally appears at the degrees of nodes.
Since the number of received potential delegations is a proxy for an agent's competence and visibility, new agents are more likely to attach to agents with high indegree.
In total, this is likely to further strengthen the inherent inequality between voters.
For increasing $\gamma$, the graph becomes increasingly flat, as a few super-voters receive nearly all delegations.
This matches observations from the LiquidFeedback dataset~\cite{KKHS+15} that ``the delegation network is slowly becoming less like a friendship network, and more like a bipartite networks of super-voters connected to normal voters.'' The special case of $\gamma = 1$ corresponds to preferential attachment as described by Barab\'asi and Albert~\cite{BA99}.

The most significant difference we expect to see between graphs generated by the preferential delegation model and real delegation graphs is the assumption that agents always delegate to more senior agents.
In particular, this causes generated graphs to be acyclic, which need not be the case in practice.
It does seem plausible that the majority of delegations goes to agents with more experience on the platform.
Even if this assumption should not hold, there is a second interpretation of our process if we assume --- as do Kahng et al.~\cite{KMP18} --- that agents can be ranked by competence and only delegate to more competent agents.
Then, we can think of the agents as being inserted in decreasing order of competence.
When a delegator chooses more competent agents to delegate to, her choice would still be biased towards agents with high indegree, which is a proxy for popularity.

In our theoretical results, we focus on the cases of $k=1$ and $k=2$, and assume $\gamma = 0$ to make the analysis tractable.
The parameter $d$ can be chosen freely between $0$ and $1$. Note that our upper bound for $k = 2$ directly translates into an upper bound for larger $k$, since the resolution mechanism always has the option of ignoring all outgoing edges except for the two first.
Therefore, to understand the effect of multiple delegation options, we can restrict our attention to $k = 2$.
This crucially relies on $\gamma = 0$, where potential delegations do not influence the probabilities of choosing future potential delegations.
Based on related results by Malyshkin and Paquette~\cite{malyshkin}, it seems unlikely that increasing $k$ beyond 2 will reduce the maximum weight by more than a constant factor.

\subsection{Lower Bounds for Single Delegation ($k = 1$, $\gamma = 0$)}
\label{sec:lowersingle}

As mentioned above, we first assume uniform delegation and a single delegation option per delegator, and derive a complementary lower bound on the maximum weight. To state our results rigorously, we say that a sequence $(\mathcal{E}_m)_m$ of events happens \emph{with high probability} if $\Prob[\mathcal{E}_m] \to 1$ for $m \to \infty$.
Since the parameter going to infinity is clear from the context, we omit it.

\begin{restatable}{theorem}{restlbwhp}
\label{thm:lbwhp}
	In the preferential delegation model with $k=1$, $\gamma=0$, and $d\in (0,1)$, with high probability, the maximum weight of any voter at time $t$ is in $\Omega(t^\beta)$, where $\beta > 0$ is a constant that depends only on $d$. 
\end{restatable}

We relegate the proof of \cref{thm:lbwhp} to \cref{app:lbwhp}.
Since bounding the expected value is conceptually clearer and more concise than a bound holding with high probability, we prove an analogous theorem in order to build intuition.

\begin{theorem}
	\label{thm:lbexp}
	In the preferential delegation model with $k=1$, $\gamma=0$, and $d\in (0,1)$, the expected maximum weight of any voter at time $t$ is in $\Omega(t^d)$. 
\end{theorem}
\begin{proof}[Proof of Theorem~\ref{thm:lbexp}]
	Let $w_t(i)$ denote the weight of node $i$ at time $t$.
	Clearly, $\mathbb{E}[\max_j w_t(j)] \geq \mathbb{E}[w_t(1)]$, and therefore we can lower-bound the expected maximum weight of any node by the expected weight of the first node.
	
	For $i \geq 1$, let $D_i$ denote the event that voter $i$ transitively delegates to voter $1$.
	In addition, denote $S_t \coloneqq w_t(1) = \sum_{i=1}^t D_i$. Our goal is to prove that $\mathbb{E}[S_t] \in \Theta(t^d)$.
	
	We begin by showing that the expected weight of voter $1$ satisfies the following recurrences:
	\begin{align}
	\label{eq:recur1}
	\mathbb{E}[S_1] &= 1 \\
	\label{eq:recur2}
	\mathbb{E}[S_{t+1}] &= \left(1 + \frac{d}{t}\right) \cdot \mathbb{E}[S_t]
	\end{align}
	
	Indeed, for \cref{eq:recur1}, voter $1$'s weight after one time step is always $1$. For \cref{eq:recur2}, by linearity of expectation, $\mathbb{E}[S_{t+1}] = \mathbb{E}[S_t] + \mathbb{P}[D_{t+1}]$.
	For $1 \leq i < t+1$, let $D'_{t+1,i}$ denote the event that in time $t+1$, the coin flip decides to delegate, voter $i$ is chosen, and voter $i$ transitively delegates to voter $1$.
	Clearly, $D_{t+1}$ is the disjoint union of all $D'_{t+1,i}$.
	Therefore, $\mathbb{P}[D_{t+1}] = \sum_{1=1}^t \mathbb{P}[D'_{t+1,i}]$.
	Since the coin tosses in step $t+1$ are independent of the previous steps, $\mathbb{P}[D'_{t+1,i}] = d \cdot \frac{1}{t} \cdot \mathbb{P}[D_i]$.
	Putting the last steps together, have $\mathbb{P}[D_{t+1}] = \frac{d}{t} \cdot \sum_{i=1}^t \mathbb{P}[D_i] = \frac{d}{t} \cdot \mathbb{E}[S_t]$.
	In total, $\mathbb{E}[S_{t+1}] = \mathbb{E}[S_t] + \frac{d}{t} \cdot \mathbb{E}[S_t] = (1 + \frac{d}{t}) \cdot \mathbb{E}[S_t]$.
	
	Clearly, the recursion in \cref{eq:recur1,eq:recur2} must have a unique solution. We claim that it is
	\begin{equation}
	\label{eq:recursol}
	\mathbb{E}[S_t] = \frac{\Gamma(t + d)}{\Gamma(d + 1) \cdot \Gamma(t)},
	\end{equation}
	where the Gamma function is Legendre's extension of the factorial to real (and complex) numbers, defined by $\Gamma(z)=\int_{z=0}^{\infty} x^{z-1}e^{-x}\, dx$.
	Indeed, The equation satisfies \cref{eq:recur1}: $\frac{\Gamma(1 + d)}{\Gamma(d + 1) \cdot \Gamma(1)} = 1$. For \cref{eq:recur2}, we have
	\begin{align*}
	\left(1 + \frac{d}{t}\right) \cdot \frac{\Gamma(t + d)}{\Gamma(d + 1) \cdot \Gamma(t)} 
	= \frac{t + d}{t} \cdot \frac{\Gamma(t + d)}{\Gamma(d + 1) \cdot \Gamma(t)} 
	= \frac{\Gamma(t + 1 + d)}{\Gamma(d + 1) \cdot \Gamma(t + 1)}.
	\end{align*}
	
	Using this closed-form solution, we can bound $\mathbb{E}[S_t]$ as follows. 
	By Gautschi's inequality~\cite[Eq.~(7)]{gautschi}, we have
	\[ (t + 1)^{d-1} \leq \frac{\Gamma(t + d)}{\Gamma(t + 1)} \leq t^{d-1}. \]
	We multiply all sides by $t$ to obtain
	\[ t \cdot (t + 1)^{d-1} \leq \frac{\Gamma(t + d)}{\Gamma(t)} \leq t^{d}. \]
	Finally, we have
	\begin{equation}
	\label{eqn:upperlowerbounds}
	\frac{t \cdot (t + 1)^{d-1}}{\Gamma(d+1)} \leq \frac{\Gamma(t + d)}{\Gamma(d+1) \cdot \Gamma(t)} \leq \frac{t^{d}}{\Gamma(d+1)}.
	\end{equation}
	
	Next, we establish the tightness of the upper and lower bounds by showing that
	\begin{equation}
	\lim_{t \to \infty} \frac{\frac{t \cdot (t + 1)^{d - 1}}{\Gamma(d+1)}}{\frac{t^d}{\Gamma(d+1)}} = 1. \label{eqn:tight}
	\end{equation}
	Indeed, simplifying yields 
	\begin{align*}
	\lim_{t \to \infty} \frac{\frac{t \cdot (t + 1)^{d - 1}}{\Gamma(d+1)}}{\frac{t^d}{\Gamma(d+1)}}  = \lim_{t \to \infty} \frac{(t + 1)^{d - 1}}{t^{d-1}} 
	= \left( \lim_{t \to \infty} \frac{t+1}{t} \right)^{d-1} 
	= 1,
	\end{align*}
	as desired.
	
	Therefore, from \cref{eqn:upperlowerbounds} and \cref{eqn:tight}, we have shown that $\mathbb{E}[S_t]$ scales in $\Theta(t^d)$.
\end{proof}

Before proceeding to the upper bound and showing the separation, we would like to point out that --- with a minor change to our model --- these lower bounds also hold for $\gamma = 1$.
While the probability of attaching to a delegator $n$ remains proportional to $(1 + \mathit{indegree}(n))^\gamma$, the probability for voters $n$ would instead be proportional to $(2 + \mathit{indegree}(n))^\gamma$.\footnote{Clearly, our results for $\gamma = 0$ hold for both variants.}
If we represent voters with a self-loop edge, both terms just equal $\mathit{degree}(n)^\gamma$, which arguably makes this implementation of preferential attachment cleaner to analyze (e.g.,~\cite{bollobas}).
Thus, we can interpret preferential attachment for $\gamma = 1$ as uniformly picking an edge and then flipping a fair coin to decide whether to attach the edge's start or endpoint.
Since every node has exactly one outgoing edge, this is equivalent to uniformly choosing a node and then, with probability $\frac{1}{2}$, instead picking its successor.
This has the same effect on the distribution of weights as just uniformly choosing a node in uniform delegation, so \cref{thm:lbwhp,thm:lbexp} also hold for $\gamma = 1$ in our modified setting.
Real-world delegation networks, which we suspect to resemble the case of $\gamma = 1$, should therefore exhibit similar behavior.

\subsection{Upper Bound for Double Delegation ($k = 2$, $\gamma = 0$)}
\label{sec:upperchoice}

Analyzing cases with $k > 1$ is considerably more challenging. One obstacle is that we do not expect to be able to incorporate optimal resolution of potential delegations into our analysis, because the computational problem is hard even when $k=2$ (see Theorem~\ref{thm:nphard}). 
Therefore, we give a pessimistic estimate of optimal resolution via a greedy delegation mechanism, which we can reason about alongside the stochastic process.
Clearly, if this stochastic process can guarantee an upper bound on the maximum weight with high probability, this bound must also hold if delegations are optimally resolved to minimize maximum weight.

In more detail, whenever a new delegator is inserted into the graph, the greedy mechanism immediately selects one of the delegation options.
As a result, at any point during the construction of the graph, the algorithm can measure the weight of the voters.
Suppose that a new delegator suggests two delegation options, to agents $a$ and $b$.
By following already resolved delegations, the mechanism obtains voters $a^*$ and $b^*$ such that $a$ transitively delegates to $a^*$ and $b$ to $b^*$.
The greedy mechanism then chooses the delegation whose voter currently has lower weight, resolving ties arbitrarily.

This situation is reminiscent of a phenomenon known as the ``power of choice.''
In its most isolated form, it has been studied in the \emph{balls-and-bins} model, for example by Azar et al.~\cite{azar}.
In this model, $n$ balls are to be placed in $n$ bins.
In the classical setting, each ball is sequentially placed into a bin chosen uniformly at random.
With high probability, the fullest bin will contain $\Theta( \log n/\log \log n )$ balls at the end of the process.
In the choice setting, two bins are independently and uniformly selected for every ball, and the ball is placed into the emptier one.
Surprisingly, this leads to an exponential improvement, where the fullest bin will contain at most $\Theta \left( \log \log n \right)$ balls with high probability.

We show that, at least for $\gamma = 0$ in our setting, this effect outweighs the ``rich-get-richer'' dynamic described earlier:

\begin{theorem}
	\label{thm:thm1}
In the preferential delegation model with $k=2$, $\gamma=0$, and $d\in (0,1)$, the maximum weight of any voter at time $t$ is $\log_2 \log t + \Theta(1)$ with high probability. 
\end{theorem}
Due to space constraints, we defer the proof to \cref{sec:upperchoiceproofs}.
In our proof we build on work by Malyshkin and Paquette~\cite{malyshkin}, who study the maximum \emph{degree} in a graph generated by preferential attachment with the power of choice. In addition, we incorporate ideas by Haslegrave and Jordan~\cite{haslegrave}.\footnote{More precisely, for the definition of the sequence $(\alpha_k)_k$ as well as in \cref{lem:alphatozero,lem:alphatozerofast}.}
 
\section{Simulations}
\label{sec:simulations}

In this section, we present our simulation results, which support the two main messages of this paper: that allowing multiple delegation options significantly reduces the maximum weight, and that it is computationally feasible to resolve delegations in a way that is close to optimal.

Our simulations were performed on a MacBook Pro (2017) on MacOS~10.12.6 with a 3.1~GHz Intel Core i5 and 16~GB of RAM.
All running times were measured with at most one process per processor core.
Our simulation software is written in Python~3.6 using Gurobi~8.0.1 to solve MILPs.
All of our simulation code is open-source and available at \href{https://github.com/pgoelz/fluid}{\small\texttt{https://github.com/pgoelz/fluid}}.

\subsection{Multiple vs.\ Single Delegations}
\label{subsec:multsingle}
For the special case of $\gamma = 0$, we have established a doubly exponential, asymptotic separation between single delegation ($k = 1$) and two delegation options per delegator ($k = 2$).
While the strength of the separation suggests that some of this improvement will carry over to the real world, we still have to examine via simulation whether improvements are visible for realistic numbers of agents and other values of $\gamma$.

To this end, we empirically evaluate two different mechanisms for resolving delegations. First, we optimally resolve delegations by solving the MILP for confluent flow in \cref{app:milp} with the Gurobi optimizer. Our second mechanism is the greedy ``power of choice'' algorithm used in the theoretical analysis and introduced in \cref{sec:upperchoice}. 

\begin{figure}
	\centering
	\begin{subfigure}[h]{.49\textwidth}
		\includegraphics[width=\textwidth]{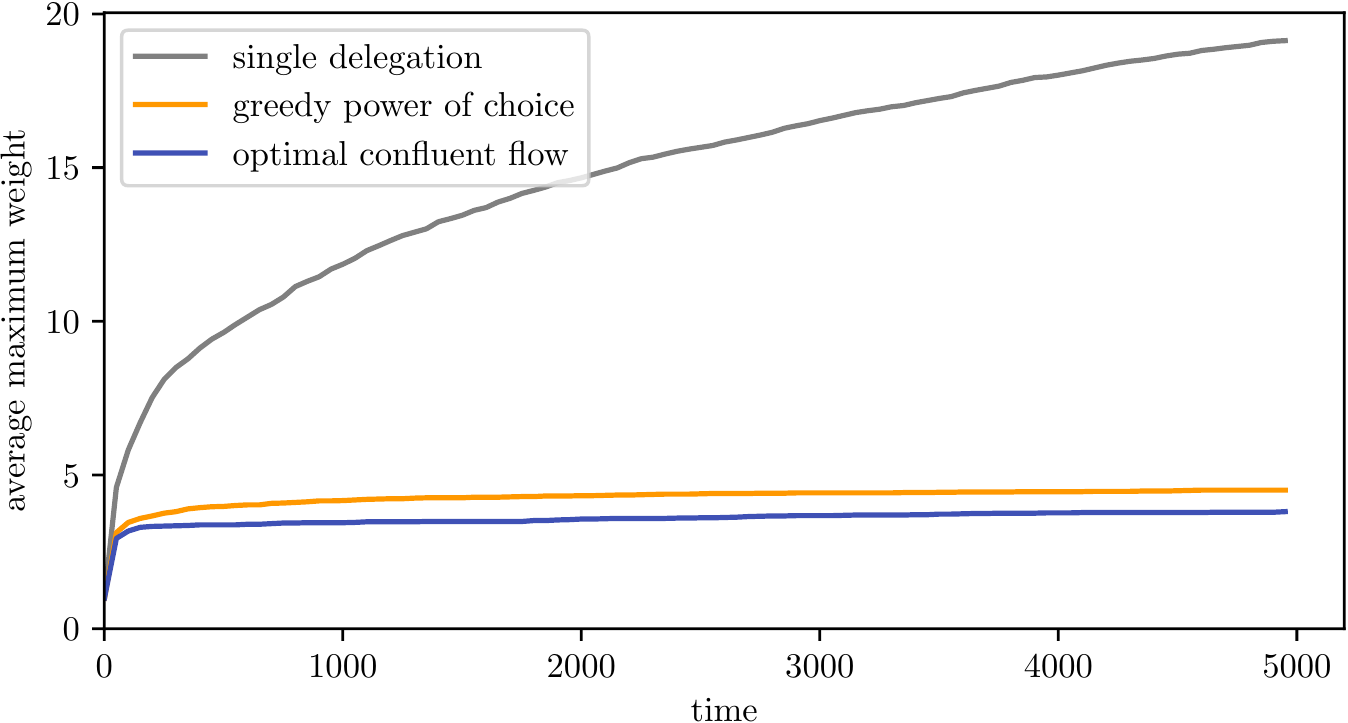}
		\caption{$\gamma = 0$, $d = 0.25$}
	\end{subfigure}
	\begin{subfigure}[h]{.49\textwidth}
		\includegraphics[width=\textwidth]{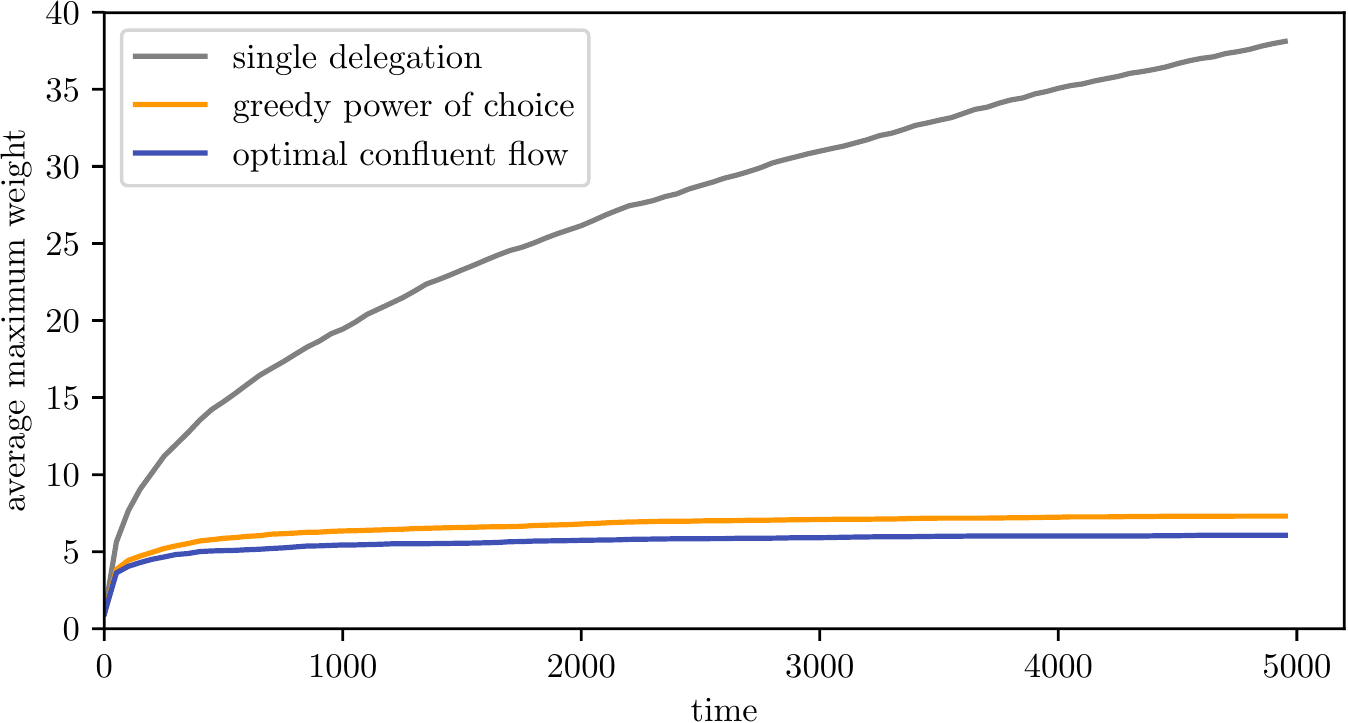}
		\caption{$\gamma = 1$, $d = 0.25$}
	\end{subfigure} \\
	\bigskip $ $\\ \medskip
	\begin{subfigure}[h]{.49\textwidth}
		\includegraphics[width=\textwidth]{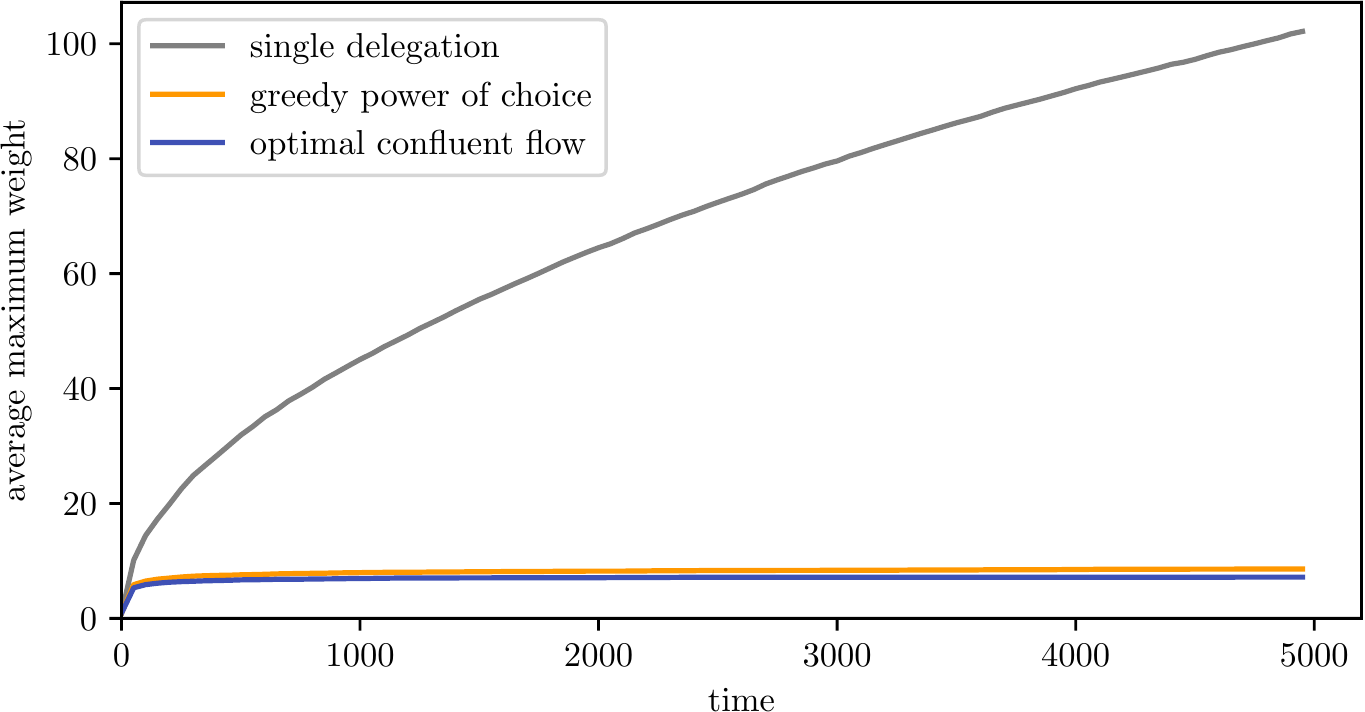}
		\caption{$\gamma = 0$, $d = 0.5$}
	\end{subfigure}
	\begin{subfigure}[h]{.49\textwidth}
		\includegraphics[width=\textwidth]{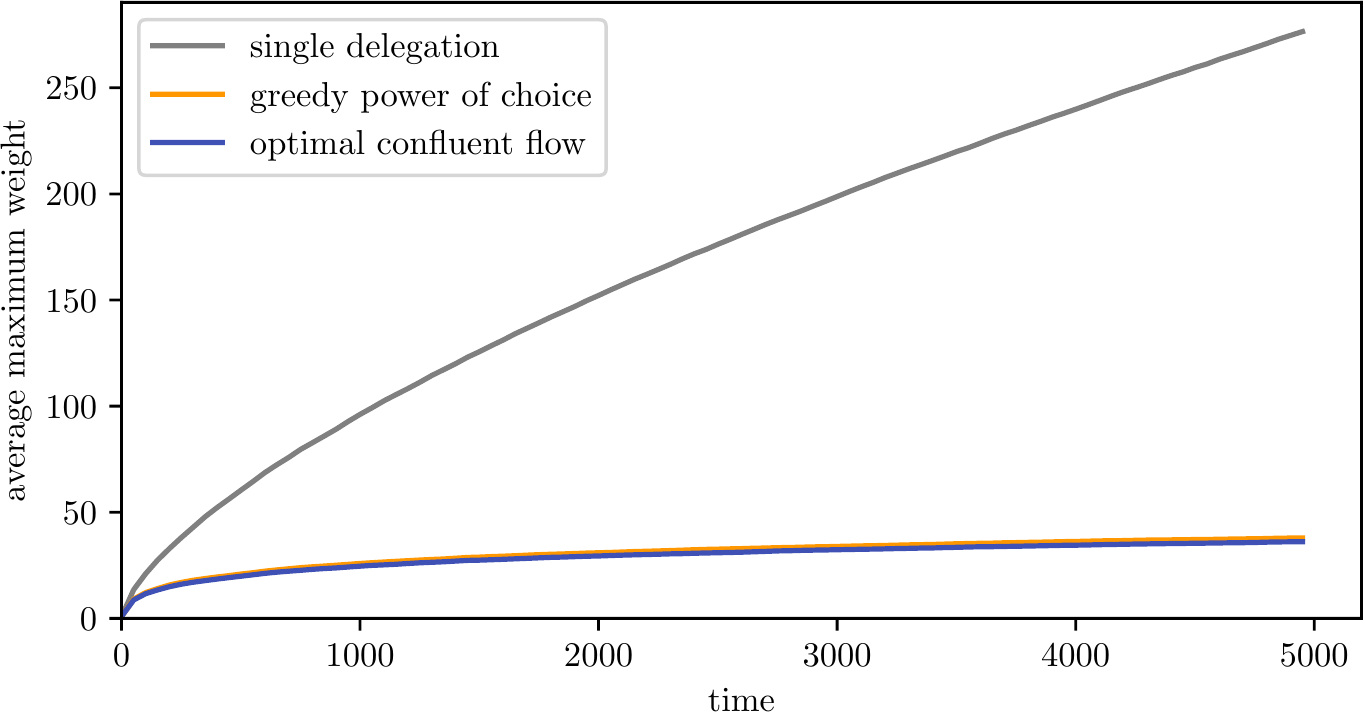}
		\caption{$\gamma = 1$, $d = 0.5$}
	\end{subfigure} \\
	\bigskip $ $\\ \medskip
	\begin{subfigure}[h]{.49\textwidth}
		\includegraphics[width=\textwidth]{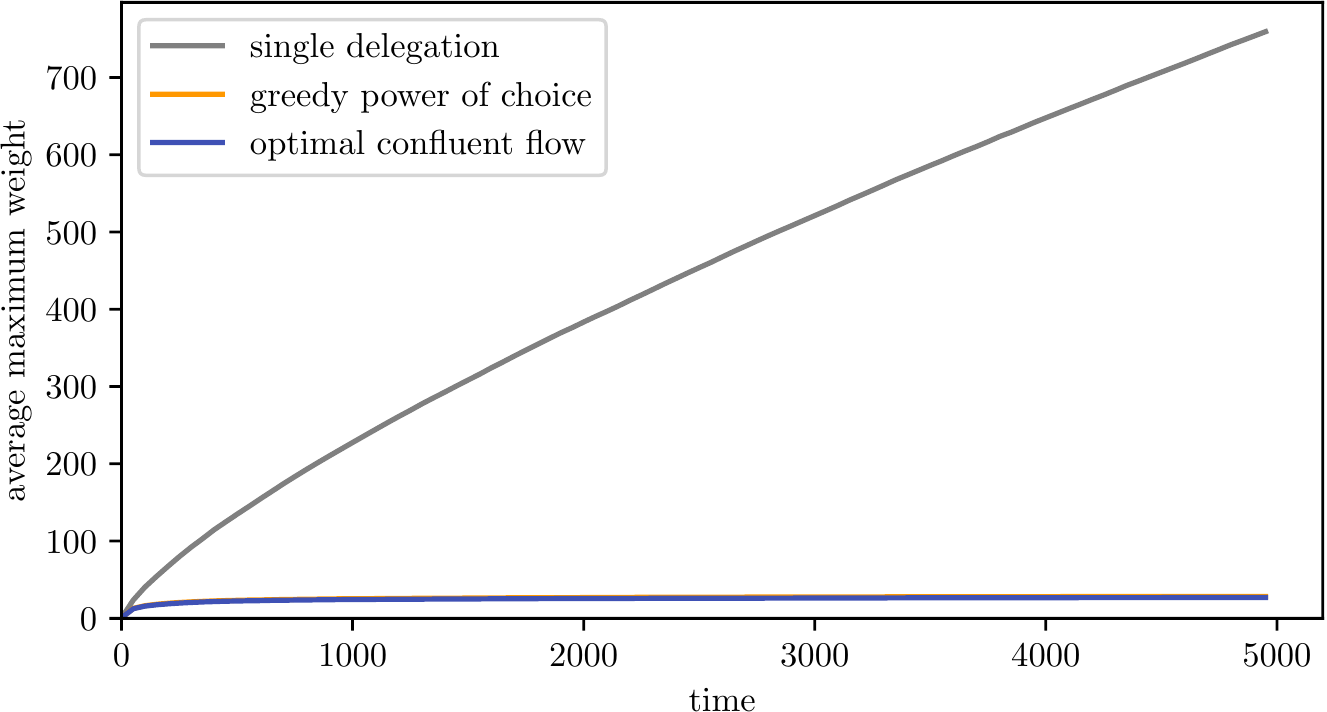}
		\caption{$\gamma = 0$, $d = 0.75$}
		\label{fig:g0d75}
	\end{subfigure}
	\begin{subfigure}[h]{.49\textwidth}
		\includegraphics[width=\textwidth]{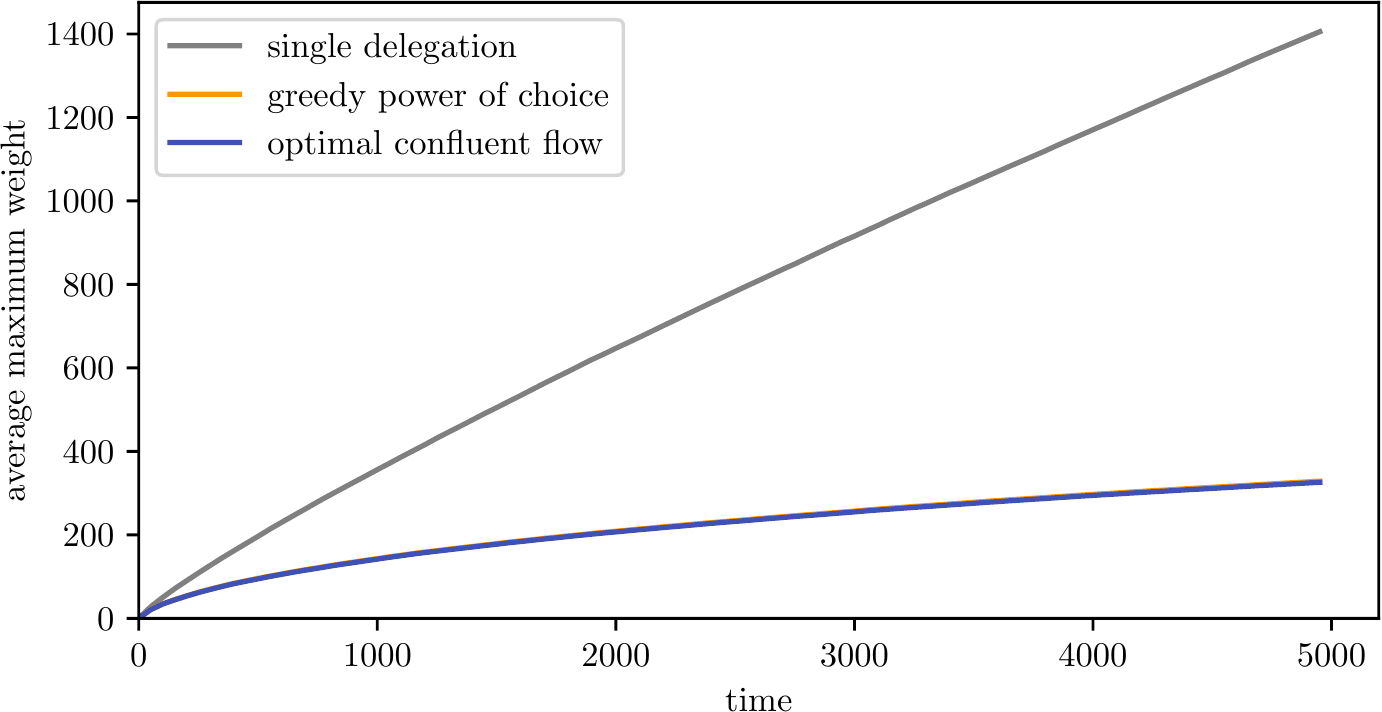}
		\caption{$\gamma = 1$, $d = 0.75$}
	\end{subfigure}
	
    \caption{Maximum weight averaged over 100 simulations of length 5\,000 time steps each. Maximum weight has been computed every 50 time steps.}
	\label{fig:singlevstwo}
\end{figure}
In \cref{fig:singlevstwo}, we compare the maximum weight produced by a single-delegation process to the optimal maximum weight in a double-delegation process, for different values of $\gamma$ and $d$.
Since our theoretical analysis used a greedy over-approximation of the optimum, we also run the greedy mechanism on the double-delegation process.
Corresponding figures for $\gamma = 0.5$ can be found in \cref{fig:singlevstwogammapointfive} in \cref{sec:appsinglevsdoublefigures}.

These simulations show that our asymptotic findings translate into considerable differences even for small numbers of agents, across different values of $d$.
Moreover, these differences remain nearly as pronounced for values of $\gamma$ up to $1$, which corresponds to classical preferential attachment.
This suggests that our mechanism can outweigh the social tendency towards concentration of votes; however, evidence from real-world elections is needed to settle this question.
Lastly, we would like to point out the similarity between the graphs for the optimal maximum weight and the result of the greedy algorithm, which indicates that a large part of the separation can be attributed to the power of choice.

If we increase $\gamma$ to large values, the separation between single and double delegation disappears.
In \cref{fig:singlevstwogammatwo} in \cref{sec:appsinglevsdoublefigures}, for $\gamma = 2$, all three curves are hardly distinguishable from the linear function $d \cdot \mathit{time}$, meaning that one voter receives nearly all the weight.
The reason is simple: In the simulations used for that figure, 99\,\% %
of all delegators give two identical delegation options, and 99.8\,\% %
of these delegators (98.8\,\% of all delegators) %
give both potential delegations to the heaviest voter in the graph. There are even values of $\gamma>1$ and $d$ such that the curve for single delegation falls below the ones for double delegation (as can be seen in \cref{fig:singlevstwogammagtone} in \cref{sec:appsinglevsdoublefigures}). 
Since adding two delegation options per step makes the indegrees grow faster, the delegations concentrate toward a single voter more quickly, and again lead to a wildly unrealistic concentration of weight. Thus, it seems that large values of $\gamma$ do not actually describe our scenario of multiple delegations.

\begin{figure}
	\centering
	\captionsetup{width=0.45\textwidth}
	\begin{minipage}{.5\textwidth}
		\centering
		\includegraphics[width=\textwidth]{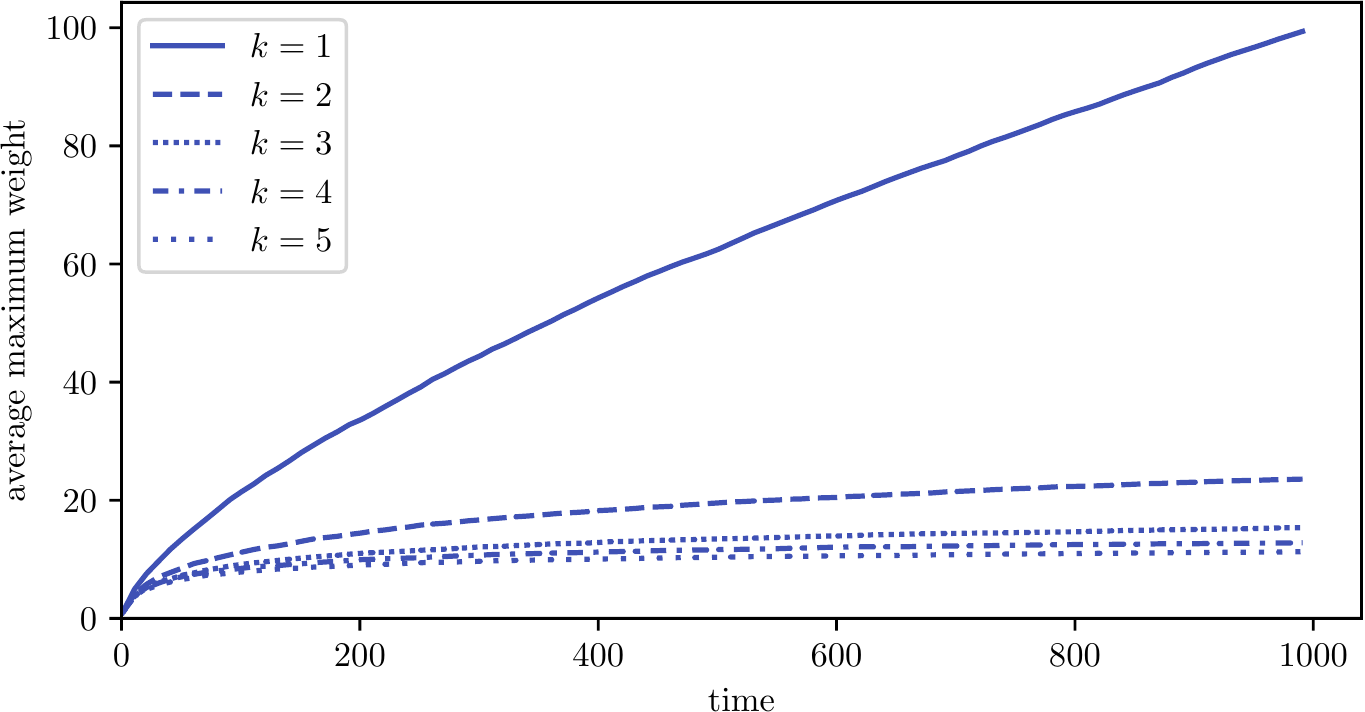}
		\captionof{figure}{Optimal maximum weight for different $k$ averaged over 100 simulations, computed every 10 steps. $\gamma = 1$, $d = 0.5$.}
		\label{fig:differentk}
	\end{minipage}%
	\begin{minipage}{.5\textwidth}
	\centering
	\includegraphics[width=\textwidth]{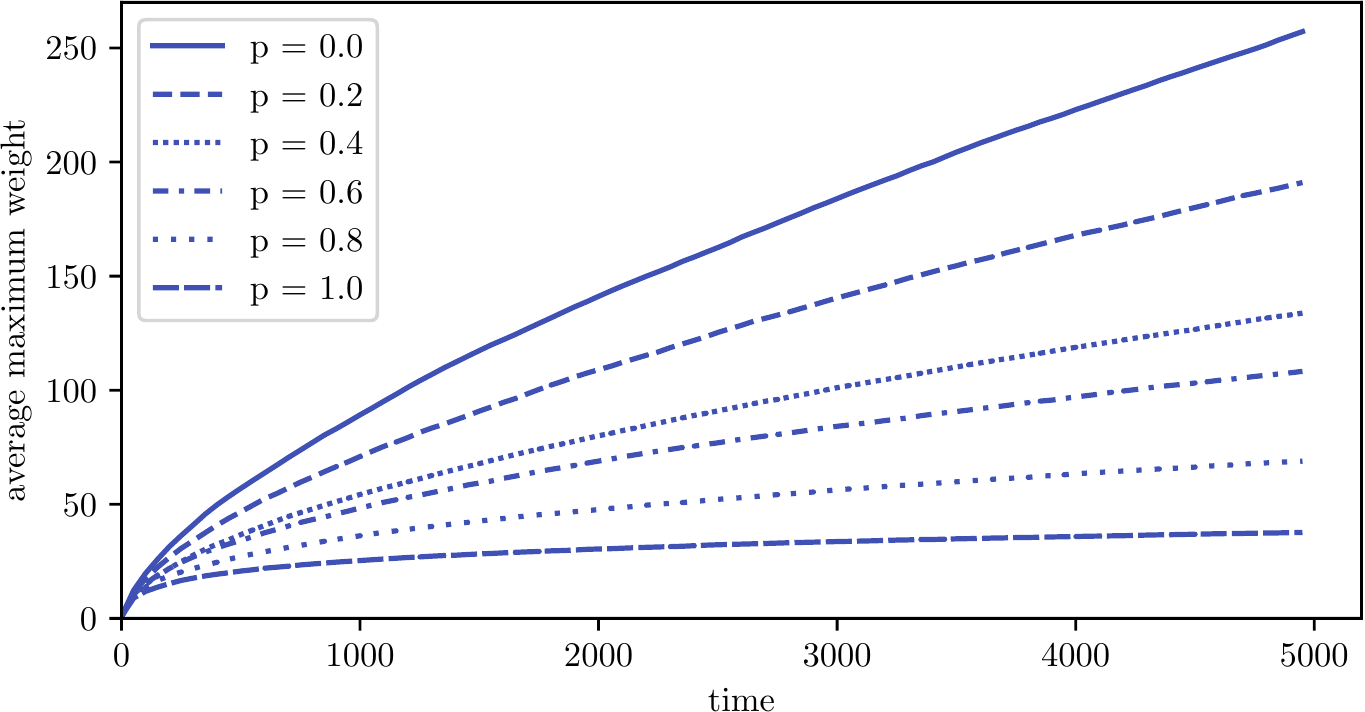}
	\captionof{figure}{Optimal maximum weight averaged over 100 simulations. Voters give two delegations with probability $p$; else one. $\gamma = 1$, $d = 0.5$.}
	\label{fig:differentp}
	\end{minipage}
\end{figure}

As we have seen, switching from single delegation to double delegation greatly improves the maximum weight in plausible scenarios.
It is natural to wonder whether increasing $k$ beyond $2$ will yield similar improvements.
As \cref{fig:differentk} shows, however, the returns of increasing $k$ quickly diminish, which is common to many incarnations of the power of choice~\cite{azar}.

\subsection{Evaluating Mechanisms}
Already the case of $k = 2$ appears to have great potential; but how easily can we tap it?

We have observed that, on average, the greedy ``power of choice'' mechanism comes surprisingly close to the optimal solution.
However, this greedy mechanism depends on seeing the order in which our random process inserts agents and on the fact that all generated graphs are acyclic, which need not be true in practice.
If the graphs were acyclic, we could simply first sort the agents topologically and then present the agents to the greedy mechanism in reverse order.
On arbitrary active graphs, we instead proceed through the strongly connected components in reversed topological order, breaking cycles and performing the greedy step over the agents in the component.
To avoid giving the greedy algorithm an unfair advantage, we use this generalized greedy mechanism throughout this section.
Thus, we compare the generalized greedy mechanism, the optimal solution, the $(1 + \log |V|)$-approximation algorithm\footnote{For one of their subprocedures, instead of directly optimizing a convex program, Chen et al.~\cite{chen} reduce this problem to finding a lexicographically optimal maximum flow in $\mathcal{O}(n^5)$. We choose to directly optimize the convex problem in Gurobi, hoping that this will increase efficiency in practice.} and a random mechanism that materializes a uniformly chosen option per delegator.

\begin{figure}
	\centering
	\begin{subfigure}[h]{0.49\textwidth}
		\includegraphics[width=\linewidth]{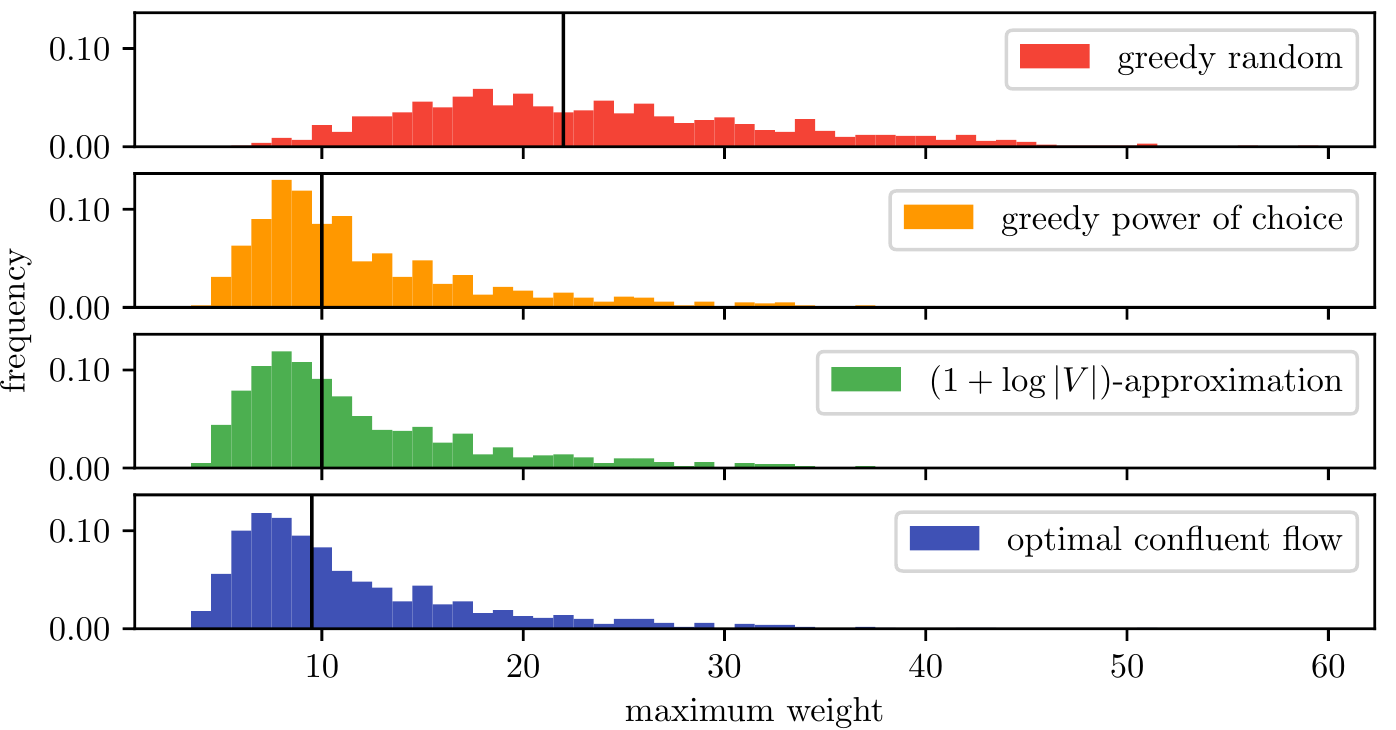}
		\caption{$t = 100$}
		\label{subfig:histg100d50t100}
	\end{subfigure}
	\begin{subfigure}[h]{0.49\textwidth}
		\includegraphics[width=\linewidth]{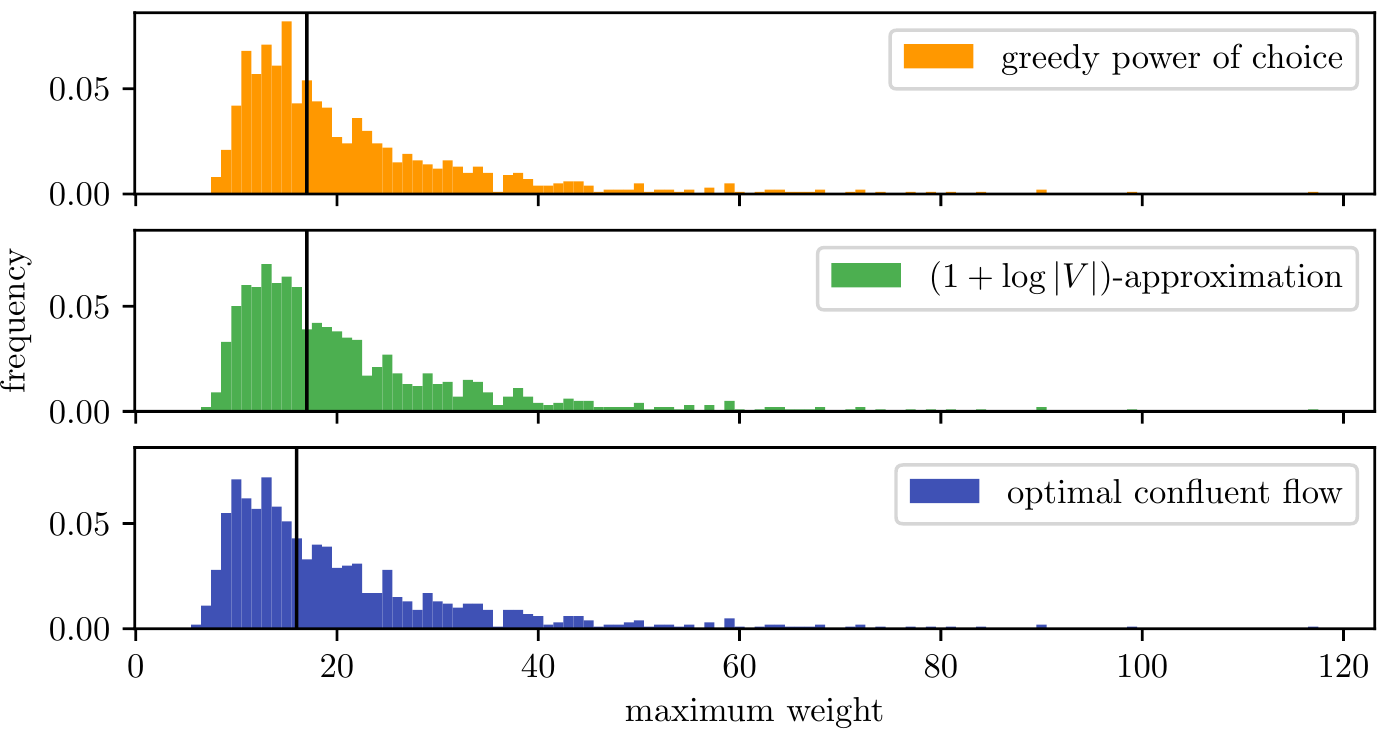}
		\caption{$t = 500$}
		\label{subfig:histg100d50}
	\end{subfigure}
	\caption{Frequency of maximum weights at time $t$ over $1\,000$ runs. $\gamma = 1$, $d = 0.5$, $k = 2$. The black lines mark the medians.}
	\label{fig:histg100d50}
\end{figure}

On a high level, we find that both the generalized greedy algorithm and the approximation algorithm perform comparably to the optimal confluent flow solution, as shown in \cref{fig:histg100d50} for $d = 0.5$ and $\gamma = 1$.
As \cref{fig:singletraceg100d50t2000} suggests, all three mechanisms seem to exploit the advantages of double delegation, at least on our synthetic benchmarks.
These trends persist for other values of $d$ and $\gamma$, as presented in \cref{app:histograms}.

The similar success of these three mechanisms might indicate that our probabilistic model for $k = 2$ generates delegation networks that have low maximum weights for arbitrary resolutions.
However, this is not the case:
The random mechanism does quite poorly on instances with as few as $t = 100$ agents, as shown in \cref{subfig:histg100d50t100}.
With increasing $t$, the gap between random and the other mechanisms only grows further, as indicated by \cref{fig:singletraceg100d50t2000}.
In general, the graph for random delegations looks more similar to single delegation than to the other mechanisms on double delegation.
Indeed, for $\gamma = 0$, random delegation is equivalent to the process with $k = 1$, and, for higher values of $\gamma$, it performs even slightly worse since the unused delegation options make the graph more centralized (see \cref{fig:randomvssingle} in \cref{app:randomvssingle}).
Because of the poor performance of random delegation, if simplicity is a primary desideratum, we recommend using the generalized greedy algorithm instead.

As \cref{fig:runningtime} and the graphs in \cref{sec:runtimegraphs} demonstrate, all three other mechanisms, including the optimal solution, easily scale to input sizes as large as the largest implementations of liquid democracy to date. Whereas the three mechanisms were close with respect to maximum weight, our implementation of the approximation algorithm is typically slower than the optimal solution (which requires a single call to Gurobi), and the generalized greedy algorithm is blazing fast. These results suggest that it would be possible to resolve delegations almost optimally even at a national scale.

\begin{figure}
	\centering
	\captionsetup{width=0.45\textwidth}
	\begin{minipage}{.5\textwidth}
		\centering
		\includegraphics[width=\linewidth]{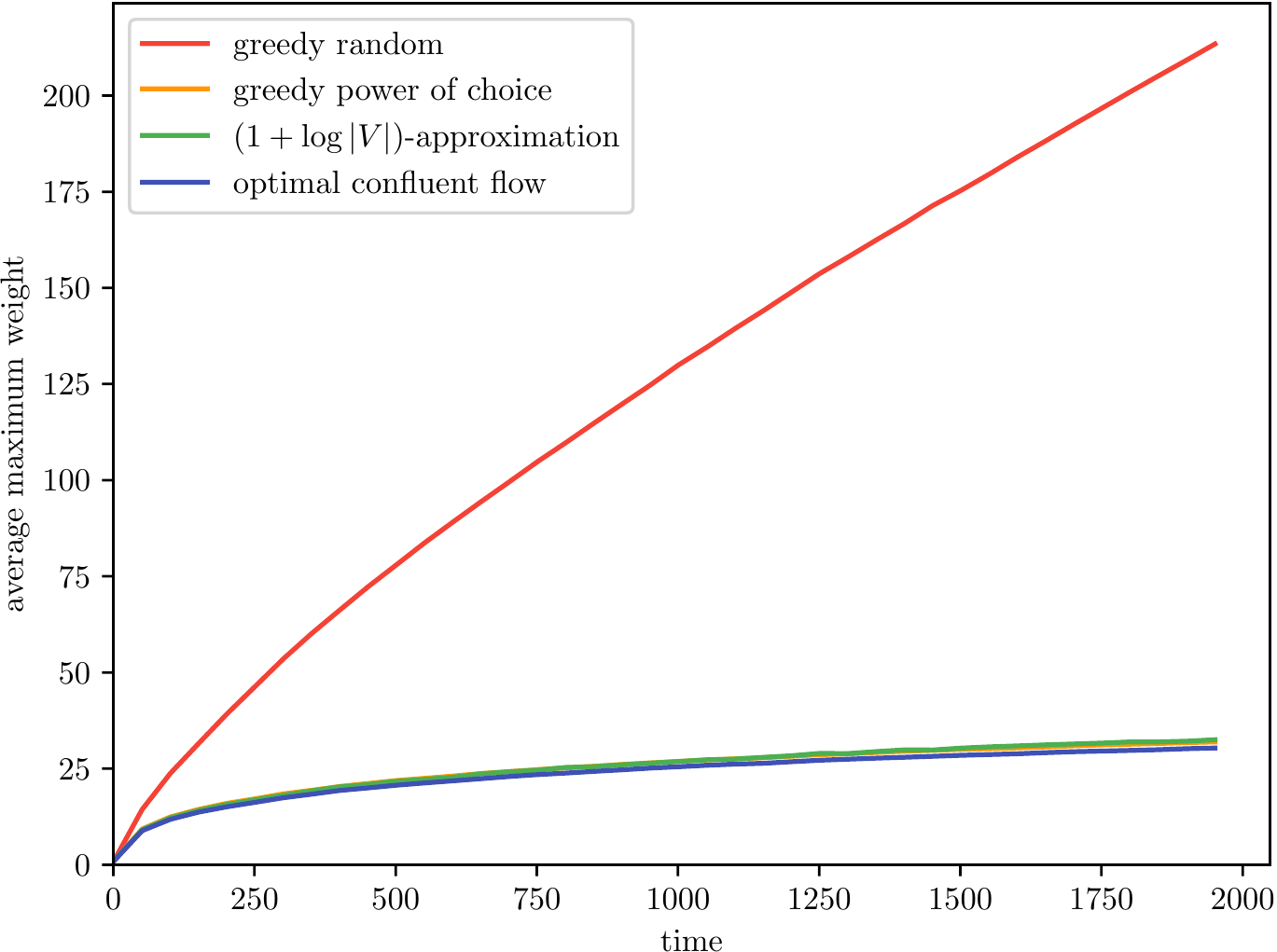}
		\captionof{figure}{Maximum weight per algorithm for $d = 0.5$, $\gamma = 1$, $k = 2$, averaged over 100 simulations.}
		\label{fig:singletraceg100d50t2000}
	\end{minipage}%
	\begin{minipage}{.5\textwidth}
		\centering
		\includegraphics[height=.748\linewidth]{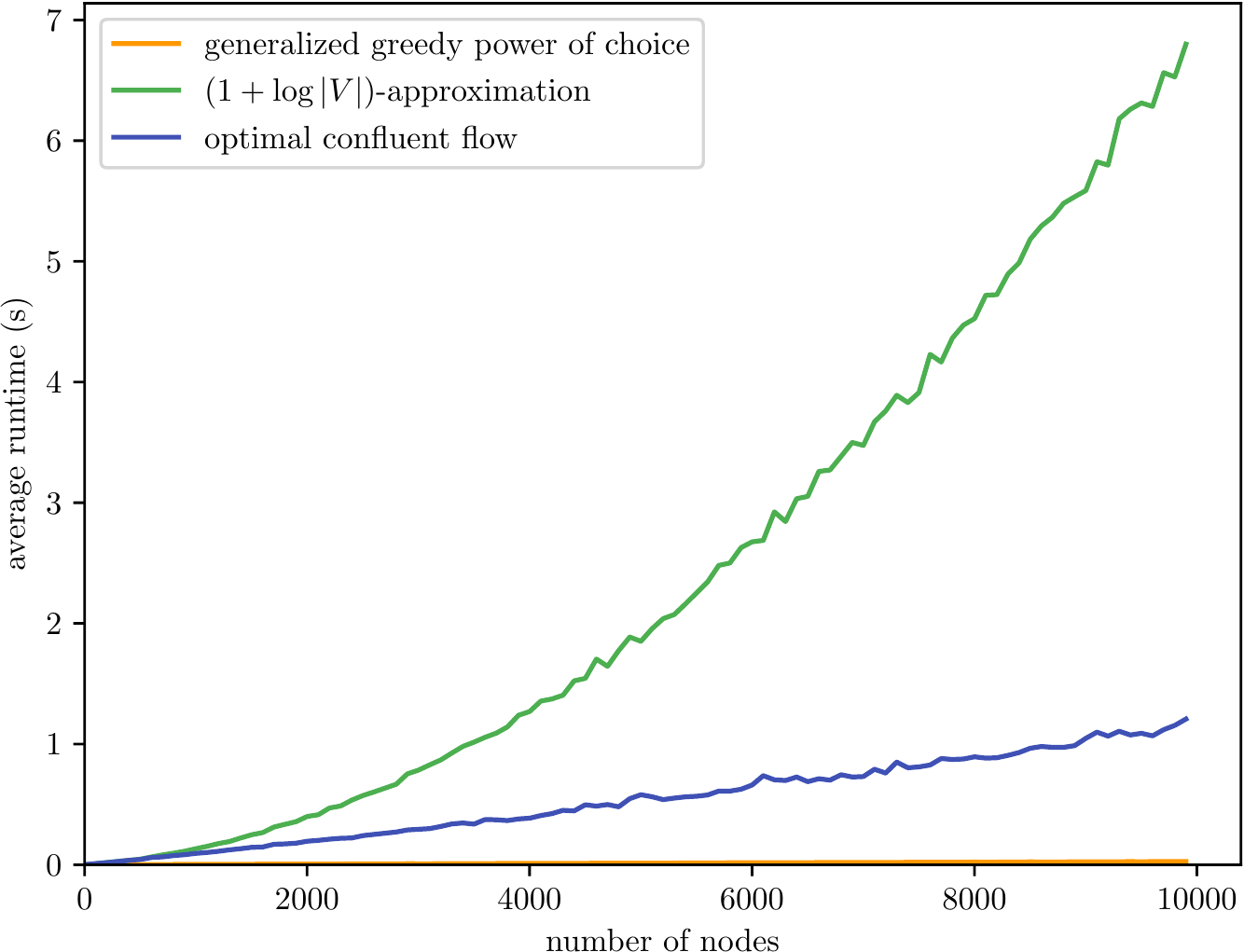}
		\captionof{figure}{Running time of mechanisms on graphs for $d = 0.5$, $\gamma = 1$, averaged over 20 simulations.}
		\label{fig:runningtime}
	\end{minipage}
\end{figure}

\section{Discussion}\label{sec:disc}
The approach we have presented and analyzed revolves around the idea of allowing agents to specify multiple delegation options, and selecting one such option per delegator. As mentioned in Section~\ref{sec:alg}, a natural variant of this approach corresponds to splittable --- instead of confluent --- flow. In this variant, the mechanism would not have to commit to a single outgoing edge per delegator.
Instead, a delegator's weight could be split into arbitrary fractions between her potential delegates.
Indeed, such a variant would be computationally less expensive, and the maximum voting weight can be no higher than in our setting.
However, we view our concept of delegation as more intuitive and transparent:
Whereas, in the splittable setting, a delegator's vote can disperse among a large number of agents, our mechanism assigns just one representative to each delegator.
As hinted at in the introduction, this is needed to preserve the high level of accountability guaranteed by classical liquid democracy.
We find that this fundamental shortcoming of splittable delegations is not counterbalanced by a marked decrease in maximum weight. Indeed, representative empirical results given in \cref{app:confsplit} show that the maximum weight trace is almost identical under splittable and confluent delegations. This conclusion is supported by additional results in \cref{app:histograms}.
Furthermore, note that in the preferential delegation model with $k=1$, splittable delegations do not make a difference, so the lower bounds given in \cref{thm:lbwhp,thm:lbexp} go through. And, when $k\geq 2$, the upper bound of Theorem~\ref{thm:thm1} directly applies to the splittable setting. Therefore, our main technical results in Section~\ref{sec:probmodel} are just as relevant to splittable delegations. 

To demonstrate the benefits of multiple delegations as clearly as possible, we assumed that every agent provides two possible delegations.
In practice, of course, we expect to see agents who want to delegate but only trust a single person to a sufficient degree.
This does not mean that delegators should be required to specify multiple delegations.
For instance, if this was the case, delegators might be incentivized to pad their delegations with very popular agents who are unlikely to receive their votes. 
Instead, we encourage voters to specify multiple delegations on a voluntary basis, and we hope that enough voters participate to make a significant impact.
Fortunately, as demonstrated in \cref{fig:differentp}, much of the benefits of multiple delegation options persist even if only a fraction of delegators specify two delegations.

Without doubt, a centralized mechanism for resolving delegations wields considerable power. Even though we only use this power for our specific goal of minimizing the maximum weight, agents unfamiliar with the employed algorithm might suspect it of favoring specific outcomes. To mitigate these concerns, we propose to divide the voting process into two stages. In the first, agents either specify their delegation options or register their intent to vote. Since the votes themselves have not yet been collected, the algorithm can resolve delegations without seeming partial. In the second stage, voters vote using the generated delegation graph, just as in classic liquid democracy, which allows for transparent decisions on an arbitrary number of issues. Additionally, we also allow delegators to change their mind and vote themselves if they are dissatisfied with how delegations were resolved. This gives each agent the final say on their share of votes, and can only further reduce the maximum weight achieved by our mechanism. We believe that this process, along with education about the mechanism's goals and design, can win enough trust for real-world deployment.

Beyond our specific extension, one can consider a variety of different approaches that push the current boundaries of liquid democracy. For example, in a recent position paper, Brill~\cite{Brill18} raises the idea of allowing delegators to specify a ranked list of potential representatives. His proposal is made in the context of alleviating delegation cycles, whereas our focus is on avoiding excessive concentration of weight. But, on a high level, both proposals envision centralized mechanisms that have access to richer inputs from agents. Making and evaluating such proposals now is important, because, at this early stage in the evolution of liquid democracy, scientists can still play a key role in shaping this exciting paradigm.

\section*{Acknowledgments}

We are grateful to Miklos Z.\ Racz for very helpful pointers to analyses of preferential attachment models.

\bibliographystyle{splncs03}
\bibliography{abb,ultimate,others}

\newpage
\appendix
\section{Proof of Lemma~\lowercase{\ref{lem:nphard}} -- Hardness of \textsc{MinMaxCongestion}}
\label{app:nphard}

We first require the following lemma. 

\begin{lem}
    \label{lem:disjoutdegree}
    Let $G$ be a directed graph in which all vertices have an outdegree of at most 2. Given vertices $s_1, s_2, t_1, t_2$, 
    it is NP-hard to decide whether there exist vertex-disjoint paths from $s_1$ to $t_1$ and from $s_2$ to $t_2$.
\end{lem}

\begin{proof}
    Without the restriction on the outdegree, the problem is NP-hard.\footnote{Chen et al.~\cite{chen} cite a previously established result for this~\cite{fortune}.}
    We reduce the general case to our special case.

    Let $G'$ be an arbitrary directed graph; let $s_1', s_2', t_1', t_2'$ be distinguished vertices.
    To restrict the outdegree, replace each node $n$ with outdegree $d$ by a binary arborescence (directed binary tree with edges facing away from the root) with $d$ sinks.
    All incoming edges into $n$ are redirected towards the root of the arborescence; outgoing edges from $n$ instead start from the different leaves of the arborescence.
    Call the new graph $G$, and let $s_1, s_2, t_1, t_2$ refer to the roots of the arborescences replacing $s_1', s_2', t_1', t_2'$, respectively.

    Clearly, our modifications to $G'$ can be carried out in polynomial time.
    It remains to show that there are vertex-disjoint paths from $s_1$ to $t_1$ and from $s_2$ to $t_2$ in $G$ iff there are vertex-disjoint paths from $s_1'$ to $t_1'$ and from $s_2'$ to $t_2'$ in $G'$.

    If there are disjoint paths in $G'$, we can translate these paths into $G$ by visiting the arborescences corresponding to the nodes on the original path one after another.
    Since both paths visit disjoint arborescences, the new paths must be disjoint.

    Suppose now that there are disjoint paths in $G$.
    Translate the paths into $G'$ by visiting the nodes corresponding to the sequence of visited arborescences.
    Since each arborescence can only be entered via its root, disjointness of the paths in $G$ implies disjointness of the translated paths in $G'$.
\end{proof}

We are now ready to prove Lemma~\ref{lem:nphard}.

\restnphard*

\begin{proof}
    We adapt the proof of Theorem~1 of Chen et al.~\cite{chen}.

    Let $G = (V, E), s_1, s_2, t_1, t_2$ be given as in Lemma~\ref{lem:disjoutdegree}.
    Without loss of generality, $G$ only contains nodes from which $t_1$ or $t_2$ is reachable, $t_1$ and $t_2$ are sinks and all four vertices are distinct.
    Let $\ell = \lceil \log_2 |V| \rceil$ and $k = 2^\ell$.
    Build the same auxiliary network as that built by Chen et al.~\cite{chen}, which consists of a binary arborescence whose $k - 1$ nodes are copies of $G$.
    The construction is illustrated in \cref{fig:approxreduction}. For more details, refer to \cite{chen}.
    \begin{figure}
        \centering
          \resizebox{\textwidth}{!}{\begin{tikzpicture}[on grid,>=stealth',node distance=.8cm]
\tikzstyle{do}=[circle,thick,draw=black,minimum size=1.8mm,inner sep=0]
\tikzstyle{s1}=[do,label={[label distance=4]left:$s_1$}]
\tikzstyle{s2}=[do,fill=black,label={[label distance=4]right:$s_2$}]
\tikzstyle{t1}=[do,label={[label distance=4]right:$t_1$}]
\tikzstyle{t2}=[do,fill=black,label={[label distance=4]left:$t_2$}]
\tikzstyle{sink}=[rectangle,thick,draw=black,minimum size=2.8mm,inner sep=0]

\node (hph) at (2.7, 0) {};
\node (vph) at (0,3) {};
\node (vsph) at (0,-1.5) {};

\node (phantomA1) [circle,minimum size=2cm,draw=black,dashed] at ($3*(vph)$) {};

\node (phantomA2) [circle,minimum size=2cm,draw=black,dashed] at ($-2*(hph)+2*(vph)$) {};
\node (phantomA3) [circle,minimum size=2cm,draw=black,dashed] at ($2*(hph)+2*(vph)$) {};

\node (phantomA4) [circle,minimum size=2cm,draw=black,dashed] at ($-3*(hph)+(vph)$){};
\node (phantomA5) [circle,minimum size=2cm,draw=black,dashed] at ($-1*(hph)+(vph)$) {};
\node (phantomA6) [circle,minimum size=2cm,draw=black,dashed] at ($1*(hph)+(vph)$) {};
\node (phantomA7) [circle,minimum size=2cm,draw=black,dashed] at ($3*(hph)+(vph)$) {};

\node (phantomA8) [circle,minimum size=2cm,draw=black,dashed] at ($-3.5*(hph)$) {};
\node (phantomA9) [circle,minimum size=2cm,draw=black,dashed] at ($-2.5*(hph)$) {};
\node (phantomA10) [circle,minimum size=2cm,draw=black,dashed] at ($-1.5*(hph)$) {};
\node (phantomA11) [circle,minimum size=2cm,draw=black,dashed] at ($-0.5*(hph)$) {};
\node (phantomA12) [circle,minimum size=2cm,draw=black,dashed] at ($0.5*(hph)$) {};
\node (phantomA13) [circle,minimum size=2cm,draw=black,dashed] at ($1.5*(hph)$) {};
\node (phantomA14) [circle,minimum size=2cm,draw=black,dashed] at ($2.5*(hph)$) {};
\node (phantomA15) [circle,minimum size=2cm,draw=black,dashed] at ($3.5*(hph)$) {};

\node (s1A1) [s1,above left=of phantomA1,fill=black,minimum size=2.5mm] {};
\node (s2A1) [s2,above right=of phantomA1] {};
\node (t2A1) [t2,below left=of phantomA1] {};
\node (t1A1) [t1,below right=of phantomA1] {};
\draw [very thick,densely dotted] (s1A1) -- ($(phantomA1) + (-.3cm,.05cm)$) -- ($(phantomA1) + (.1cm,.2cm)$) -- ($(phantomA1) + (.2cm,-.5cm)$) -- (t1A1);
\draw [very thick] (s2A1) -- ($(phantomA1) + (.2cm,.45cm)$) -- ($(phantomA1) + (.4cm,.0cm)$) -- ($(phantomA1) + (-.1cm,-.05cm)$) -- ($(phantomA1) + (-.2cm,-.6cm)$) -- (t2A1);
\node (s1A2) [s1,above left=of phantomA2] {};
\node (s2A2) [s2,above right=of phantomA2] {};
\node (t2A2) [t2,below left=of phantomA2] {};
\node (t1A2) [t1,below right=of phantomA2] {};
\draw [very thick,densely dotted] (s1A2) -- ($(phantomA2) + (-.3cm,.05cm)$) -- ($(phantomA2) + (.1cm,.2cm)$) -- ($(phantomA2) + (.2cm,-.5cm)$) -- (t1A2);
\draw [very thick] (s2A2) -- ($(phantomA2) + (.2cm,.45cm)$) -- ($(phantomA2) + (.4cm,.0cm)$) -- ($(phantomA2) + (-.1cm,-.05cm)$) -- ($(phantomA2) + (-.2cm,-.6cm)$) -- (t2A2);
\node (s1A3) [s1,above left=of phantomA3] {};
\node (s2A3) [s2,above right=of phantomA3] {};
\node (t2A3) [t2,below left=of phantomA3] {};
\node (t1A3) [t1,below right=of phantomA3] {};
\draw [very thick,densely dotted] (s1A3) -- ($(phantomA3) + (-.3cm,.05cm)$) -- ($(phantomA3) + (.1cm,.2cm)$) -- ($(phantomA3) + (.2cm,-.5cm)$) -- (t1A3);
\draw [very thick] (s2A3) -- ($(phantomA3) + (.2cm,.45cm)$) -- ($(phantomA3) + (.4cm,.0cm)$) -- ($(phantomA3) + (-.1cm,-.05cm)$) -- ($(phantomA3) + (-.2cm,-.6cm)$) -- (t2A3);
\node (s1A4) [s1,above left=of phantomA4] {};
\node (s2A4) [s2,above right=of phantomA4] {};
\node (t2A4) [t2,below left=of phantomA4] {};
\node (t1A4) [t1,below right=of phantomA4] {};
\draw [very thick,densely dotted] (s1A4) -- ($(phantomA4) + (-.3cm,.05cm)$) -- ($(phantomA4) + (.1cm,.2cm)$) -- ($(phantomA4) + (.2cm,-.5cm)$) -- (t1A4);
\draw [very thick] (s2A4) -- ($(phantomA4) + (.2cm,.45cm)$) -- ($(phantomA4) + (.4cm,.0cm)$) -- ($(phantomA4) + (-.1cm,-.05cm)$) -- ($(phantomA4) + (-.2cm,-.6cm)$) -- (t2A4);
\node (s1A5) [s1,above left=of phantomA5] {};
\node (s2A5) [s2,above right=of phantomA5] {};
\node (t2A5) [t2,below left=of phantomA5] {};
\node (t1A5) [t1,below right=of phantomA5] {};
\draw [very thick,densely dotted] (s1A5) -- ($(phantomA5) + (-.3cm,.05cm)$) -- ($(phantomA5) + (.1cm,.2cm)$) -- ($(phantomA5) + (.2cm,-.5cm)$) -- (t1A5);
\draw [very thick] (s2A5) -- ($(phantomA5) + (.2cm,.45cm)$) -- ($(phantomA5) + (.4cm,.0cm)$) -- ($(phantomA5) + (-.1cm,-.05cm)$) -- ($(phantomA5) + (-.2cm,-.6cm)$) -- (t2A5);
\node (s1A6) [s1,above left=of phantomA6] {};
\node (s2A6) [s2,above right=of phantomA6] {};
\node (t2A6) [t2,below left=of phantomA6] {};
\node (t1A6) [t1,below right=of phantomA6] {};
\draw [very thick,densely dotted] (s1A6) -- ($(phantomA6) + (-.3cm,.05cm)$) -- ($(phantomA6) + (.1cm,.2cm)$) -- ($(phantomA6) + (.2cm,-.5cm)$) -- (t1A6);
\draw [very thick] (s2A6) -- ($(phantomA6) + (.2cm,.45cm)$) -- ($(phantomA6) + (.4cm,.0cm)$) -- ($(phantomA6) + (-.1cm,-.05cm)$) -- ($(phantomA6) + (-.2cm,-.6cm)$) -- (t2A6);
\node (s1A7) [s1,above left=of phantomA7] {};
\node (s2A7) [s2,above right=of phantomA7] {};
\node (t2A7) [t2,below left=of phantomA7] {};
\node (t1A7) [t1,below right=of phantomA7] {};
\draw [very thick,densely dotted] (s1A7) -- ($(phantomA7) + (-.3cm,.05cm)$) -- ($(phantomA7) + (.1cm,.2cm)$) -- ($(phantomA7) + (.2cm,-.5cm)$) -- (t1A7);
\draw [very thick] (s2A7) -- ($(phantomA7) + (.2cm,.45cm)$) -- ($(phantomA7) + (.4cm,.0cm)$) -- ($(phantomA7) + (-.1cm,-.05cm)$) -- ($(phantomA7) + (-.2cm,-.6cm)$) -- (t2A7);
\node (s1A8) [s1,above left=of phantomA8] {};
\node (s2A8) [s2,above right=of phantomA8] {};
\node (t2A8) [t2,below left=of phantomA8] {};
\node (t1A8) [t1,below right=of phantomA8] {};
\draw [very thick,densely dotted] (s1A8) -- ($(phantomA8) + (-.3cm,.05cm)$) -- ($(phantomA8) + (.1cm,.2cm)$) -- ($(phantomA8) + (.2cm,-.5cm)$) -- (t1A8);
\draw [very thick] (s2A8) -- ($(phantomA8) + (.2cm,.45cm)$) -- ($(phantomA8) + (.4cm,.0cm)$) -- ($(phantomA8) + (-.1cm,-.05cm)$) -- ($(phantomA8) + (-.2cm,-.6cm)$) -- (t2A8);
\node (s1A9) [s1,above left=of phantomA9] {};
\node (s2A9) [s2,above right=of phantomA9] {};
\node (t2A9) [t2,below left=of phantomA9] {};
\node (t1A9) [t1,below right=of phantomA9] {};
\draw [very thick,densely dotted] (s1A9) -- ($(phantomA9) + (-.3cm,.05cm)$) -- ($(phantomA9) + (.1cm,.2cm)$) -- ($(phantomA9) + (.2cm,-.5cm)$) -- (t1A9);
\draw [very thick] (s2A9) -- ($(phantomA9) + (.2cm,.45cm)$) -- ($(phantomA9) + (.4cm,.0cm)$) -- ($(phantomA9) + (-.1cm,-.05cm)$) -- ($(phantomA9) + (-.2cm,-.6cm)$) -- (t2A9);
\node (s1A10) [s1,above left=of phantomA10] {};
\node (s2A10) [s2,above right=of phantomA10] {};
\node (t2A10) [t2,below left=of phantomA10] {};
\node (t1A10) [t1,below right=of phantomA10] {};
\draw [very thick,densely dotted] (s1A10) -- ($(phantomA10) + (-.3cm,.05cm)$) -- ($(phantomA10) + (.1cm,.2cm)$) -- ($(phantomA10) + (.2cm,-.5cm)$) -- (t1A10);
\draw [very thick] (s2A10) -- ($(phantomA10) + (.2cm,.45cm)$) -- ($(phantomA10) + (.4cm,.0cm)$) -- ($(phantomA10) + (-.1cm,-.05cm)$) -- ($(phantomA10) + (-.2cm,-.6cm)$) -- (t2A10);
\node (s1A11) [s1,above left=of phantomA11] {};
\node (s2A11) [s2,above right=of phantomA11] {};
\node (t2A11) [t2,below left=of phantomA11] {};
\node (t1A11) [t1,below right=of phantomA11] {};
\draw [very thick,densely dotted] (s1A11) -- ($(phantomA11) + (-.3cm,.05cm)$) -- ($(phantomA11) + (.1cm,.2cm)$) -- ($(phantomA11) + (.2cm,-.5cm)$) -- (t1A11);
\draw [very thick] (s2A11) -- ($(phantomA11) + (.2cm,.45cm)$) -- ($(phantomA11) + (.4cm,.0cm)$) -- ($(phantomA11) + (-.1cm,-.05cm)$) -- ($(phantomA11) + (-.2cm,-.6cm)$) -- (t2A11);
\node (s1A12) [s1,above left=of phantomA12] {};
\node (s2A12) [s2,above right=of phantomA12] {};
\node (t2A12) [t2,below left=of phantomA12] {};
\node (t1A12) [t1,below right=of phantomA12] {};
\draw [very thick,densely dotted] (s1A12) -- ($(phantomA12) + (-.3cm,.05cm)$) -- ($(phantomA12) + (.1cm,.2cm)$) -- ($(phantomA12) + (.2cm,-.5cm)$) -- (t1A12);
\draw [very thick] (s2A12) -- ($(phantomA12) + (.2cm,.45cm)$) -- ($(phantomA12) + (.4cm,.0cm)$) -- ($(phantomA12) + (-.1cm,-.05cm)$) -- ($(phantomA12) + (-.2cm,-.6cm)$) -- (t2A12);
\node (s1A13) [s1,above left=of phantomA13] {};
\node (s2A13) [s2,above right=of phantomA13] {};
\node (t2A13) [t2,below left=of phantomA13] {};
\node (t1A13) [t1,below right=of phantomA13] {};
\draw [very thick,densely dotted] (s1A13) -- ($(phantomA13) + (-.3cm,.05cm)$) -- ($(phantomA13) + (.1cm,.2cm)$) -- ($(phantomA13) + (.2cm,-.5cm)$) -- (t1A13);
\draw [very thick] (s2A13) -- ($(phantomA13) + (.2cm,.45cm)$) -- ($(phantomA13) + (.4cm,.0cm)$) -- ($(phantomA13) + (-.1cm,-.05cm)$) -- ($(phantomA13) + (-.2cm,-.6cm)$) -- (t2A13);
\node (s1A14) [s1,above left=of phantomA14] {};
\node (s2A14) [s2,above right=of phantomA14] {};
\node (t2A14) [t2,below left=of phantomA14] {};
\node (t1A14) [t1,below right=of phantomA14] {};
\draw [very thick,densely dotted] (s1A14) -- ($(phantomA14) + (-.3cm,.05cm)$) -- ($(phantomA14) + (.1cm,.2cm)$) -- ($(phantomA14) + (.2cm,-.5cm)$) -- (t1A14);
\draw [very thick] (s2A14) -- ($(phantomA14) + (.2cm,.45cm)$) -- ($(phantomA14) + (.4cm,.0cm)$) -- ($(phantomA14) + (-.1cm,-.05cm)$) -- ($(phantomA14) + (-.2cm,-.6cm)$) -- (t2A14);
\node (s1A15) [s1,above left=of phantomA15] {};
\node (s2A15) [s2,above right=of phantomA15] {};
\node (t2A15) [t2,below left=of phantomA15] {};
\node (t1A15) [t1,below right=of phantomA15] {};
\draw [very thick,densely dotted] (s1A15) -- ($(phantomA15) + (-.3cm,.05cm)$) -- ($(phantomA15) + (.1cm,.2cm)$) -- ($(phantomA15) + (.2cm,-.5cm)$) -- (t1A15);
\draw [very thick] (s2A15) -- ($(phantomA15) + (.2cm,.45cm)$) -- ($(phantomA15) + (.4cm,.0cm)$) -- ($(phantomA15) + (-.1cm,-.05cm)$) -- ($(phantomA15) + (-.2cm,-.6cm)$) -- (t2A15);

\draw [->,thick] (t2A1) edge (s1A2);
\draw [->,thick] (t2A2) edge (s1A4);
\draw [->,thick] (t2A3) edge (s1A6);
\draw [->,thick] (t2A4) edge (s1A8);
\draw [->,thick] (t2A5) edge (s1A10);
\draw [->,thick] (t2A6) edge (s1A12);
\draw [->,thick] (t2A7) edge (s1A14);

\draw [->,thick] (t1A1) edge (s1A3);
\draw [->,thick] (t1A2) edge (s1A5);
\draw [->,thick] (t1A3) edge (s1A7);
\draw [->,thick] (t1A4) edge (s1A9);
\draw [->,thick] (t1A5) edge (s1A11);
\draw [->,thick] (t1A6) edge (s1A13);
\draw [->,thick] (t1A7) edge (s1A15);

\node (s2A8) [sink] at ($(t2A8)+(vsph)$) {}
edge [<-,thick] (t2A8);
\node (s1A8) [sink] at ($(t1A8)+(vsph)$) {}
edge [<-,thick] (t1A8);
\node (s2A9) [sink] at ($(t2A9)+(vsph)$) {}
edge [<-,thick] (t2A9);
\node (s1A9) [sink] at ($(t1A9)+(vsph)$) {}
edge [<-,thick] (t1A9);
\node (s2A10) [sink] at ($(t2A10)+(vsph)$) {}
edge [<-,thick] (t2A10);
\node (s1A10) [sink] at ($(t1A10)+(vsph)$) {}
edge [<-,thick] (t1A10);
\node (s2A11) [sink] at ($(t2A11)+(vsph)$) {}
edge [<-,thick] (t2A11);
\node (s1A11) [sink] at ($(t1A11)+(vsph)$) {}
edge [<-,thick] (t1A11);
\node (s2A12) [sink] at ($(t2A12)+(vsph)$) {}
edge [<-,thick] (t2A12);
\node (s1A12) [sink] at ($(t1A12)+(vsph)$) {}
edge [<-,thick] (t1A12);
\node (s2A13) [sink] at ($(t2A13)+(vsph)$) {}
edge [<-,thick] (t2A13);
\node (s1A13) [sink] at ($(t1A13)+(vsph)$) {}
edge [<-,thick] (t1A13);
\node (s2A14) [sink] at ($(t2A14)+(vsph)$) {}
edge [<-,thick] (t2A14);
\node (s1A14) [sink] at ($(t1A14)+(vsph)$) {}
edge [<-,thick] (t1A14);
\node (s2A15) [sink] at ($(t2A15)+(vsph)$) {}
edge [<-,thick] (t2A15);
\node (s1A15) [sink] at ($(t1A15)+(vsph)$) {}
edge [<-,thick] (t1A15);
        \end{tikzpicture}}
        \caption{Auxiliary network generated from $G$, here for $k = 16$. Recreation of \cite[Fig.~2]{chen}.}
        \label{fig:approxreduction}
    \end{figure}
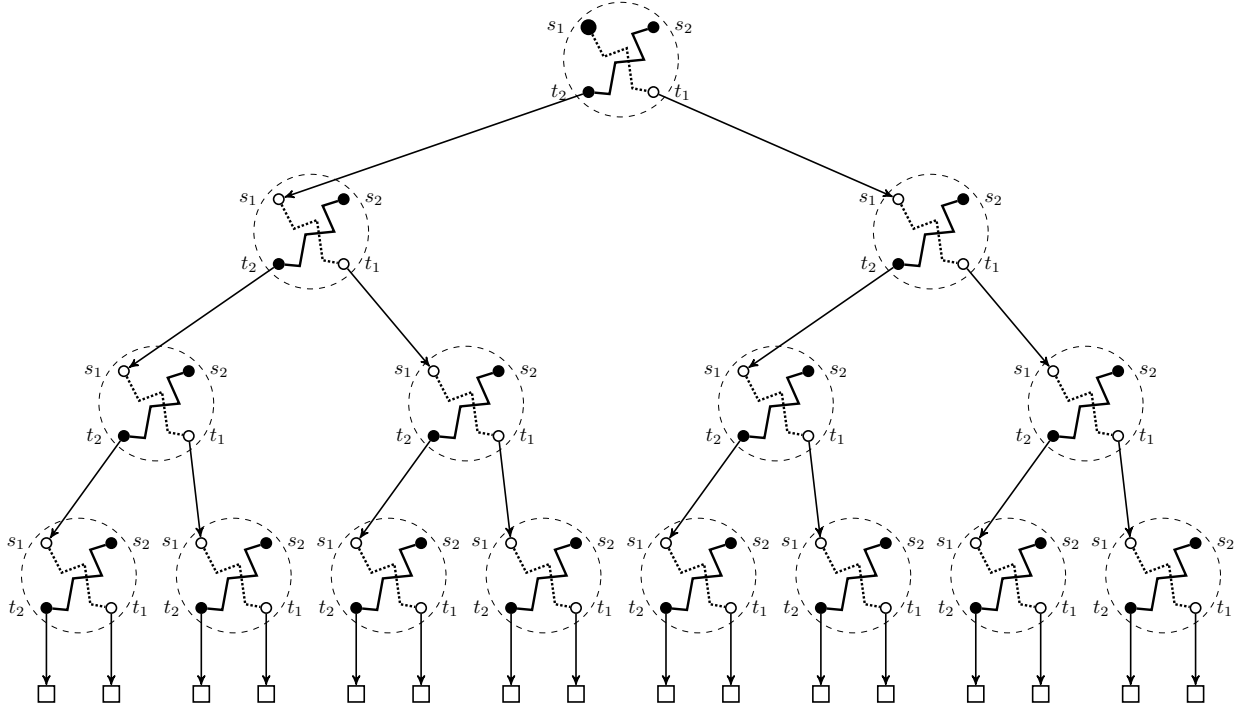

    Without loss of generality, we can have polynomially-bounded positive integer demands.
    To express a demand of $d$ at a node $n$ in our unit-demand setting, add $d - 1$ nodes with a single outgoing edge to $n$.

    Denote the number of nodes in the network by $\phi \coloneqq (k - 1) \cdot |V| + k$, and set $\Phi \coloneqq \ell \cdot \phi + 1$.
    In \cite{chen}, every copy of $s_2$ and $t_2$ has demand 1, the copy of $s_1$ at the root has demand 2, and all other nodes have demand 0.
    Instead, we give these nodes demands of $\Phi$, $2 \Phi$ and $1$, respectively.
    Note that the size of the generated network\footnote{Even after unfolding our non-unitary-demand nodes.} is polynomial in the size of $G$ and that the outdegree of each node is at most 2.
    From every node, one of the sinks $S$ displayed as rectangles in \cref{fig:approxreduction} is reachable.
    Since the minimum-distance-to-$S$ spanning forest describes a flow, a flow in the network exists.
    
    Suppose that $G$ contains vertex-disjoint paths $P_1$ from $s_1$ to $t_1$ and $P_2$ from $s_2$ to $t_2$.
    In each copy of $G$ in the network, route the flow along these paths.
    We can complete the confluent flow inside of this copy in such a way that the demand of every node is routed to $t_1$ or $t_2$:
    By assumption, each of the nodes can reach one of these two path endpoints.
    Iterate over all nodes in order of ascending distance to the closest endpoint and make sure that their flow is routed to an endpoint.
    For the endpoints themselves, there is nothing to do.
    For positive distance, a node might be part of a path and thus already connected to an endpoint.
    Else, look at its successor in a shortest path to an endpoint.
    By the induction hypothesis, all flow from this successor is routed to an endpoint, so route the node's flow to this successor.
    If we also use the edges between copies of $G$ and between the copies and the sinks, we obtain a confluent flow.
    Each sink except for the rightmost one can only collect the demand of two nodes with demand $\Phi$ plus a number of nodes with demand $1$.
    The rightmost sink collects the demand from the single node with demand $2 \Phi$ plus some unitary demands.
    Thus, the congestion of the system can be at most $2 \Phi + \phi$.

    Now, consider the case in which $G$ does not have such vertex-disjoint paths.
    In every confluent flow and in every copy, there are three options:
    \begin{itemize}
        \item the flow from $s_1$ flows to $t_2$ and the flow from $s_2$ flows to $t_1$,
        \item the flow from $s_1$ and $s_2$ flows to $t_1$, or
        \item the flow from $s_1$ and $s_2$ flows to $t_2$.
    \end{itemize}
    In each case, the flow coming in through $s_1$ is joined by additional demand of at least $\Phi$.
    Consider the path from the copy of $s_1$ at the root to a sink.
    By a simple induction, the congestion at the endpoind of the $i$th copy of $G$ is at least $(i + 1) \cdot \Phi$.
    Thus, the total congestion at the sink must be at least $(\ell + 1) \cdot \Phi$. The lemma now follows from the fact that 
	$$
	\frac{\log_2 k}{2}(2\Phi+\phi) =\frac{\ell}{2}(2\Phi+\phi)<(\ell+1)\cdot\Phi.
	$$
\end{proof}

\section{Proof of Theorem~\lowercase{\ref{thm:lbwhp}} -- Lower Bound for \lowercase{$k = 1$} with High Probability}
\label{app:lbwhp}

\restlbwhp*

\begin{proof}
It suffices to show that, with high probability, there exists a voter at every time $t$ whose weight is bounded from below by a function in $\Omega(t^\beta)$.

For ease of exposition, we pretend that $\imax \coloneqq \log_2 \frac{t}{\log t}$ is an integer.\footnote{The same argument works for $\imax \coloneqq \left\lfloor \log_2 \frac{t}{\log t} \right\rfloor$ if we appropriately bound the term.}
We divide the $t$ agents into $\imax + 1$ blocks $B_0, \dots, B_{\imax}$.
The first block $B_0$ contains agents $1$ to $\tau \coloneqq \log t$, and every subsequent block $B_i$ contains agents $(\tau \, 2^{i - 1}, \tau \, 2^i]$.

We keep track of the total weight $S_i$ of all voters in $B_0$ after the entirety of block $B_i$ has been added.
Furthermore, we define an event $X_i$ saying that a high enough number of agents in block $B_i$ transitively delegate into $B_0$.
If all $X_i$ hold, $S_{\imax}$ scales like a power function.
Then, we show that, as $t$ increases, the probability of any $X_i$ failing goes to zero.
Thus, our bound on $S_{\imax}$ holds with high probability.
The total weight of $B_0$ and the weight of the maximum voter in $B_0$ can differ by at most a factor of $\tau$, which is logarithmic.
Thus, with high probability, there is a voter in $B_0$ whose weight is a power function.

In more detail, let $\e \coloneqq \frac{1}{2}$ and let $d' \coloneqq (1 - \e) \, d = \frac{d}{2}$.
For each $i \geq 0$, let $Y_i$ denote the number of votes from block $i$ transitively going into $B_0$.
Clearly, $S_i = \sum_{j = 0}^{i} Y_i$.
For $i > 0$, let $X_i$ denote the event that
\[ Y_i > d' \, \frac{\tau \, \left(1 + \frac{d'}{2}\right)^{i-1}}{2}\text{.} \]

\paragraph{Bounding the Expectation of $Y_i$} 
We first prove by induction on $i$ that, if $X_1$ through $X_i$ hold, then 
\begin{equation}
\label{eq:si}
S_i \geq \tau \left( 1 + \frac{d'}{2} \right)^i.
\end{equation}
For $i = 0$, $S_0 = \tau$ and the claim holds. For $i > 0$, by the induction hypothesis, $S_{i - 1} \geq \tau \left( 1 + \frac{d'}{2} \right)^{i - 1}$.
    By the assumption $X_i$, $$Y_i > d' \, \frac{\tau \, \left(1 + \frac{d'}{2}\right)^{i-1}}{2}.$$
    Thus,
    \[ S_{i} = S_{i - 1} + Y_i \geq \tau \left( 1 + \frac{d'}{2} \right)^{i - 1} + d' \, \frac{\tau \, \left(1 + \frac{d'}{2}\right)^{i-1}}{2} = \tau \, \left(1 + \frac{d'}{2}\right)^{i-1} \, \left(1 + \frac{d'}{2}\right) = \tau \, \left(1 + \frac{d'}{2}\right)^i \text{.}\]
    This concludes the induction and establishes \cref{eq:si}.

Now, for any agent $j$ in $B_i$, the probability of transitively delegating into $B_0$ is
\[ d \, \frac{\sum_{v \in V \cap B_0} w_{j-1}(v)}{j - 1} \geq d \, \frac{S_{i - 1}}{\tau \, 2^i}.\]
    Conditioned on $X_1,\ldots,X_{i-1}$, we can thus lower-bound $Y_i$ by a binomial variable $\mathit{Bin}\left(\tau \, 2^{i - 1}, d \, \frac{S_{i - 1}}{\tau \, 2^i}\right)$ to obtain
\[ \mathbb{E}[Y_i\ |\ X_1,\ldots,X_{i-1}] \geq \tau \, 2^{i - 1} \, d \, \frac{S_{i - 1}}{\tau \, 2^i} = d \, \frac{S_{i - 1}}{2} \geq d \, \frac{\tau \, \left(1 + \frac{d'}{2}\right)^{i-1}}{2}. \]
Denoting the right hand side by 
$$
\mu\coloneq d \, \frac{\tau \, \left(1 + \frac{d'}{2}\right)^{i-1}}{2},
$$ 
note that $X_i$ holds if $Y_i > (1 - \e) \, \mu$.

\paragraph{Failure Probability Goes to 0}
Now, we must show that, with high probability, all $X_i$ hold.
By underapproximating the probability of delegation by a binomial random variable as before and by using a Chernoff bound, we have for all $i > 0$
    \begin{align*}
        \Prob[X_i \mid X_1, \dots, X_{i - 1}] \geq \Prob \left[ \mathit{Bin}\left( \tau \, 2^{i - 1}, d \, \frac{\tau \, \left(1 + d'/2\right)^{i - 1}}{\tau \, 2^i} \right) > (1 - \e{}) \, \mu \right] \geq 1 - e^{-\frac{\e{}^2 \, \mu}{2}}.
    \end{align*}
By the union bound, 
\begin{align*}
    \Prob[\exists i, 1 \leq i \leq \imax. \, \text{$X_i$ fails}] \leq \sum_{i=1}^{\imax} e^{-\frac{\e^2 \, d \, \tau \, \left(1 + d'/2 \right)^{i-1}}{4}}.
\end{align*}
We wish to show that the right hand side goes to 0 as $t$ increases. We have
\begin{align*}
    \sum_{i=1}^{\imax} e^{-\frac{\e^2 \, d \, \tau \, \left(1 + d'/2 \right)^{i-1}}{4}} &\leq \imax \left( e^{-\frac{\e^2 \, d \, \tau}{4}} \right) & \text{(by monotonicity)}\\
    &= \left( \log_2 \frac{t}{\log t}\right) \left( t^{-\frac{\e^2 \, d}{4}} \right), & \text{(by definitions of $\imax$, $\tau$)}
\end{align*}
which indeed approaches $0$ as $t$ increases.

\paragraph{Bounding the Maximum Weight}
Note that the weight of $B_0$ at time $t$ is exactly $S_{\imax}$.
Set $x \coloneqq 1 + d'/2 > 1$, which is a constant.
With high probability, by \cref{eq:si}, 
\begin{align*}
    \frac{S_{\imax}}{\tau} &\geq \left( 1 + \frac{d'}{2} \right)^{\imax} 
    = x^{\log_2 \frac{t}{\log t}} 
    = \left( \frac{t}{\log t} \right)^{\log_2 x} \text{.}
\end{align*}

Since $x > 1$, $\log_2 x > 0$.
For any $0 < \beta < \log_2 x$, $\frac{S_{\imax}}{\tau} \in \Omega(t^{\beta})$ with high probability.
Since $B_0$ has weight $S_{\imax}$ and contains at most $\tau$ voters, with high probability there is some voter in $B_0$ with that much weight.
\end{proof}

\section{Proof of Theorem~\lowercase{\ref{thm:thm1}} -- Upper Bound}
\label{sec:upperchoiceproofs}

Because the proof of \cref{thm:thm1} is quite intricate and technical, we begin with a sketch of its structure.
Proofs for the individual lemmas can be found in the subsequent subsections.

\subsection{Proof Sketch}
For our analysis, it would be natural to keep track of the number of voters $v$ with a specific weight $w_j(v) = k$ at a specific point $j$ in time.
In order to simplify the analysis, we instead keep track of random variables
\[ F_j(k) \coloneqq \sum_{\substack{v \in V \\ w_j(v) \geq k}} w_j(v)\text{,} \]
i.e., we sum up the weights of all voters with weight at least $k$.
Since the total weight increases by one in every step, we have
\begin{gather}
    \label{eq:fullfj1}
    \forall j. \; F_j(1) = j \text{, and} \\
    \label{eq:fullfj}
    \forall j, k.\;F_j(k) \leq j \text{.}
\end{gather}
If $F_j(k) < k$ for some $j$ and $k$, the maximum weight of any voter must be below $k$. 

If we look at a specific $k > 1$ in isolation, the sequence $(F_j(k))_j$ evolves as a Markov process initialized at $F_1(k) = 0$ and then governed by the rule
\begin{equation}
\label{eq:fdifference}
F_{m+1}(k) - F_m(k) = \begin{cases}
                         1 & \Prob = d \, \left( \frac{F_m(k)}{m} \right)^2 \\
                         k & \Prob = d \, \left( \left( \frac{F_m(k-1)}{m} \right)^2 - \left( \frac{F_m(k)}{m} \right)^2 \right) \\
                         0 & \text{else}
                         \end{cases}\text{.}
\end{equation}
In the first case, both potential delegations of a new delegator lead to voters who already had weight at least $k$.
We must thus give her vote to one of them, increasing $F_m(k)$ by one.
In the second case, a new delegator offers two delegations leading to voters of weight at least $k - 1$, at least one of which has exactly weight $k - 1$.
Our greedy algorithm will then choose a voter with weight $k - 1$.
Because this voter is suddenly counted in the definition of $F_j(k)$, $F_m(k)$ increases by $k$.
Finally, if a new voter appears, or if a new delegator can transitively delegate to a voter with weight less than $k - 1$, then $F_m(k)$ does not change.

In order to bound the maximum weight of a voter, we first need to get a handle on the general distribution of weights.
For this, we define a sequence of real numbers $(\alpha_k)_k$ such that, for every $k \geq 1$, the sequence $\frac{F_j(k)}{j}$ converges in probability to $\alpha_k$.
Set $\alpha_1 \coloneqq 1$. For every $k > 1$, let $\alpha_k$ be the unique root $0 < x < \alpha_{k-1}$ of the polynomial
\begin{equation}
    \label{eq:alphadef}
    a_k(x, p) \coloneqq d \, x^2 + k \, d \, (p^2 - x^2) - x
\end{equation}
for $p$ set to $\alpha_{k-1}$.\footnote{The equation $0 = a_k(x, p)$ can be obtained from \cref{eq:fdifference} by na\"ively assuming that $\frac{F_{j}(k - 1)}{j}$ converges to a value $p$ and $\frac{F_j(k)}{j}$ converges to $x$, then plugging these values in the expectation of the recurrence.}
Since $a_k(0, \alpha_{k-1}) > 0$ and $a_k(\alpha_{k-1}, \alpha_{k-1}) < 0$, such a solution exists by the intermediate value theorem.
Because the polynomial is quadratic, such a solution must be unique in the interval.
It follows that the $\alpha_k$ form a strictly decreasing sequence in the interval $(0, 1]$.

The sequence $(\alpha_k)_k$ converges to zero, and eventually does so very fast.
However, this is not obvious from the definition and, depending on $d$, the sequence can initially fall slowly.
In \cref{lem:alphatozero}, we demonstrate convergence to zero, and in \cref{lem:alphatozerofast}, we show that the sequence falls in $\mathcal{O}(k^{-2})$.
Based on this, in \cref{lem:3.2}, we choose an integer $k_0$ such that the sequence falls very fast from there.
In the same lemma, we define a more nicely behaved sequence $(f(k))_{k \geq k_0}$ that is a strict upper bound on $(\alpha_k)_{k \geq k_0}$ and that is contained between two doubly-exponentially decaying functions.

\begin{restatable}{lem}{lemthreeone}
\label{lem:3.1}
For all $k \geq 1$, $\e > 0$ and functions $\omega(m)$ such that $\omega(m) \to \infty$ and $\omega(m) < m$ (for sufficiently large $m$),
\[ \Prob \left[\exists j,\; \omega(m) \leq j \leq m.\;\frac{F_j(k)}{j} > \alpha_k + \e{} \right] \to 0\text{.} \]
\end{restatable}
\begin{proof}[Proof sketch (detailed in \cref{pf:3.1})]
    The proof proceeds by induction on $k$.
    For $k = 1$, the claim directly holds.
    For larger $k$, we use a suitably chosen $\delta$ in place of $\e$ and $\omega_0$ in place of $\omega$ for the induction hypothesis.
    With the induction hypothesis, we bound the $\frac{F_{m}(k - 1)}{m}$ term in the recurrence in \cref{eq:fdifference}.
    Furthermore, all steps $F_j(k) - F_{j - 1}(k)$ where $\frac{F_{j - 1}(k)}{j - 1} \geq \alpha_k$ holds can be dominated by independent and \emph{identically distributed} random variables $\eta_j'$.

    Denote by $\pi$ the first point $j \geq \omega_0(m)$ such that $\frac{F_j(k)}{j} \leq \alpha_k + \frac{\e{}}{2}$.
    The $\eta_j'$ then dominate all steps $F_j(k) - F_{j - 1}(k)$ for $\omega_0(m) < j \leq \pi$.
    Using Chernoff's bound and suitably chosen $\delta$ and $\omega_0$, we show that, with high probability, $\pi \leq \omega(m)$.

    Because of this, if $\frac{F_j(k)}{j} > \alpha_k + \e$ for some $j \geq \omega(m)$, the sequence $\left(\frac{F_j(k)}{j}\right)_j$ must eventually cross from below $\alpha_k + \frac{\e{}}{2}$ to above $\alpha_k + \e$ without in between falling below $\alpha_k$.
    On this segment, we can overapproximate the sequence by a random walk with steps distributed as $\eta_j'$.
    Since the sequence might previously fall below $\alpha_k$ an arbitrary number of times, we overapproximate the probability of ever crossing $\alpha_k + \e$ for $j \geq \omega(m)$ by a sum over infinitely many random walks.
    This sum converges to $0$ for $m \to \infty$, which shows our claim. 
\end{proof}

The above lemma gives us a good characterization of the behavior of $(F_j(k))_j$ for any fixed $k$ (and large enough $j$).
To prove an upper bound on the maximum weight, however, we are ultimately interested in statements about $F_j(k(m))$, where $k(m) \in \Theta(\log_2 \log m)$ and the range of $j$ varies with $m$.
In order to obtain such results, we will first show in \cref{lem:3.3} that whole ranges of $k$ simultaneously satisfy bounds with high probability.

As in the previous lemma, we can only show our bounds with high probability for $j$ past a certain period of initial chaos.
Taking a role similar to $\omega(m)$ in \cref{lem:3.1}, we will define a function $\phi(m, k)$ that gives each $k$ a certain amount of time to satisfy the bounds, depending on $m$:
Let $\rho(m) \coloneqq (\log \log m)^{\frac{1}{3}}$ and define $\phi(m,k) \coloneqq \rho(m) \, C^{2^{k+1}}$, where $C$ is an integer sufficiently large such that
\begin{equation}
    \label{eq:defcapc}
    \log C > \text{max} \left(1, \; c_1, \; \log \left( \frac{2}{1-d} \right) + \frac{c_1}{2} \right).
\end{equation}
In the above, $c_1$ is a postive constant defining the lower bound on $f(k)$ in \cref{lem:3.2}.

Additionally, let $k_*(m)$ be the smallest integer such that
\begin{equation}
    \label{eqn:kstar}
    C^{2^{k_*(m)+1}} \geq \sqrt{m}.
\end{equation}
Note that $C^{2^{k_*(m)+1}} < m$ because increasing the double exponent in increments of $1$ is equivalent to squaring the term.
By applying logarithms to $C^{2^{k_*(m)+1}} \geq \sqrt{m}$ and $C^{2^{k_*(m)+1}} < m$, we obtain $\log_2 \log_C m - 2 \leq k_*(m) < \log_2 \log_C m - 1$, from which it follows that $k_*(m) = \log_2 \log m + \Theta(1)$.

\begin{restatable}{lem}{lemthreethree}
\label{lem:3.3}
With high probability, for all $k_0 \leq k \leq k_*(m)$, and for all $\phi(m,k) \leq j \leq m$, \[ \frac{F_j(k)}{j} \leq f(k)\text{.} \]
\end{restatable}
\begin{proof}[Proof sketch (detailed in \cref{pf:3.3})]
Let $\mathcal{G}_k$ be the event
\[ \mathcal{G}_k \coloneqq \left\{\forall j, \phi(m,k) \leq j \leq m. \; \frac{F_j(k)}{j} \leq f(k) \right\}. \]
Our goal is to show that $\mathcal{G}_k$ holds for all $k$ in our range.
Similarly to an induction, we begin by showing $\mathcal{G}_{k_0}$ with high probability and then give evidence for how, under the assumption $\mathcal{G}_{k}$, $\mathcal{G}_{k + 1}$ is likely to happen.
Instead of an explicit induction, we piece together these parts in a union bound.

    The base case $\mathcal{G}_{k_0}$ follows from \cref{lem:3.1} with $\omega(m) \coloneqq \phi(m, k_0)$ and $\e \coloneqq f(k_0) - \alpha_{k_0}$.

    For the step, fix some $k \geq k_0$, and assume $\mathcal{G}_k$.
    We want to give an upper bound on the probability that $\mathcal{G}_{k + 1}$ happens.
We split this into multiple substeps:
    First, we prove that, given $\mathcal{G}_k$, some auxiliary event $\mathcal{E}(k + 1)$ happens only with probability converging to 0.
    Then, we show that $\overline{\mathcal{E}(k + 1)} \subseteq \mathcal{G}_{k + 1}$ where $\overline{\mathcal{E}}$ denotes the complement of an event $\mathcal{E}$.
    This means that, whenever the unlikely event does not take place, $\mathcal{G}_{k + 1}$ holds.
This allows the step to be repeated.

    If $\mathcal{G}_k$ does not hold for any $k_0 \leq k \leq k_*(m)$, then $\overline{\mathcal{G}_{k_0}}$ or one of the $\mathcal{E}(k)$ must have happened.
    The union bound converges to zero for $m \to \infty$, proving our claim.
\end{proof}

As promised, the last lemma enables us to speak about the behavior of $F_j(k(m))$. We will use a sequence of such statements to show that, with high probability, $F_j(k(m))$ for some $k(m)$ does not change over a whole range of $j$:
\begin{restatable}{lem}{lemthreesix}
\label{lem:3.6}
There exists $M > 0$ and an integer $r > 0$ such that, for $j_0(m) \coloneqq (\log \log m)^M$,
$F_m(k_*(m)+r) = F_{j_0(m)}(k_*(m)+r)$
    holds with high probability. In addition, there is $\beta > \frac{1}{2}$ such that, with high probability,
    \begin{equation}
        \label{eq:3.6add}
        F_{j_0(m)}(k_*(m) + r - 1) \leq j_0(m)^{1 - \beta}\text{.}
    \end{equation}
\end{restatable}
\begin{proof}[Proof sketch (detailed in \cref{pf:3.6})]
    In \cref{lem:beta0}, we finally get a statement about $F_j(k_*(m))$:
    By choosing different $k$ for different $j$ in \cref{lem:3.3}, we obtain a constant $\beta_0 > 0$ such that, with high probability,
    \[ \forall j, \log \log m \leq j \leq m. \; \frac{F_j(k_*(m))}{j} \leq j^{- \beta_0}\text{.} \] 
    
    We now increase $\beta_0$ until it is larger than $\frac{1}{2}$.
    Set $r'_0 \coloneqq 0$ and $M_0 \coloneqq 1$.
    In \cref{lem:3.5}, we boost a proposition of the form
    \[ \forall j, (\log \log m)^{M_i} \leq j \leq m. \; \frac{F_j(k_*(m) + r'_i)}{j} \leq j^{- \beta_i} \]
    holding with high probability to obtain, for some $M_{i + 1} > 0$ and with high probability,
    \[ \forall j, (\log \log m)^{M_{i + 1}} \leq j \leq m. \; \frac{F_j(k_*(m) + r'_i + 1)}{j} \leq j^{- \frac{3}{2}\,\beta_i}\text{.} \]
    If we set $r'_{i + 1} \coloneqq r'_i + 1$ and $\beta_{i + 1} \coloneqq \frac{3}{2} \, \beta_i$, we can repeatedly apply this argument until some $\beta_i > \frac{1}{2}$.
    Let $M$, $r'$ and $\beta$ denote $M_i$, $r'_i$ and $\beta_i$, respectively, for this $i$.
    If, furthermore, $r \coloneqq r' + 1$, \cref{eq:3.6add} follows as a special case.

    We then simply union-bound the probability of $F_j(k_*(m) + r)$ increasing for any $j$ between $j_0(m)$ and $m$.
    Using the above over-approximation in \cref{eq:fdifference} gives us an over-harmonic series, whose value goes to zero with $m \to \infty$.
\end{proof}

We can now prove \cref{thm:thm1}. Let $Q_i$ denote the maximum weight after $i$ time steps.
\begin{proof}[Proof of \cref{thm:thm1}]
By \cref{lem:3.6}, with high probability, $F_m(k_*(m)+r) = F_{j_0(m)}(k_*(m)+r)$. Therefore, we have that with high probability
\begin{align*}
    F_m(k_*(m)+r) &= F_{j_0(m)}(k_*(m)+r) &\\
    & \leq F_{j_0(m)}(k_*(m)+r-1) & \text{(by monotonicity)} \\
    & \leq j_0(m)^{1-\beta} & \text{(by \cref{eq:3.6add})} \\
    & = \left((\log \log m)^M\right)^{1-\beta} &\\
    & \leq (\log \log m)^{M + 1}.&
\end{align*}

For any $j$ and $k$, $Q_j\leq \max \{k, F_j(k)\}$.
Since, for large enough $m$, $k_*(m) + r < (\log \log m)^{M + 1}$, the maximum weight $Q_m$ is at most $(\log \log m)^{M + 1}$ with high probability. This result holds for general $m$, so we are allowed to plug in $j_0(m)$ for $m$.
Then, $Q_{j_0(m)} \leq \left(\log \log j_0(m) \right)^{M + 1}$. Moreover, $\left(\log \log j_0(m) \right)^{(M + 1)^2} < j_0(m)$ for sufficiently large $m$ because $M$ is a constant and polylogarithmic terms grow asymptotically slower than polynomial terms. Rewriting this yields 
\begin{equation}
\label{eq:jeefa}
Q_{j_0(m)} \leq \left(\log \log j_0(m) \right)^{M + 1} < j_0(m)^{1/(M + 1)}.
\end{equation}

Now, note that $k_*(m) + r \geq \left( j_0(m)^{1/(M + 1)} \right)$ for large enough $m$. Therefore, \cref{eq:jeefa} implies that, with high probability, a graph generated in $j_0(m)$ time steps has no voters of weight $k_*(m) + r$ or higher. In other words, with high probability, $F_{j_0(m)}(k_*(m)+r) = 0$, so with high probability $F_{m}(k_*(m)+r) = 0$ (again by Lemma~\ref{lem:3.6}). 
    This means that the maximum weight after $m$ time steps is also upper-bounded by $k_*(m)+r = \log_2 \log m + \Theta(1)$.
\end{proof}

\subsection{Bounds on $(\alpha_k)_k$}
\begin{lem}
\label{lem:alphatozero}
\[ \alpha_{\infty} \coloneqq \lim_{k \to \infty} \alpha_k = \alpha_{\infty} = 0\]
\end{lem}
\begin{proof}
Set
\[ a(x) \coloneqq a_k(x, x) = d \, x^2 - x \text{,}\]
which is independent of $k$.
Since $a(x)$ is continuous, it must take on a maximum value $\e$ on the interval $[\alpha_{\infty}, 1]$ by the extreme value theorem.
Thus, $a(\alpha_k) \leq \e$ for all $k$.
It holds that $\e \leq 0$, where $\e = 0$ iff $a_{\infty} = 0$.
For some fixed $k > 1$ and for $x \in (\alpha_{\infty}, 1)$, consider
\begin{align*}
\frac{\mathrm{d}}{\mathrm{d}x} a_k(x, \alpha_{k-1}) &= - 2\,d\,x\,(k - 1) - 1 \\
&> - 2 \, k - 1\text{.}
\end{align*}
By the mean value theorem, $\frac{a_k(\alpha_{k}, \alpha_{k - 1}) - a_k(\alpha_{k - 1}, \alpha_{k - 1})}{\alpha_{k} - \alpha_{k - 1}} > -2\,k - 1$.
This is equivalent to $\alpha_{k} - \alpha_{k - 1} < \frac{a_k(\alpha_{k-1}, \alpha_{k - 1}) - a_k(\alpha_{k}, \alpha_{k - 1})}{2\,k + 1}$.
Since $a_k(\alpha_{k - 1}, \alpha_{k - 1}) = a(\alpha_{k-1}) \leq \e$ and $a_k(\alpha_{k}, \alpha_{k - 1}) = 0$ by definition, it follows that
\[ \alpha_{k} - \alpha_{k - 1} \leq \frac{\e}{2\,k + 1} < \frac{\e}{3\,k} \text{.}\]
Then,
    \[ \alpha_{\infty} = \alpha_1 + \sum_{k = 2}^{\infty}(\alpha_k - \alpha_{k - 1}) < 1 + \frac{\e}{3} \, \sum_{k=2}^{\infty} \frac{1}{k}\text{.}\]
If $\e < 0$, the harmonic series on the right-hand side makes the term diverge to negative infinity, which is a contradiction.
Thus, $\e = 0$ and $a_{\infty} = 0$.
\end{proof}

\begin{lem}
\label{lem:alphatozerofast}
\[ \alpha_k \in \mathcal{O}\left(\frac{1}{k^2}\right) \]
\end{lem}
\begin{proof}
We will show that, for large enough $k > k'$, $\alpha_{k} < \left(\frac{k - 1}{k}\right)^2 \, \alpha_{k - 1}$.
Then, by induction, $\alpha_{k} \leq \left(\prod_{i=k' + 1}^{k} \left(\frac{i - 1}{i}\right)^2\right) \, \alpha_{k'} = \frac{1}{k^2} \, \left(k'\right)^2 \, \alpha_{k'} \in \mathcal{O}\left(\frac{1}{k^2}\right)$ for all $k \geq k'$.

Consider the value
\begin{align*}
    a_{k}\left(\left(\frac{k - 1}{k}\right)^2 \, \alpha_{k - 1}, \alpha_{k - 1} \right) &= - d \, (k - 1) \, \left(\frac{k - 1}{k}\right)^4 \, \alpha_{k - 1}^2 + d \, k \, \alpha_{k - 1}^2 - \left(\frac{k - 1}{k}\right)^2 \, \alpha_{k - 1} \\
    & = - d \, \alpha_{k - 1}^2 \, \left( \frac{(k - 1)^5 - k^5}{k^4} \right) - \left(\frac{k - 1}{k}\right)^2 \, \alpha_{k - 1} \\
    & = \alpha_{k - 1} \, \left(- d \, \alpha_{k - 1} \, \left( \frac{- 5 \, k^4 + \mathcal{O}\left(k^3\right)}{k^4} \right) - \left(\frac{k - 1}{k}\right)^2 \right)\text{.}
\end{align*}
Note that
\[ \lim_{k \to \infty} - d \, \alpha_{k - 1} \, \left( \frac{- 5 k^4 + \mathcal{O}\left(k^3\right)}{k^4} \right) - \left(\frac{k - 1}{k}\right)^2 = -d \, 0 \, (-5) - 1 = -1 \text{,}\]
where the limit of $\alpha_k$ has been shown in \cref{lem:alphatozero}.
Thus, for sufficiently large $k$, $a_{k}\left(\left(\frac{k - 1}{k}\right)^2 \, \alpha_{k - 1}, \alpha_{k - 1} \right) < 0$.
As mentioned right after the definition of $\alpha_k$, $a_k(0, \alpha_{k - 1}) > 0$ and $a_k(\alpha_{k - 1}, \alpha_{k - 1}) < 0$.
Since $a_k(x, \alpha_{k-1})$ is a quadratic polynomial in $x$, there can be no root in $\left[\left(\frac{k - 1}{k}\right)^2 \, \alpha_{k - 1}, \alpha_{k - 1}\right]$.
Therewith, $\alpha_k < \left(\frac{k - 1}{k}\right)^2 \, \alpha_{k - 1}$, as desired.
\end{proof}

\begin{lem}
\label{lem:3.2}
    There is a fixed integer $k_0$, a function $f$ with some starting value $f(k_0)$ and $f(k) \coloneqq k \cdot f(k-1)^2$ for $k > k_0$ and constants $c_1, c_2 > 0$ such that
    \begin{itemize}
        \item for all $k \geq k_0$, $f(k) > \alpha_k$ and
        \item for all $n \geq 0$,
            \[ \exp \left(-c_1 \, 2^n \right) \leq f(k_0 + n) \leq \exp \left(-c_2 \, 2^n\right)\text{.} \]
    \end{itemize}
\end{lem}
\begin{proof}
Choose $k_0$ such that $\alpha_{k_0} < \frac{1}{2\,e^2\, k_0}$.
We can do so because, by \cref{lem:alphatozerofast}, $\alpha_k$ falls quadratically.
Let $f(k)$ be defined as in the statement of the lemma with 
    \begin{equation}
        \label{eq:fk0}
        f(k_0) \coloneqq \alpha_{k_0} + \frac{1}{2\, e^2\, k_0} < \frac{1}{e^2 \, k_0}\text{.}
    \end{equation}

Since $k \geq 1$, by the definition of $\alpha_k$,
\begin{equation}
    \label{eq:alphaupper}
\alpha_k = d \, (1 - k) \, \alpha_k^2 + k \, d \, \alpha_{k - 1}^2 \leq k \, d \, \alpha_{k - 1}^2 < k \, \alpha_{k - 1}^2 \text{.}
\end{equation}

By construction, $f(k_0) > \alpha_{k_0}$.
If $\alpha_{k - 1} < f(k - 1)$, then $\alpha_{k} < k \, \alpha_{k - 1}^2 < k \, f(k - 1)^2 = f(k)$, where the first inequality is \cref{eq:alphaupper}.
Thus, $(f(k))_{k \geq k_0}$ strictly dominates $(\alpha_k)_{k \geq k_0}$.

We will now show the doubly exponential bounds on $f(k)$.
A simple induction on $n \geq 0$ shows that 
\[f(k_0+n) = f(k_0)^{2^{n}} \, \prod_{i=1}^{n} (k_0+i)^{2^{n-i}}, \]
and taking the logarithm of both sides yields
\[\log(f(k_0+n)) = 2^{n} \, \log(f(k_0)) + 2^{n} \sum_{i=1}^{n} 2^{-i} \log(k_0+i). \]
Therefore, because $\log(k_0 + i) > 0$, we see that
\[\log(f(k_0 + n)) > 2^{n} \, \log(f(k_0)), \]
so setting $c_1 = -\log(f(k_0))$ yields the desired lower bound.

For the upper bound, note that $\log(k_0 + i) < \log(k_0) + i$, which means that
\begin{align*}
    \sum_{i=1}^{n} 2^{-i} \log(k_0+i) & \leq \sum_{i=1}^{\infty} 2^{-i} (\log(k_0) +i) \\
    &= \log(k_0) + \frac{2}{(2-1)^2} \\
    &= \log(k_0) + 2.
\end{align*}
Therefore, we have that
\begin{align*}
    \log(f(k_0+n)) &\leq 2^{n} \, \big(\log(f(k_0)) + \log(k_0) + 2\big) \\
    & = 2^{n} \, \log(f(k_0)\,k_0\,e^2),
\end{align*}
    and because we have $f(k_0) \, k_0 \, e^2 < 1$ by \cref{eq:fk0}, we can let $c_2 = -\log(f(k_0)\,k_0\,e^2)$ to complete the upper bound.
\end{proof}

\subsection{Proof of \cref{lem:3.1}}
\lemthreeone*
\begin{proof}
    \label{pf:3.1}
By induction on $k \geq 1$.
For $k = 1$, $\frac{F_j(1)}{j} = \frac{j}{j} = 1 = \alpha_1$ for all $j$ and the claim follows.

Now let $k > 1$.
Since $\alpha_k < 1$ and since decreasing $\e$ only strengthens our statement, we may assume without loss of generality that
\begin{equation}
\label{eq:epssmall}
\frac{\e{}}{2} < 1 - \alpha_k\text{.}
\end{equation}
Let $\mathcal{A}$ be the event
\[\mathcal{A} \coloneqq \left\{ \forall j,\; \omega_0(m) \leq j \leq m. \; \frac{F_j(k-1)}{j} \leq \alpha_{k-1} + \delta \right\}\text{,}\]
where $\delta > 0$ and $\omega_0$ are fixed values depending on $d$, $k$ and $\e$, which we will give later.
By the induction hypothesis, $\Prob [\mathcal{A}] \to 1$. Therefore, it suffices to show that
\[\Prob \left[ \exists j,\; \omega(m) \leq j \leq m. \; \tfrac{F_j(k)}{j} > \alpha_k + \e \; \middle| \; A \right] \to 0. \]
From here on, we assume that $\mathcal{A}$ holds and show that $\frac{F_j(k)}{j} \leq \alpha_k + \e$ with high probability.

\subsubsection*{Overapproximating $F_j(k)$:}
Let $j$ be such that $\omega_0(m) \leq j \leq m$.
Our goal in this section is to overapproximate $F_j(k)$ as a sum of independent and identically distributed random variables, at least under certain conditions.
We begin by decomposing $F_j(k)$ as a sum of differences $\chi_{i+1} \coloneqq F_{i+1}(k) - F_i(k)$, distributed as in \cref{eq:fdifference}:
\begin{equation}
    \label{eq:fsumchi}
    F_{j}(k) = F_{\omega_0(m)}(k) + \sum_{i=\omega_0(m) + 1}^j \chi_i \leq \omega_0(m) + \sum_{i=\omega_0(m) + 1}^j \chi_i\text{,}
\end{equation}
where the inequality follows from \cref{eq:fullfj}.
By $\mathcal{A}$, it holds that $\frac{F_{i-1}(k-1)}{i-1} \leq \alpha_{k-1} + \delta$ for all $i$ such that $\omega_0(m) < i \leq m$. Thus, for all such $i$
\[ \chi_i \leq \eta_i \coloneqq \begin{cases}
                         1 & \Prob = d \, \left( \frac{F_{i-1}(k)}{i-1} \right)^2 \\
                         k & \Prob = d \, \left( \left( \alpha_{k-1} + \delta \right)^2 - \left( \frac{F_{i-1}(k)}{i-1} \right)^2 \right) \\
                         0 & \text{else}
                         \end{cases}\text{.} \]
By setting $g_i \coloneqq \frac{F_i(k)}{i} - \alpha_k$, we can rewrite the above as
\[ \eta_i = \begin{cases}
                         1 & \Prob = d \, \left( \alpha_k + g_{i-1} \right)^2 \\
                         k & \Prob = d \, \left( \left( \alpha_{k-1} + \delta \right)^2 - \left( \alpha_k + g_{i-1} \right)^2 \right) \\
                         0 & \text{else}
                         \end{cases}\text{.} \]
Choose $\delta > 0$ such that $2 \, \alpha_{k-1} \, \delta + \delta^2 = \frac{\e}{4 \, k \, d}$.
The quadratic equation in $\delta$ must have a positive solution because $2\,\alpha_{k-1}$ and $\frac{\e}{4\,k\,d}$ are positive.
Under the additional assumption that $g_{i-1} \geq 0$, we can overapproximate $\eta_i$ by moving $d (2\,\alpha_k\,g_{i-1} + g_{i-1}^2)$ probability from the first to the second case to obtain:
\[ \eta_i \leq \eta_i' \coloneqq \begin{cases}
                         1 & \Prob = d \, \alpha_k^2 \\
                         k & \Prob = d \, \left( \alpha_{k-1}^2 - \alpha_k^2 + \frac{\e}{4\,k\,d} \right) \\
                         0 & \text{else}
                         \end{cases}\text{.}\]
The $\eta_i'$ are independent and identically distributed. By the definition of $\alpha_k$ in \cref{eq:alphadef}, $\Exp{\eta_i'} = \alpha_k + \frac{\e{}}{4}$.

\subsubsection*{Starting Point $\omega_0(m) \leq \pi \leq \omega(m)$:}
Let $\pi$ be the first $j \geq \omega_0(m)$ such that $\frac{F_j(k)}{j} \leq \alpha_k + \frac{\e{}}{2}$ (write $\pi = \infty$ if no such $j$ exists).
We will use $\pi$ as a starting point for the following analysis, where we show that, with high probability, no $j \geq \pi$ violates our desired property $\frac{F_j(k)}{j} \leq \alpha_k + \e$.
Since we want this to hold for all $j \geq \omega(m)$, we must first show that, with high probability, $\pi \leq \omega(m)$.

Assume that this is not the case, i.e., that $\pi > \omega(m)$.
Then, in particular, $\frac{F_{\omega(m)}}{\omega(m)} > \alpha_k + \frac{\e{}}{2}$.
Furthermore, for all $i$ such that $\omega_0(m) < i \leq \omega(m)$, $g_{i-1} = \frac{F_{i-1}(k)}{i-1} - \alpha_k > \frac{\e{}}{2} > 0$, and therefore $\chi_i \leq \eta_i \leq \eta_i'$.
\begin{align*}
    \Prob [\pi > \omega(m)] & \leq \Prob \left[ \frac{F_{\omega(m)}(k)}{\omega(m)} > \alpha_k + \frac{\e}{2} \right] & \\
    & \leq \Prob \left[ \sum_{i = \omega_0(m) + 1}^{\omega(m)} \chi_i > \left(\alpha_k + \frac{\e}{2}\right) \, \omega(m) - \omega_0(m) \right] & \text{(by \cref{eq:fsumchi})} \\
    & \leq \Prob \left[ \sum_{i = \omega_0(m) + 1}^{\omega(m)} \eta_i' > \alpha_k\,(\omega(m) - \omega_0(m))+\frac{\e}{2} \, \omega(m) - (1 - \alpha_k) \, \omega_0(m) \right] & \\
    \intertext{We choose $\omega_0(m) \coloneqq \frac{\frac{\e}{2}-\frac{\e}{3}}{1 - \alpha_k - \frac{\e}{3}}\,\omega(m)$.
Using \cref{eq:epssmall}, one verifies that the fraction is well-defined and that $0 \leq \omega_0(m) \leq \omega(m)$.
With this definition, it holds that $\frac{\e}{2} \, \omega(m) - (1 - \alpha_k) \, \omega_0(m) = \frac{\e}{3} \, (\omega(m) - \omega_0(m))$.
    Thus, we can rewrite the last inequality as}
    \Prob [\pi > \omega(m)] &\leq \Prob \left[ \sum_{i = \omega_0(m) + 1}^{\omega(m)} \eta_i' > \left(\alpha_k + \frac{\e}{3}\right) \,(\omega(m) - \omega_0(m)) \right]\text{.} &
\end{align*}
The $\eta_i'$ are bounded by $k$, and $\Exp{\sum_{i = \omega_0(m) + 1}^{m} \eta_i'} = \left(\alpha_k + \frac{\e}{4}\right) \, (\omega(m) - \omega_0(m))$ is smaller by a constant factor than $\left(\alpha_k + \frac{\e}{3}\right) \,(\omega(m) - \omega_0(m))$.
Therefore, by Chernoff's bound, the probability decays exponentially in $\omega(m) - \omega_0(m)$.
Since $\omega(m) - \omega_0(m) \in \Theta(\omega(m)) \to \infty$, $\pi \leq \omega(m)$ with high probability.

\subsubsection*{Behavior from $\pi$ on:}
We may now assume $\pi \leq \omega(m)$.
In this section, we bound the probability that $F_j(k)$ surpasses the line $(\alpha_k + \e)\,j$ at some time $j \geq \pi$.

Consider a random walk started at a position $a$ and at some time $t$, whose steps are distributed as $\eta_i'$.
Let $M_a^t(j)$ denote its position at time $j$, i.e., after $j - t$ random steps.
Should the random walk ever drop below the line $\alpha_k\,j$, it is set to $- \infty$ and stops evolving.
Define a function
\[p(t_0) \coloneqq \sup_{t \geq t_0} \Prob \left[\exists j \geq t.\; M_{\left(\alpha_k + \frac{\e}{2}\right)\,t}^t(j) > (\alpha_k + \e)\,j \right] \text{.}\]
Since the random walk only dominates $F_j(k)$ as long as $g_j \geq 0$, i.e., as long as $F_j(k) \geq \alpha_k \, j$, we dissect the evolution of $F_j(k)$ for increasing $j \geq \pi$ into segments.
Set $\rho_0 \coloneqq \pi$.
If, after some $\rho_i$, the process drops below the line $\alpha_k\,j$ and then enters the range $\left(\alpha_k + \frac{\e}{2}\right)\,j - k < F_j(k) \leq \left(\alpha_k + \frac{\e}{2}\right)\,j$ again, call the time of entering $\rho_{i + 1}$.
Else, write $\rho_{i + 1} = \infty$.
Clearly, if $\rho_i < \infty$, $\rho_{i + 1} \geq \rho_i + 1$.
Thus, for all $i$, $\rho_i \geq \pi + i \geq \omega_0(m) + i$.
\begin{align*}
\Prob [\exists j \geq \pi.\; F_j(k) > (\alpha_k + \e) \, j] & \leq \sum_{i=0}^{\infty} \Prob [\exists j, \rho_{i + 1} > j \geq \rho_{i}.\; F_j(k) > (\alpha_k + \e) \, j]\\
\intertext{For crossing the line $(\alpha_k + \e)\,j$ in the time range $[\rho_j, \rho_{j+1})$, the process must get from a position $\leq (\alpha_k + \frac{\e}{2})\,j$ to $> (\alpha_k + \e)\,j$ without dropping below $\alpha_k \, j$ in between. This event is stochastically dominated by the event that our random walk, when started from the potentially higher position $(\alpha_k + \frac{\e}{2})\,\rho_i$ at time $\rho_i$, will ever cross $(\alpha_k + \e)\,j$:}
\Prob [\exists j \geq \pi.\; F_j(k) > (\alpha_k + \e) \, j] & \leq \sum_{i=0}^{\infty} p\left(\rho_i\right) \\
\intertext{Because of the supremum in its definition, $p(t_0)$ is monotonically decreasing, and we can replace $\rho_i$ by its lower bound:}
\Prob [\exists j \geq \pi.\; F_j(k) > (\alpha_k + \e) \, j]& \leq \sum_{i=0}^{\infty} p\left(\omega_0(m) + i\right)\text{,} \\
\intertext{which, as we will show in \cref{lem:ptozero}, is}
& \leq \sum_{i=\omega_0(m) + 1}^{\infty} s \, e^{-r \, i}
\end{align*}
for positive constants $r$ and $s$.
This series converges because $\sum_{i=0}^{\infty} s \, e^{-r \, i} = s \, \frac{e^r}{e^r - 1}$. With increasing $m$, the lower bound of the sum increases and the term goes to $0$, proving our claim.
\end{proof}

\begin{lem}
\label{lem:ptozero}
There are positive constants $r$ and $s$, depending only on $\e$ and $k$, such that, for all $t_0$,
\[ p(t_0) \leq s \, e^{-r\,t_0}\text{.}\]
\end{lem}
\begin{proof}
Let $M'(n)$ be a random walk with step distribution $\eta_i'$ that--in contrast to $M_a^t(j)$--starts at time and position 0 and that does not have a stopping condition.
$M_a^t(j)$ is dominated by $a + M'(j - t)$ for all $a$, $t$ and $j$ where $j \geq t$.
The expectation of $M'(n)$ equals $n \, \Exp{\eta_i'} = n \, \left(\alpha_k + \frac{\e}{4}\right)$.
For any fixed $t$,
\begin{align*}
    \mathcal{E}_t \coloneqq{}& \left\{ \exists j \geq t.\; M_{\left(\alpha_k + \frac{\e{}}{2}\right)\,t}^t(j) > (\alpha_k + \e{})\,j \right\} \\
    \subseteq &{} \left\{ \exists n \geq 0. \; M'(n) > \left(\alpha_k + \e{}
    \right)\,n + \frac{\e{}}{2}\,t \right\} \\
    \subseteq &{} \left\{ \exists n \geq 0. \; M'(n) - \Exp{M'(n)} > \frac{3}{4}\,\e{}\,n + \frac{\e{}}{2}\,t \right\} \text{.}
\end{align*}
Call the last event $\mathcal{E}_t'$.
By Hoeffding's inequality, for any $n \geq 0$,
\begin{align*}
\Prob\left[M'(n) - \Exp{M'(n)} > \frac{3}{4}\,\e{}\,n + \frac{\e{}}{2}\,t\right] &\leq \exp \left( - \frac{2 \, \left(\frac{3}{4}\,\e\,n + \frac{\e}{2}\,t \right)^2}{n \, k^2} \right)\text{.} \\
\intertext{Set $r \coloneqq \frac{3\,\e^2}{2\,k^2}$ to obtain}
\Prob\left[M'(n) - \Exp{M'(n)} > \frac{3}{4}\,\e{}\,n + \frac{\e{}}{2}\,t\right] &\leq \exp \left( - r \, \left( \frac{3}{4}\,n + t + \frac{t^2}{3\,n} \right)  \right) \\
&\leq \exp \left( - r \, \left( \frac{3}{4}\,n + t \right) \right) \text{.}
\end{align*}
Thus,
\begin{align*}
    \Prob[\mathcal{E}_t'] &\leq \sum_{n=0}^{\infty} \exp \left( - r \, \left( \frac{3}{4}\,n + t \right) \right) = e^{- r \, t} \sum_{n=0}^{\infty} \exp \left( -\frac{3}{4} \, r \, n\right) = e^{- r \, t} \, \frac{e^{\frac{3}{4}\,r}}{e^{\frac{3}{4}\,r} - 1}\text{.}
\end{align*}
By setting $s \coloneqq \frac{e^{\frac{3}{4}\,r}}{e^{\frac{3}{4}\,r} - 1}$, have $\Prob[\mathcal{E}_t'] \leq e^{- r \, t}\,s$.
Since this bound decreases monotonically in $t$,
\[ p(t_0) = \sup_{t \geq t_0} \Prob[\mathcal{E}_t] \leq \sup_{t \geq t_0} e^{- r \, t}\,s = e^{-r \, t_0}\,s \text{.}\]
\end{proof}

\subsection{Proof of \cref{lem:3.3}}
\lemthreethree*
\begin{proof}~
    \label{pf:3.3}

\subsubsection*{Definitions}
Define the event
\[ \mathcal{G}_k \coloneqq \left\{\forall j, \phi(m,k) \leq j \leq m. \; \frac{F_j(k)}{j} \leq f(k) \right\}. \]
As in \cref{lem:3.1}, we dissect $F_j(k +1)$ for any $\phi(m,k+1) \leq j \leq m$ as
\[ \frac{F_j(k+1)}{j} = \frac{F_{\phi(m,k)}(k+1)}{j} + \frac{1}{j} \sum_{i=\phi(m,k)+1}^j \xi_i (k+1)\text{,} \]
where $\xi_{i + 1}(k+1) \coloneqq F_{i+1}(k+1) - F_i(k+1)$ as in \cref{eq:fdifference}. We now define a random variable $X_{j,k + 1}$ to bound the sum of $\xi_i$ terms from above.
Let $X_{j,k+1}$ be distributed as
\[ X_{j,k + 1} \sim (k+1) \cdot \text{Binom}(j - \phi(m,k), d \, f(k)^2)\text{,} \]
and note that, on $\mathcal{G}_k$, $\sum_{i=\phi(m,k) + 1}^j \xi_i (k+1)$ is stochastically dominated by $X_{j,k + 1}$.
This is because the $(j - \phi(m, k))$ many $\xi_i$ are independent, bounded by $(k - 1)$ and have non-zero value with probability $d \, \left(\frac{F_{i-1}(k)}{i - 1}\right)^2 = d \, f(k)^2$.

Now, consider the event 
\[ \mathcal{E}(k+1) = \left\{\exists j, \phi(m, k+1) \leq j \leq m. \; \sum_{i=\phi(m,k)+1}^j \xi_i(k+1) > \gamma \, \Exp{X_{j,k+1}} \right\}, \]
where $\gamma \coloneqq \frac{d+1}{2 \, d} > 1$ is a constant chosen such that $d \, \gamma = \frac{d+1}{2} < 1$.

\subsubsection*{$\mathcal{E}(k + 1)$ Is Unlikely}
We  bound $\Prob[\mathcal{E}(k+1) \cap \mathcal{G}_k]$ using standard Chernoff bounds as follows:
\begin{align*}
    \Prob[\mathcal{E}(k+1) \cap \mathcal{G}_k] &\leq \Prob[ \exists j, \phi(m,k+1) \leq j \leq m. \; X_{j,k + 1} > \gamma \, \Exp{X_{j,k + 1}}] \\
     &\leq \Prob[ \exists j, \phi(m,k+1) \leq j \leq m. \; \frac{X_{j,k + 1}}{k + 1} > \gamma \, \Exp{\frac{X_{j,k + 1}}{k + 1}}] \\
    & \leq \sum_{j=\phi(m,k+1)}^{m} \exp \left(-\frac{(\gamma-1)^2 \, \Exp{\frac{X_{j,k}}{k + 1}}}{3}\right) \\
    & \leq \sum_{j=\phi(m,k+1)}^{m} \exp \left(- \frac{(\gamma-1)^2}{3} \, (j - \phi(m,k)) \, d \, f(k)^2\right) \\
    \intertext{Set $c \coloneqq \frac{(\gamma-1)^2 \, d}{3}$. From $\gamma > 1$, it follows that $c$ is positive. We then have}
    \Prob[\mathcal{E}(k+1) \cap \mathcal{G}_k] &\leq \sum_{j=\phi(m,k + 1)}^{\infty} \exp \left( - c \, f(k)^2 \, (j - \phi(m,k)) \right) \\
    & \leq \sum_{\ell=0}^{\infty} \exp \left(-c \, f(k)^2 \, (\ell + \phi(m,k+1) - \phi(m,k)) \right) & \text{($\ell \coloneqq j - \phi(m,k+1)$)}\\
    & \leq e^{-c \, f(k)^2 \, (\phi(m,k+1) - \phi(m,k))} \sum_{\ell=0}^{\infty} e^{-c \, f(k)^2 \, \ell} \\
    & \leq \frac{1}{1 - e^{-c \, f(k)^2}} \, e^{-c \, f(k)^2 \, (\phi(m,k+1) - \phi(m,k))}. & \text{(geometric series)} \\
    \intertext{
                Furthermore, since the sequence of $f(k)$ converges to $0$, we can find a constant $C'$ (independent of $k$) such that
        \[ C'\geq \frac{f(k)^2}{1 - e^{-c \, f(k)^2}}. \]
        Indeed, note that if we let $x \coloneqq f(k)^2$, then it suffices to show that $\lim_{x \rightarrow 0} \frac{x^2}{1 - e^{-c \, x^2}}$ is constant. 
        Because both the numerator and denominator go to $0$ in the limit as $x$ approaches $0$, we can apply L'H\^{o}pital's rule to obtain
        \begin{equation*}
            \lim_{x \rightarrow 0} \frac{x}{1 - e^{-c \, x}} = \lim_{x \rightarrow 0} \frac{1}{c \, e^{-c \, x}} = \frac{1}{c}.
        \end{equation*}
        Plugging this in yields 
    }
    \Prob[\mathcal{E}(k+1) \cap \mathcal{G}_k] &\leq \frac{C'}{f(k)^2} \left( e^{-c \, f(k)^2 \, (\phi(m,k+1) - \phi(m,k))} \right). & \\
    \intertext{Now, expanding the definition of $\phi$, and applying \cref{lem:3.2} to bound $f(k)^2$, we have}
    \Prob[\mathcal{E}(k+1) \cap \mathcal{G}_k] &\leq C' \left( \frac{e^{-c \, \rho(m)\, \left(C^{2^{k+2}} - C^{2^{k+1}} \right) \exp{(-c_1 \, 2^{k - k_0 + 1})}}}{\exp{(-c_1\,  2^{k - k_0 + 1})}} \right) &\text{(by \cref{lem:3.2})}\\
    & \leq C' \left( e^{-c \, \rho(m)\, \left(C^{2^{k+2}} - C^{2^{k+1}} \right) \exp{(-c_1 \, 2^{k - k_0 + 1})} + c_1\,  2^{k - k_0 + 1}} \right) \\
    & \leq C' \left( e^{-c \, \rho(m)\, C^{2^{k+1}} \left(C^{2^{k+1}} - 1 \right) \exp{(-c_1 \, 2^{k})} + c_1\,  2^{k - k_0 + 1}} \right)\text{.} \\
    \intertext{Because $\left(C^{2^{k+1}} - 1 \right) > 1$ and $C^{2^{k+1}} > C^{2^{k}}$ for $k > 0$, we have}
    \Prob[\mathcal{E}(k+1) \cap \mathcal{G}_k] &\leq C' \left( e^{-c \, \rho(m)\, C^{2^k} \exp{(-c_1 \, 2^k)}+ c_1\,  2^{k - k_0 + 1}} \right).
    \stepcounter{equation}\tag{\theequation}\label{eqn:Ekplusone}
\end{align*}

\subsubsection*{$\mathcal{G}_k$ and $\overline{\mathcal{E}(k+1)}$ Together Imply $\mathcal{G}_{k + 1}$}

Given both $\mathcal{G}_{k}$ and the complement of $\mathcal{E}(k+1)$, we know that 
\begin{align*}
    \frac{F_j(k+1)}{j} &= \frac{F_{\phi(m,k)}(k+1)}{j} + \frac{1}{j} \sum_{i = \phi(m,k)+1}^j \xi_i(k+1) \\
    & \leq \frac{\phi(m,k)}{j} + \frac{1}{j} \left( \gamma \, \Exp{X_{j,k+1}} \right) & \text{(by \cref{eq:fullfj}, $\overline{\mathcal{E}(k + 1)}$)} \\
    & \leq \frac{\phi(m,k)}{\phi(m,k+1)} + \frac{\frac{d+1}{2d}}{j} \, (k+1) \, (j - \phi(m,k))\, d\, f(k)^2 & \text{($j \geq \phi(m, k + 1)$, $\Exp{X_{j,k+1}}$)}\\
    & \leq C^{2^{k+1} - 2^{k+2}} + \left( \frac{d+1}{2} \right) \, (k+1) \, f(k)^2 & \text{(def. of $\phi$, $j - \phi(m,k) \leq j$).}\\
    \intertext{Since $2^{k + 1} - 2^{k + 2} = - 2^{k + 1}$, and since, by \cref{eq:defcapc}, $C > \frac{2}{1-d} \, e^{c_1 \, 2^{-1}} \geq \frac{2}{1-d} \, e^{c_1 \, 2^{-k_0}}$,}
    \frac{F_j(k+1)}{j} & \leq \left( \frac{2}{1-d} \, e^{c_1 \, 2^{-k_0}} \right)^{-2^{k+1}} + \frac{d+1}{2} \, (k+1) \, f(k)^2 & \text{(\cref{eq:defcapc})}\\
    & \leq \left( \frac{2}{1-d} \right)^{-2^{k+1}} e^{-c_1 \, 2^{-k_0} \, 2^{k+1}} + \frac{d+1}{2} \, (k+1) \, f(k)^2 \\
    & \leq \left(\frac{1-d}{2} \right) e^{-c_1 2^{k+1-k_0}} + \frac{d+1}{2} \, (k+1) \, f(k)^2 \\
    & \leq \left(\frac{1-d}{2} \right) f(k+1) + \frac{d+1}{2} \, f(k+1) & \text{(by \cref{lem:3.2})}\\
    & = f(k+1).
\end{align*}

\subsubsection*{Combining the Previous Steps}
In the previous step, we established $\mathcal{G}_k \cap \overline{\mathcal{E}(k+1)} \subseteq \mathcal{G}_{k+1}$.
This implies $\mathcal{G}_k \cap \overline{\mathcal{G}_{k + 1}} \subseteq \mathcal{E}(k + 1) \cap \mathcal{G}_k$.
Conceptually, this means that if $\mathcal{G}_k$ happens but not $\mathcal{G}_{k + 1}$, this can be blamed on the unlikely event $\mathcal{E}(k + 1)$.
We show the implication:
By taking the complement of both sides, we have $\overline{\mathcal{G}_{k+1}} \subseteq \overline{\mathcal{G}_k} \cup \mathcal{E}(k+1)$.
By intersecting of both sides with $\mathcal{G}_k$, we obtain $\overline{\mathcal{G}_{k+1}} \cap \mathcal{G}_k \subseteq (\overline{\mathcal{G}_k} \cup \mathcal{E}(k+1)) \cap \mathcal{G}_k$. From here, note that $\overline{\mathcal{G}_k} \cap \mathcal{G}_k = \varnothing$ and therefore we have $\overline{\mathcal{G}_{k+1}} \cap \mathcal{G}_k \subseteq \mathcal{E}(k+1) \cap \mathcal{G}_k$.

We are interested in the probability that $\mathcal{G}_k$ fails for some $k_0 \leq k \leq k_*(m)$, which we can upper-bound as
\begin{align*}
    \Prob\left[\exists k,  k_0 \leq k \leq k_*(m). \; \overline{\mathcal{G}_k}\right] &\leq \Prob\left[\overline{\mathcal{G}_{k_0}}\right] + \sum_{k=k_0}^{k_*(m) - 1} \Prob\left[\overline{\mathcal{G}_{k+1}} \cap \mathcal{G}_k\right] \\
     &\leq \Prob\left[\overline{\mathcal{G}_{k_0}}\right] + \sum_{k=k_0}^{k_*(m) - 1} \Prob[\mathcal{E}(k+1) \cap \mathcal{G}_k] \\
    \intertext{by the above. By \cref{eqn:Ekplusone},}
    \Prob\left[\exists k, k_0 \leq k \leq k_*(m). \; \overline{\mathcal{G}_k}\right] & \leq \Prob\left[\overline{\mathcal{G}_{k_0}}\right] + \sum_{k={k_0}}^{k_*(m) - 1} C' \left( e^{-c \, \rho(m)\, C^{2^k} \exp{(-c_1 \, 2^k)}+ c_1\,  2^{k-k_0+1}} \right) \\
    & = \Prob\left[\overline{\mathcal{G}_{k_0}}\right] + \sum_{k={k_0}}^{k_*(m) - 1} C' \left( e^{-c \, \rho(m)\, (C/e^{c_1})^{2^k}+ c_1\,  2^{k-k_0+1}} \right). \\
\intertext{
    For $m$ sufficiently large, it is easy to verify that the term inside the sum is monotonically decreasing in $k$. Indeed, note that for large $m$, the exponent is dominated by $-c' \, \rho(m)\, (C/e^{c_1})^{2^k}$ because $c_1\,  2^{k-k_0+1}$ has no dependence on $m$, and by \cref{eq:defcapc}, $C/e^{c_1} > 1$. Therefore, for some constant $A \approx c \, (C\, / \, e^{c_1})^{2^{k_0}}$, we have
}
    \Prob\left[\exists k_0 \leq k \leq k_*(m). \; \overline{\mathcal{G}_k}\right] &\leq \Prob\left[\overline{\mathcal{G}_{k_0}}\right] + (k_*(m) - k_0) \, e^{-A \, \rho(m)} \\
    & \leq \Prob\left[\overline{\mathcal{G}_{k_0}}\right] + k_*(m) \, e^{-A \, \rho(m)}.
\end{align*}
The left summand converges to zero as discussed in the proof sketch.
Furthermore, because $k_*(m) = \mathcal{O}(\log\log m)$, the right summand tends to $0$ with $m \to \infty$.
Indeed, note that
\begin{align*}
    k_*(m) \, e^{-A \, \rho(m)} &= e^{\log k_*(m)} \, e^{-A \, (\log \log m)^{1/3}} \\
    & = e^{\log k_*(m) - A (\log \log m)^{1/3}} \\
    & \in e^{\mathcal{O}\left(\log \log \log m\right) - \Omega\left((\log \log m)^{1/3}\right)} .
\end{align*}
For large enough $m$, the exponent diverges to negative infinity. As a result, $k_*(m) \, e^{-A \, \rho(m)} \to 0$ and thus $\Prob[\exists k_0 \leq k \leq k_*(m). \; \overline{\mathcal{G}_k}] \to 0$.
\end{proof}

\subsection{Proof of \cref{lem:3.6}}
\begin{lem}
\label{lem:beta0}
There exists a constant $\beta_0 > 0$ such that, with high probability, 
\[ \forall j, \log \log m \leq j \leq m. \; \frac{F_j(k_*(m))}{j} \leq j^{- \beta_0}. \] 
\end{lem}
\begin{proof}
First, let $k_0 \leq k \leq k_*(m)$ and $\phi(m,k) \leq j' \leq m$. Putting together the previous results yields
\begin{align*}
    \frac{F_{j'}(k_*(m))}{j'} & \leq \frac{F_{j'}(k)}{j'} & \text{($F_{j'}(k)$ is monotone in $k$)} \\
    & \leq f(k) & \text{(\cref{lem:3.3})} \\
    & \leq \exp \left(-c_2 \, 2^{k-k_0} \right) & \text{(\cref{lem:3.2})} \\
    \intertext{with high probability. By arithmetic,}
    \frac{F_{j'}(k_*(m))}{j'} &= \left(C^{1/\log C}\right)^{-c_2 \, 2^{k - k_0}} \\
    &= \left(C^{1/\log C}\right)^{-c_2 \, \frac{2^{k + 2}}{2^{k_0 + 2}}} \\
    &= \left(C^{-2^{k + 2}}\right)^{\frac{c_2}{2^{k_0 + 2} \, \log C}} \\
    &= \left( \frac{1}{C^{2^{k+2}}} \right)^\beta \\
    &= \left(\frac{\rho(m)}{\phi(m,k+1)} \right)^\beta,
\end{align*}
where $\beta \coloneqq \frac{c_2}{2^{k_0 + 2} \, \log C} > 0$ is a constant independent of $j$, $k$ and $m$.
From the strong formulation in \cref{lem:3.3}, it follows that this holds with high probability for all such $j$ and $k$ simultaneously.

Now, consider some $j$ in the range $\log \log m \leq j \leq m$ mentioned in the lemma.
We can find some $k(j)$ between $k_0$ and $k_*(m)$ such that \begin{equation}
    \label{eq:jphi}
    \phi(m,k(j)) \leq j \leq \phi(m,k(j) + 1).
\end{equation}
In order to show that we can find such a $k(j)$, by the monotonicity of $\phi(m, k)$ in $k$, it suffices to show that $\phi(m,k_*(m) + 1) \geq m$ and $\phi(m,k_0) \leq \log \log m$.
Because $\phi(m,k_0) = (\log \log m)^{1/3} \, C^{2^{k_0+1}}$, and because $C^{2^{k_0+1}}$ is a constant, for large enough $m$, $\phi(m, k_0)$ is indeed less than $\log \log m$.
As for the upper bound, we have
\begin{align*}
    \phi(m,k_*(m) + 1) &= (\log \log m)^{1/3} \, C^{2^{k_*(m)+2}} \\
    & = (\log \log m)^{1/3} \, C^{2 \cdot 2^{k_*(m)+1}} \\
    & = (\log \log m)^{1/3} \, \left(C^{2^{k_*(m)+1}}\right)^2 \\
    & \geq (\log \log m)^{1/3} \, \left(\sqrt{m}\right)^2 & \text{(by \cref{eqn:kstar})}\\
    & = (\log \log m)^{1/3} \, m \\ 
    & \geq m. & \text{(for large $m$)}
\end{align*}

Because of \cref{eq:jphi}, we can apply the results from the first section to get
\begin{align*}
\frac{F_{j}(k_*(m))}{j} &\leq \left(\frac{\rho(m)}{\phi(m,k(j) + 1)} \right)^\beta \\
&\leq \left( \frac{\rho(m)}{j} \right)^\beta \\
&= \left( \frac{(\log \log m)^{1/3}}{j} \right)^\beta \\
&\leq \left( \frac{j^{1/3}}{j} \right)^\beta \\
&= j^{- \frac{2}{3} \, \beta} \\
& = j^{- \beta_0},
\end{align*}
where $\beta_0 \coloneqq \frac{2}{3} \beta$.
\end{proof}

\begin{lem}
\label{lem:3.5}
Assume that there exist constants $M_1 > 0$, $0 < \beta < \frac{1}{2}$ and a function $k(m) \in \log_2 \log m + \Theta(1)$ such that, with high probability,
\[ \forall j, (\log \log m)^{M_1} \leq j \leq m. \; \frac{F_j(k(m))}{j} \leq j^{- \beta}\text{.} \]
Then, there is an $M_2 > 0$ such that, with high probability,
\[ \forall j, (\log \log m)^{M_2} \leq j \leq m. \; \frac{F_j(k(m) + 1)}{j} \leq j^{- \frac{3}{2}\,\beta}\text{.} \]
\end{lem}
\begin{proof}
Let $\mathcal{C}$ be the event in the hypothesis, and assume that it holds in the following.
Moreover, set $j_0(m) \coloneqq (\log \log m)^{M_1}$.

Let $j$ be such that $j_0(m) \leq j \leq m$.
As in \cref{lem:3.1}, write $F_j (k(m) + 1) = F_{j_0(m)}(k(m) + 1) + \sum_{i = j_0(m) + 1}^{j} \chi_i$, where $\chi_i = F_{i}(k(m) + 1) - F_{i - 1}(k(m) + 1)$.
On $\mathcal{C}$, $\frac{\chi_i}{k(m) + 1}$ is dominated by a Bernoulli variable with mean $d \, (i - 1)^{-2 \, \beta}$ by the recurrence in \cref{eq:fdifference}. %
The Bernoulli variable in turn is dominated by a Poisson variable with mean $\lambda = - \log \left(1 - d \, (i - 1)^{-2 \, \beta}\right)$ since its probability of being $0$ is $e^{- \lambda} = 1 - d \, (i - 1)^{-2 \, \beta}$.
All these random variables are independent.

For a general Bernoulli variable with mean $0 < p < 1$ and its dominating Poisson variable with mean $- \log(1-p)$, look at the ratio of these means $- \frac{\log(1 - p)}{p}$.
Its derivative with respect to $p$ is $\frac{\log (1 - p)}{p^2} + \frac{1}{p \, (1 - p)}$, which is positive for all $p$.
Thus, the ratio must increase monotonically in $p$, and for $p \coloneqq d \, (i - 1)^{-2 \, \beta}$, the ratio must decrease monotonically in $i$.
    Thus, if we set $c \coloneqq d \, \left(- \frac{\log(1 - d\,1^{-2 \, \beta})}{d \, 1^{-2 \, \beta}}\right) = - \log (1 - d) > 0$, the mean of the Poisson variable corresponding to $\frac{\chi_i}{k(m) + 1}$ can be overapproximated by $c \, (i - 1)^{-2 \, \beta}$ for large enough $m$.

Consider $\Delta_j \coloneqq \sum_{i = j_0(m) + 1}^{j} \chi_i$.
Since the sum of independent Poisson variables is a Poisson variable with the sum of the means as its mean, we can dominate $\frac{\Delta_j}{k(m) + 1}$ by a Poisson variable $X_j$ with mean
\begin{align*}
    \Exp{X_j} &\coloneqq \frac{c}{1 - 2 \, \beta} \, j^{1 - 2 \beta}& \\
    & \geq \left. \frac{c}{1 - 2 \, \beta} \, i^{1 - 2 \beta} \right|_{j_0(m) - 1}^{j-1} & \text{(for large enough $m$)} \\
    & = c \, \int_{j_0(m) - 1}^{j-1} i^{-2 \, \beta} \, \mathrm{d}i & \\
    & \geq c \, \sum_{i = j_0(m)}^{j - 1} i^{-2 \, \beta} & \text{(by comparison with Riemann sum)} \\
    & = \sum_{i = j_0(m) + 1}^{j} c \, (i - 1)^{-2 \, \beta}\text{.}
\end{align*}
By the tail bound described in \cite{canonne}, for any $C' > 1$,
\begin{align*}
    \Prob[X_j \geq C' \, \Exp{X_j}] &\leq \exp \left( - 2\, \frac{\left((C' - 1)\,\Exp{X_j}\right)^2}{2 \, \Exp{X_j}} \, \frac{C'\,\log C' - C' + 1}{(C' - 1)^2} \right) \\
    &= \exp \left(- (C' \, \log C' - C' + 1) \, \Exp{X_j} \right) \\
    &= \exp \left(- \frac{c}{1 - 2 \, \beta}\,(C' \, \log C' - C' + 1) \, j^{1 - 2 \, \beta} \right) \\
    \intertext{We can fix a sufficiently large constant $C'$, dependent on $d$ and $\beta$ but independent from $m$, such that $C' \leq \frac{c}{1 - 2 \, \beta} \, (C' \log C' - C' + 1)$ and thus}
    \Prob[X_j \geq C' \, \Exp{X_j}] &\leq \exp \left(- C' \, j^{1 - 2 \, \beta} \right)
\end{align*}
for all $j$ between $j_0(m)$ and $m$.
Therefore,
\[\Prob[\exists j, j_0(m) \leq j \leq m.\;\Delta_j \geq (k(m) + 1) \, C' \, \Exp{X_j}] \leq \sum_{j=j_0(m)}^{\infty} \exp\left( - C' \, j^{1 - 2 \, \beta} \right)\text{.} \]
Since $\beta < \frac{1}{2}$ by assumption, the terms $\exp (- C'\,j^{1 - 2 \beta})$ fall faster than those of the sequence $\sum \frac{1}{j^2}$.
By the direct comparison test, the series converges, and the probability goes to $0$ for $m \to \infty$. 
With high probability, it must hold that
\begin{align*}
    F_j(k(m)+1) &\leq F_{j_0(m)}(k(m)+1) + (k(m) + 1) \, C' \, \frac{c}{1 - 2 \, \beta} \, j^{1 - 2 \, \beta} \\
    &\leq j_0(m) + (k(m) + 1) \, C' \, \frac{c}{1 - 2 \, \beta} \, j^{1 - 2 \, \beta} \text{.} \\
    \intertext{
As $k(m) = \log_2 \log m + \Theta(1)$, $k(m) + 1$ is also in $\log_2 \log m + \Theta(1)$.
We can choose $M_2$ large enough such that $j_0(m) \leq \frac{1}{2} \, (\log \log m)^{M_2 \, (1 - \frac{3}{2}\,\beta)}$ and such that $(k(m) + 1) \, C' \, \frac{c}{1 - 2 \, \beta} \leq \frac{1}{2} \, (\log \log m)^{\frac{\beta}{2} \, M_2}$ for sufficiently large $m$.
Then, for all $j$ between $(\log \log m)^{M_2}$ and $m$,
    }
    F_j(k(m) + 1) &\leq \frac{1}{2} \, j^{1 - \frac{3}{2}\,\beta} + \frac{1}{2} \, j^{\frac{\beta}{2}} \, j^{1 - 2\,\beta} = j^{1 - \frac{3}{2}\,\beta}\text{.}
\end{align*}
\end{proof}

\lemthreesix*

\begin{proof}
    \label{pf:3.6}
We repeatedly strengthen \cref{lem:beta0} using \cref{lem:3.5}, increasing $\beta$ in every step until $\beta > \frac{1}{2}$.
Note that \cref{lem:3.5} does not apply to $\beta$ exactly equal to $\frac{1}{2}$.
In this case, we weaken our hypthesis by slightly decreasing $\beta$ such that we obtain $\beta > \frac{1}{2}$ in the next step.

After increasing $\beta$ $r'$ many times, we obtain $\beta > \frac{1}{2}$ and $M$ such that the event
\[\mathcal{C} \coloneqq \forall j, j_0(m) \leq j \leq m. \; F_j(k_*(m) + r') \leq j^{1-\beta}\]
holds with high probability. For $r \coloneqq r' + 1$, this shows \cref{eq:3.6add}.

Assuming $\mathcal{C}$, we can bound
\begin{align*}
    \Prob[\exists j, j_0(m) \leq j \leq m. \; &F_j(k_*(m)+r'+1) > F_{j_0(m)}(k_*(m)+r'+1)] \\
    & \leq \sum_{i=j_0(m)}^{m - 1} d \, \left( \frac{F_i(k_*(m) + r')}{i} \right)^2 & \text{(\cref{eq:fdifference})}\\
    & \leq d \, \sum_{i=j_0(m)}^{m - 1} \left( \frac{i^{1-\beta}}{i} \right)^2 \\
    & = d \, \sum_{i=j_0(m)}^m i^{-2\beta} \\
    & = o(1)
\end{align*}
because $\beta > 1/2$ and therefore the series is over-harmonic.
\end{proof}

\section{MILP Formulation for Minimizing Congestion}
\label{app:milp}
	
Congestion minimization for confluent flow can be expressed as a mixed integer linear program (MILP).

To stress the connection to \textsc{MinMaxWeight}, denote the congestion at a voter $i$ by $w(i)$.
For each potential delegation $(u,v)$, $f(u,v)$ gives the amount of flow between $u$ and $v$.
This flow must be nonnegative (\ref{eq:milpnonneg}) and satisfy flow conservation (\ref{eq:milpflow}).
Congestion is defined in \cref{eq:milpw}.
To minimize maximum congestion, we introduce a variable $z$ that is higher than the congestion of any voter (\ref{eq:milpz}), and minimize $z$ (\ref{eq:milpmin}).

So far, we have described a Linear Program for optimizing splittable flow.
To restrict the solutions to confluent flow, we must enforce an `all-or-nothing' constraint on outflow from any node, i.e. at most one outgoing edge per node can have positive flow.
We express this using a convex-hull reformulation.
We introduce a binary variable $x_{u,v}$ for each edge (\ref{eq:milpbin}), and set the sum of binary variables for all outgoing edges of a node to $1$ (\ref{eq:milpbinone}).
If $M$ is a constant larger than the maximum possible flow, we can then bound $f(u, v) \leq M \, x_{u, v}$ (\ref{eq:milpconf}) to have at most one positive outflow per node.

The final MILP is thus

\begin{align}
    &\text{minimize}  && z && \label{eq:milpmin} \\
    &\text{subject to}
                         && f(m,n) \geq 0 \qquad && \forall (m,n) \in E, \label{eq:milpnonneg} \\
                         &&& \displaystyle \sum_{(n,m) \in E} f(n, m) = 1 + \displaystyle \sum_{(m,n) \in E} f(m,n) \qquad && \forall n \in N \setminus V, \label{eq:milpflow} \\
                         &&& \displaystyle w(v) = 1 + \sum \limits_{(n,v) \in E} f(n,v) \qquad && \forall v \in V, \label{eq:milpw} \\
                         &&& z \geq w(v) \qquad && \forall v \in V, \label{eq:milpz} \\
                         &&& x_{m,n} \in \{0,1\} \qquad && \forall (m,n) \in E, \label{eq:milpbin} \\
                         &&& \displaystyle \sum_{(n,m) \in E} x_{n, m} = 1 \qquad && \forall n \in N \setminus V, \label{eq:milpbinone} \\
                         &&& f(m,n) \leq M \cdot x_{m,n} \qquad && \forall (m,n) \in E. \label{eq:milpconf}
\end{align}

\section{Additional Figures}
\label{app:figs}

 \subsection{Single vs.\ Double Delegation}
 \label{sec:appsinglevsdoublefigures}

\begin{figure}[H]
\centering
\begin{subfigure}[h]{.49\textwidth}
\includegraphics[width=\textwidth]{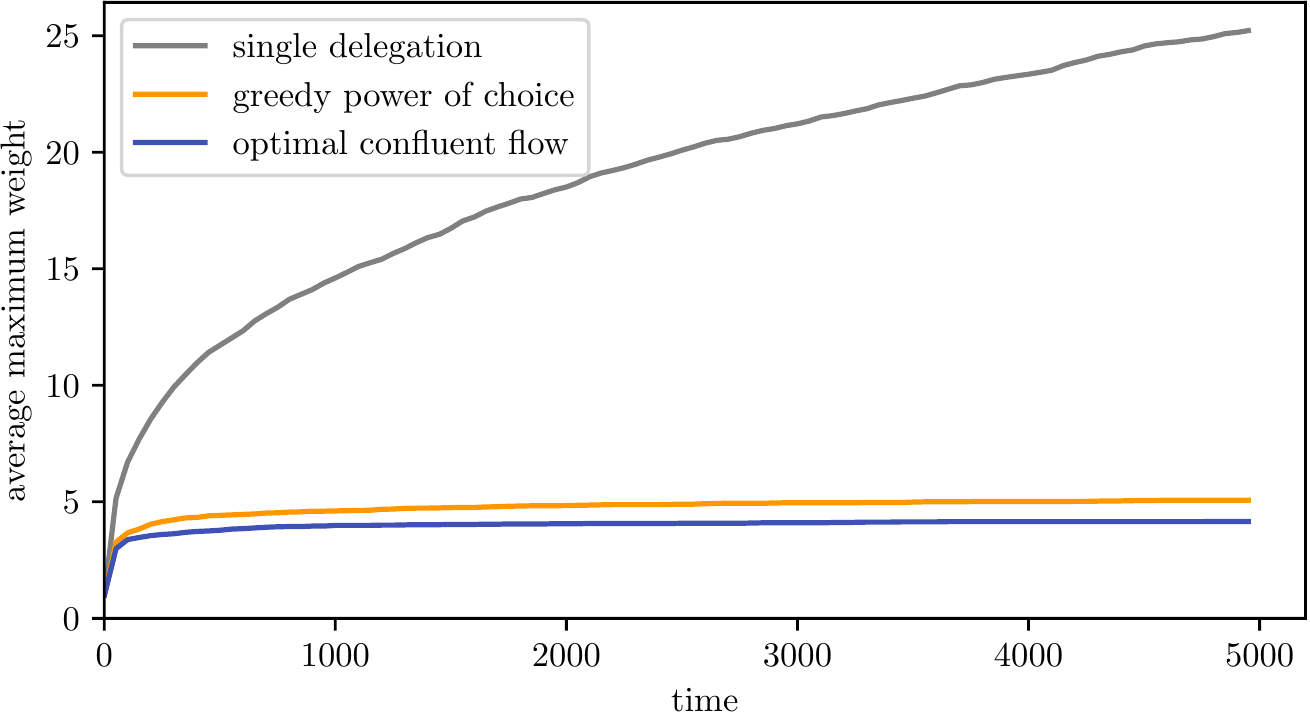}
\caption{$\gamma = 0.5$, $d = 0.25$}
\end{subfigure}\\
    \begin{subfigure}[h]{.49\textwidth}
\includegraphics[width=\textwidth]{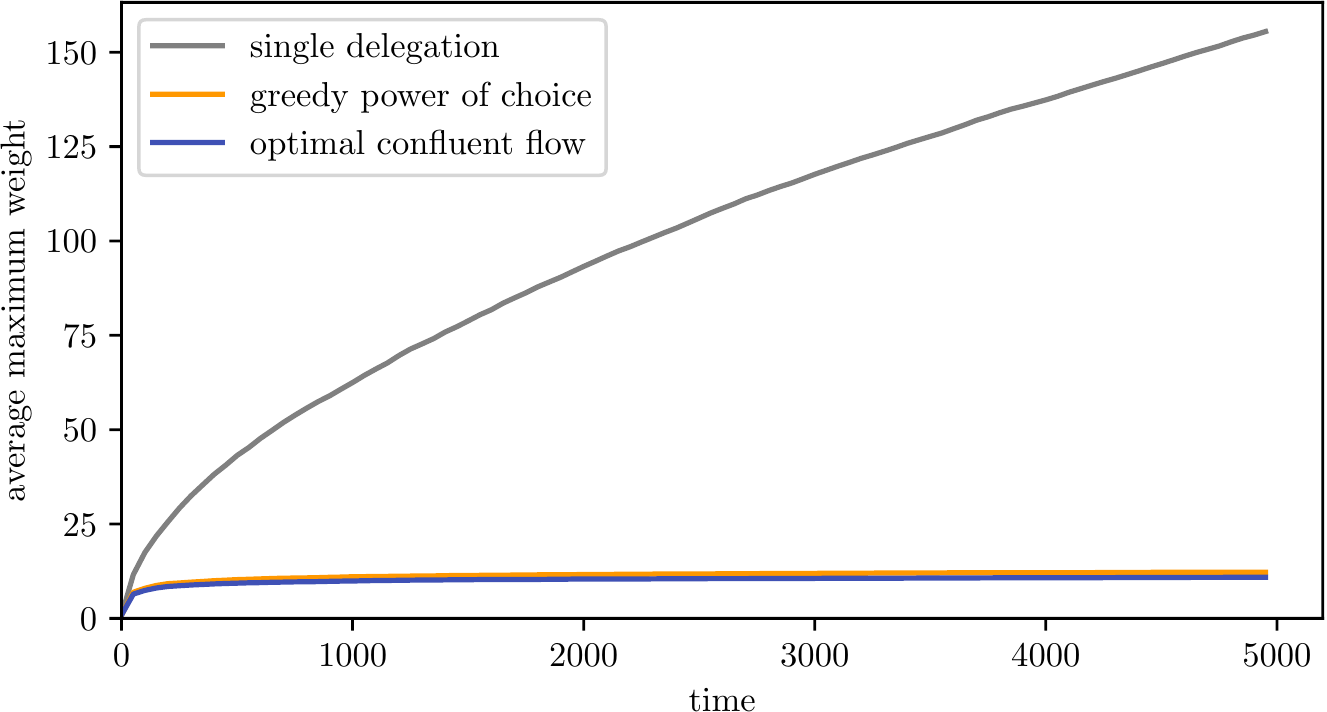}
\caption{$\gamma = 0.5$, $d = 0.5$}
\end{subfigure} \\
\begin{subfigure}[h]{.49\textwidth}
\includegraphics[width=\textwidth]{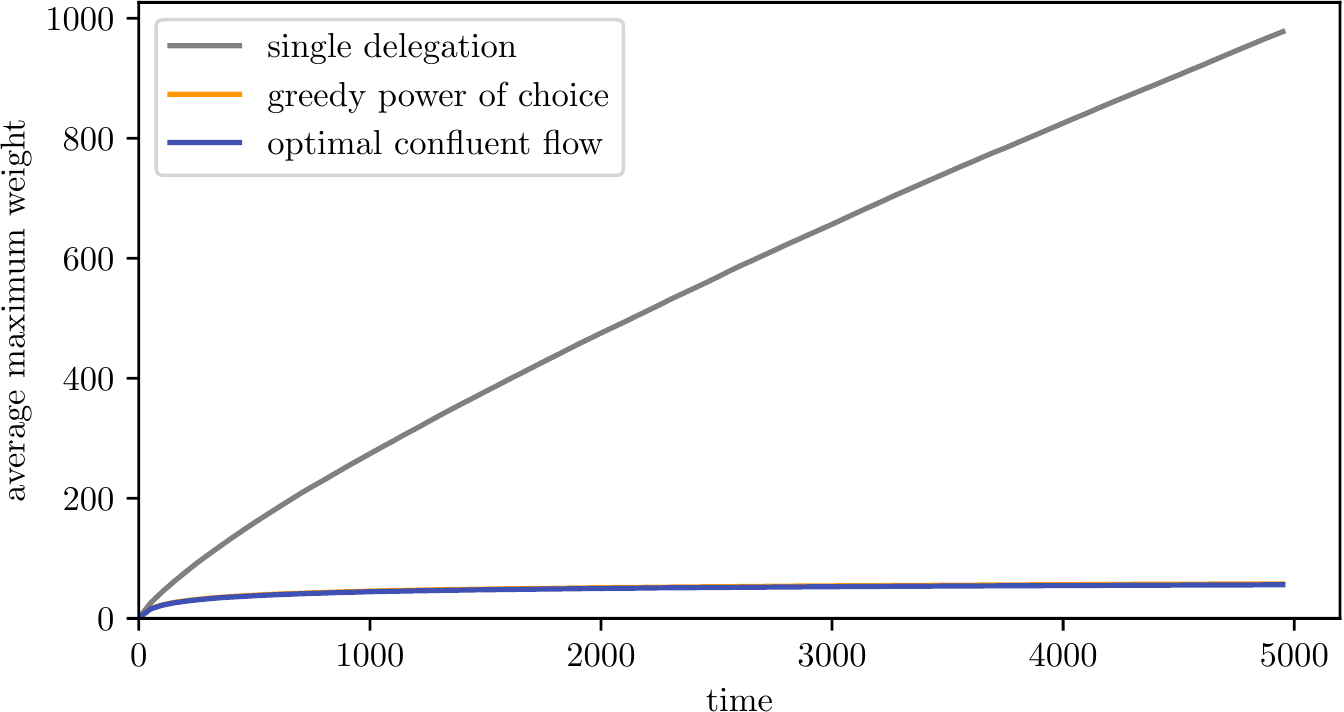}
\caption{$\gamma = 0.5$, $d = 0.75$}
\end{subfigure}
    \caption{Maximum weight averaged over 100 simulations of length 5\,000 time steps each. Maximum weight has been computed every 50 time steps.}
    \label{fig:singlevstwogammapointfive}
\end{figure}

\begin{figure}[ht]
	\centering
	\includegraphics[width=.5\textwidth]{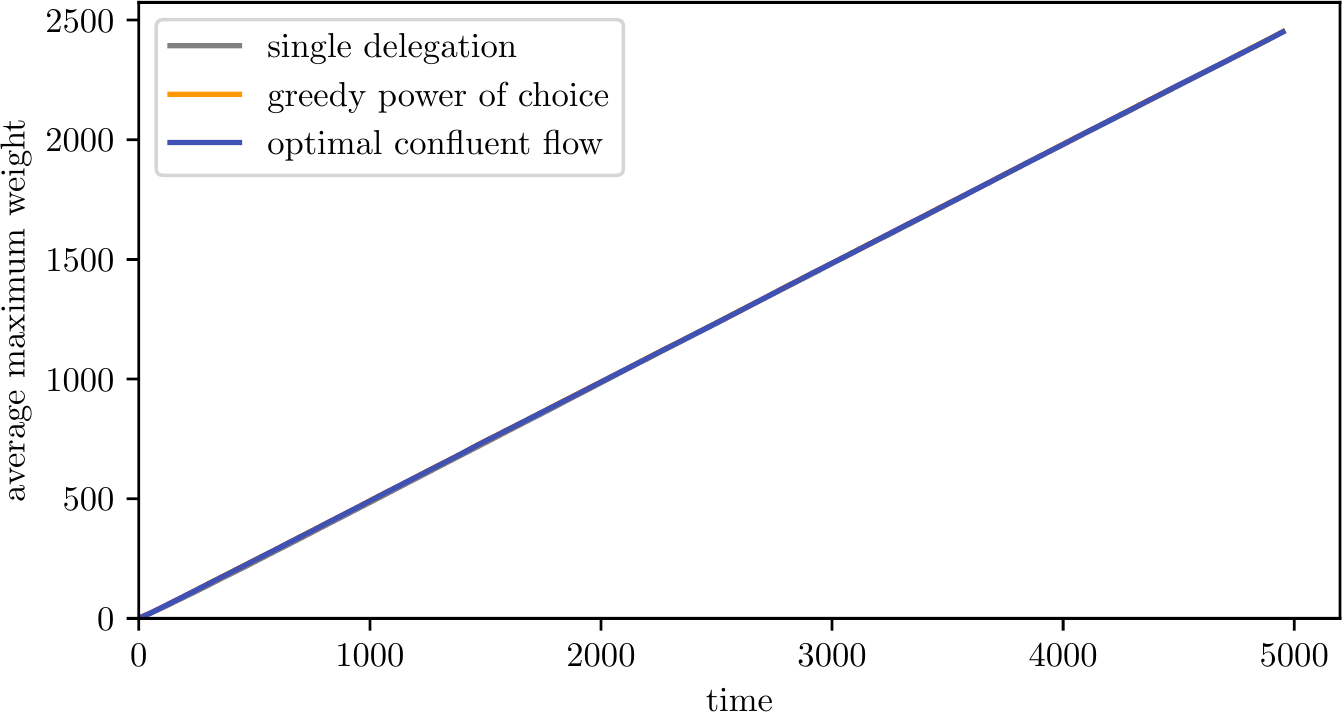}
	\captionof{figure}{Maximum weight averaged over 100 simulations, computed every 50 time steps. $\gamma = 2$, $d = 0.5$.}
	\label{fig:singlevstwogammatwo}
\end{figure}

\begin{figure}[H]
\centering
\begin{subfigure}[h]{.49\textwidth}
\includegraphics[width=\textwidth]{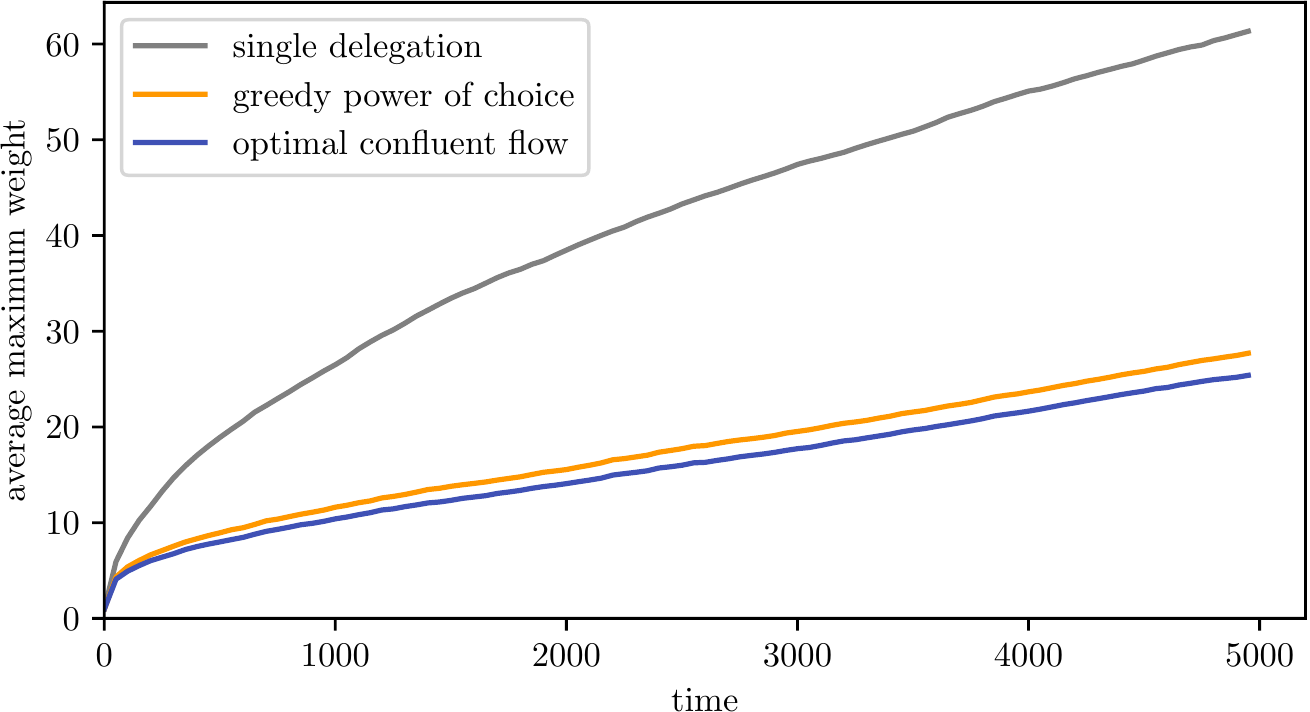}
\caption{$\gamma = 1.25$, $d = 0.25$: 1.6\,\%}
\end{subfigure}
\begin{subfigure}[h]{.49\textwidth}
\includegraphics[width=\textwidth]{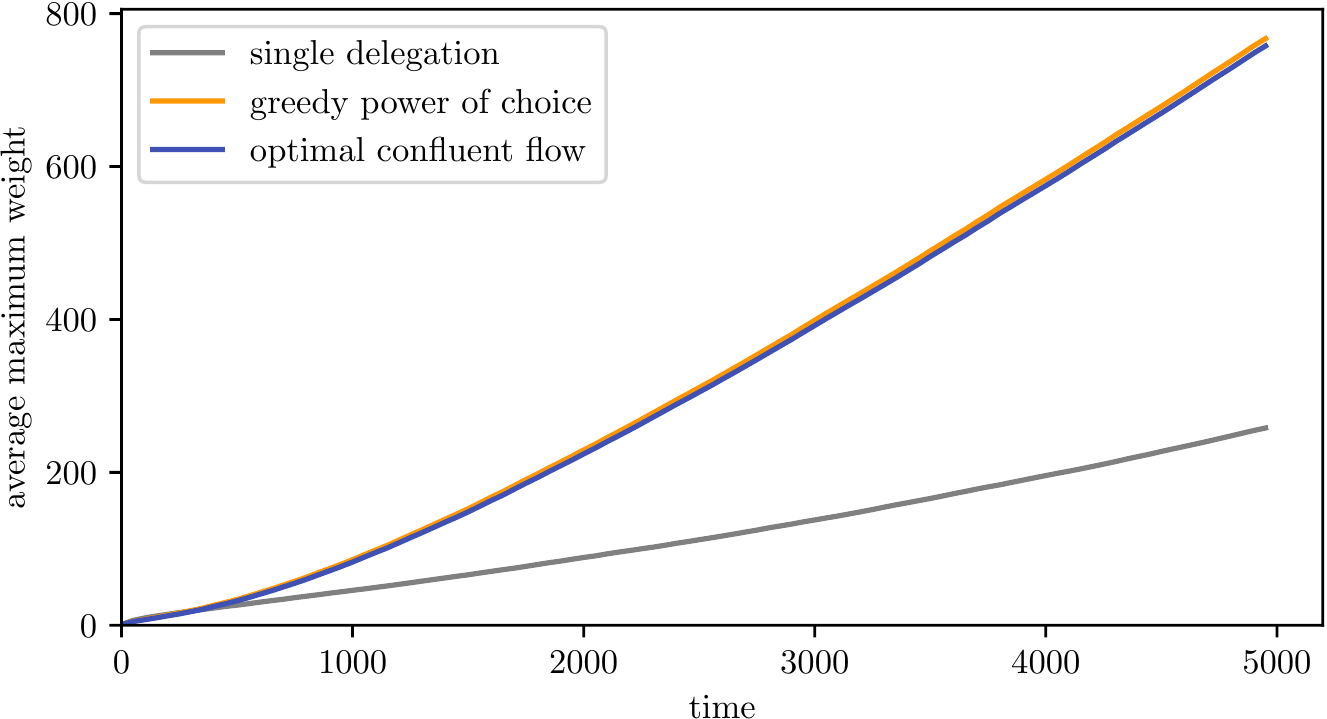}
\caption{$\gamma = 1.5$, $d = 0.25$: 59.8\,\%}
\end{subfigure} \\
    \begin{subfigure}[h]{.49\textwidth}
\includegraphics[width=\textwidth]{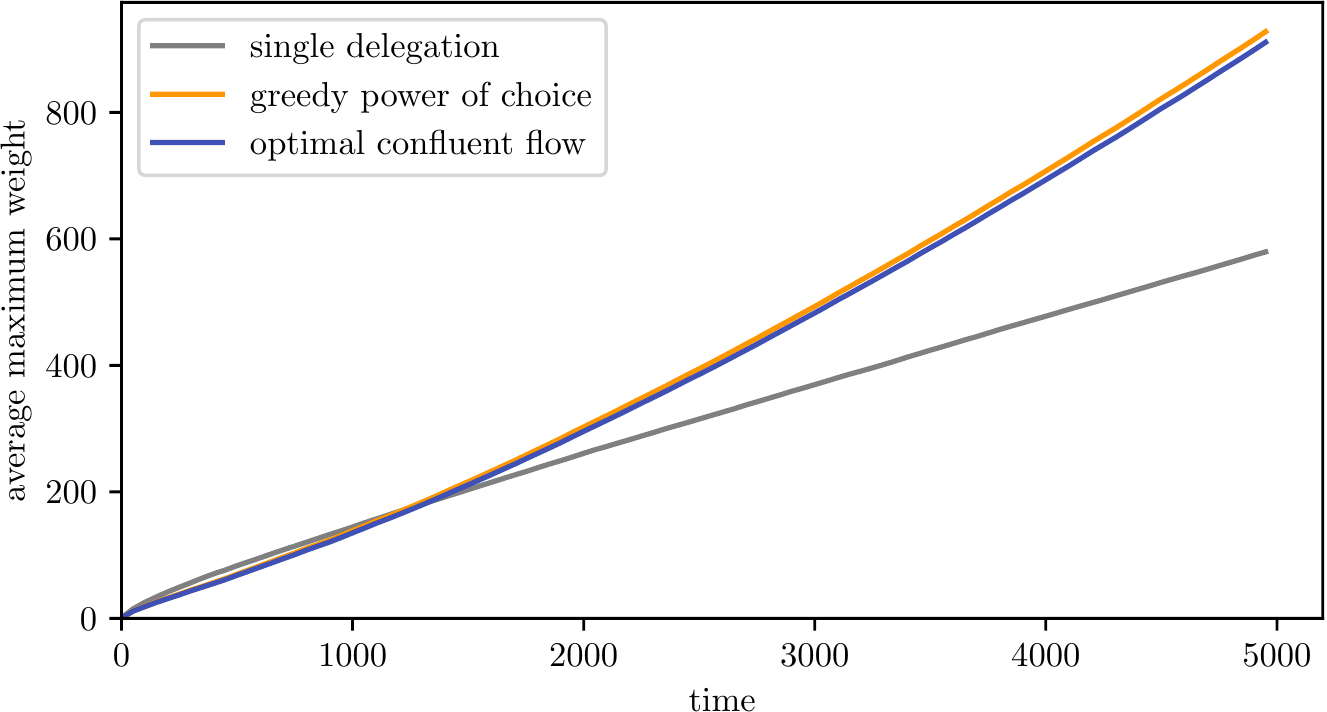}
\caption{$\gamma = 1.25$, $d = 0.5$: 28.2\,\%}
\end{subfigure}
    \begin{subfigure}[h]{.49\textwidth}
\includegraphics[width=\textwidth]{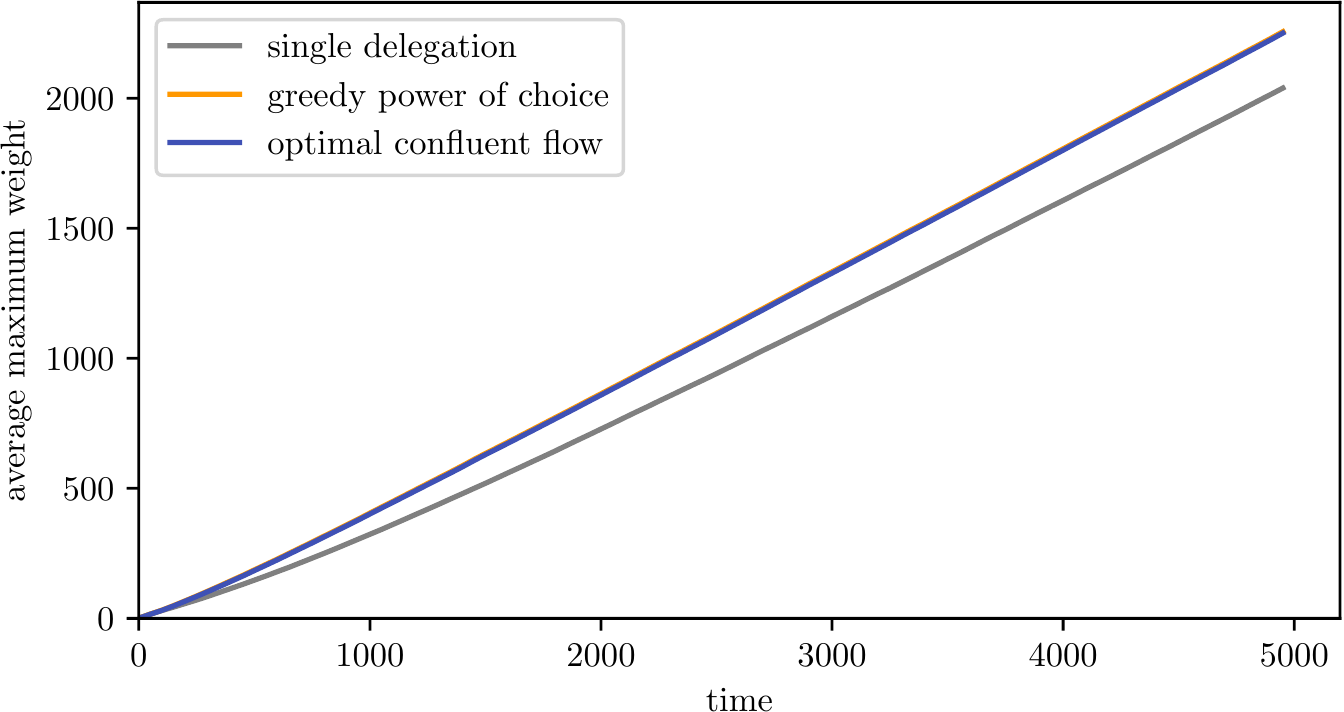}
\caption{$\gamma = 1.5$, $d = 0.5$: 87.7\,\%}
\end{subfigure} \\
\begin{subfigure}[h]{.49\textwidth}
\includegraphics[width=\textwidth]{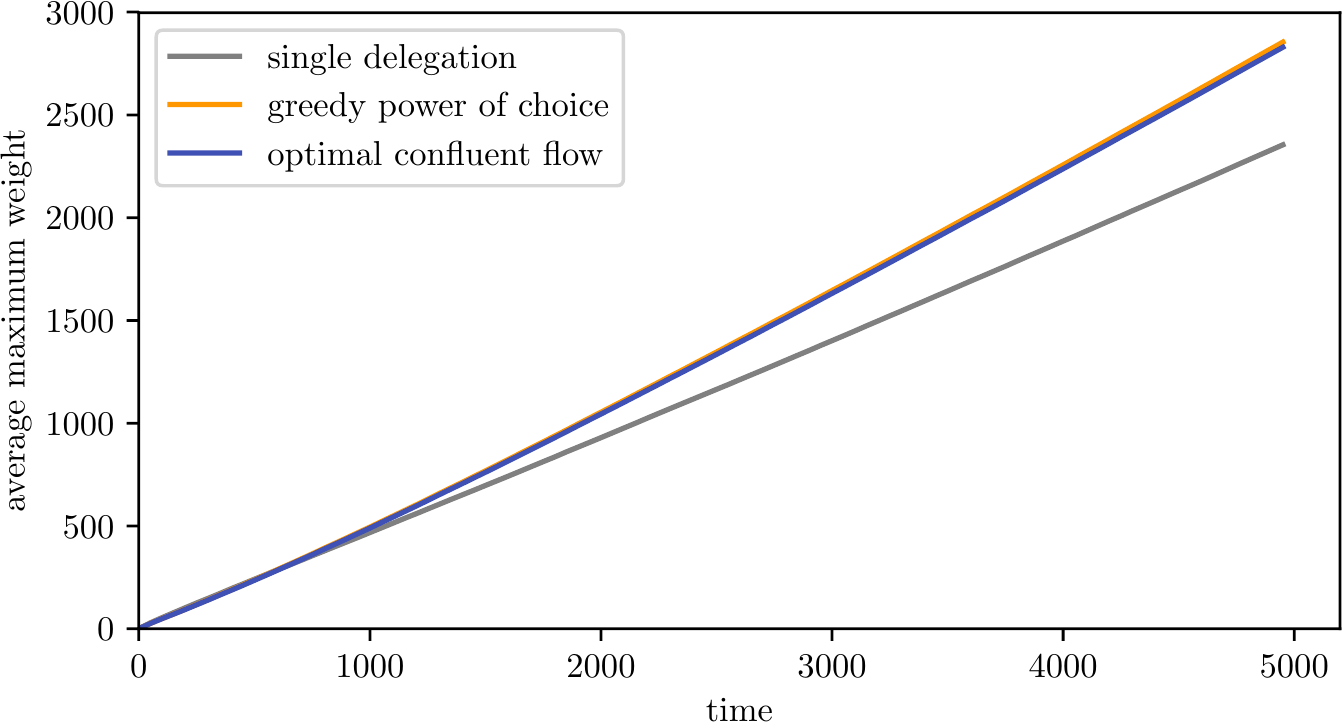}
\caption{$\gamma = 1.25$, $d = 0.75$: 55.4\,\%}
\end{subfigure} 
\begin{subfigure}[h]{.49\textwidth}
\includegraphics[width=\textwidth]{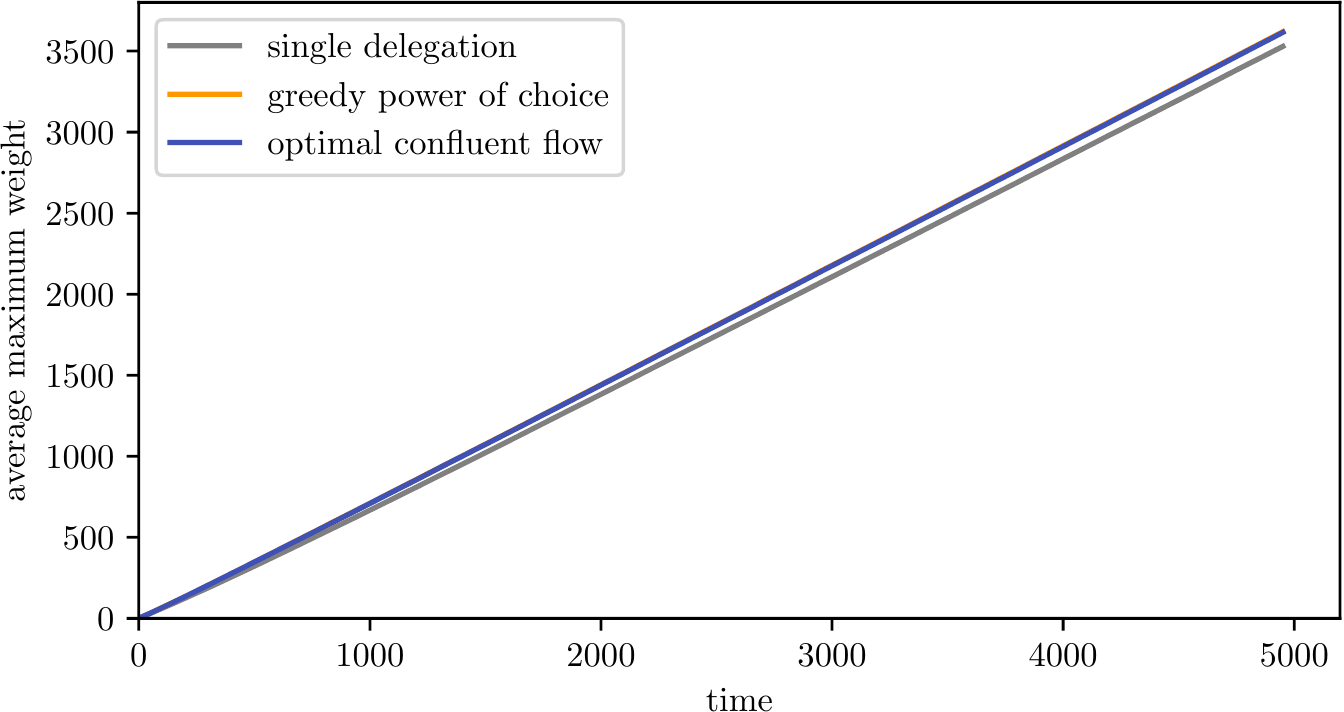}
\caption{$\gamma = 1.5$, $d = 0.75$: 94.1\,\%}
\end{subfigure}
    \caption{Maximum weight averaged over 100 simulations of length 5\,000 time steps each. Maximum weight has been computed every 50 time steps. The subfigure captions contain the percentage of delegators who give two identical delegation options.}
    \label{fig:singlevstwogammagtone}
\end{figure}

\subsection{Random vs.\ Single Delegation}
\label{app:randomvssingle}

\begin{figure}[H]
	\centering
	\begin{subfigure}[h]{.32\textwidth}
		\includegraphics[width=\textwidth]{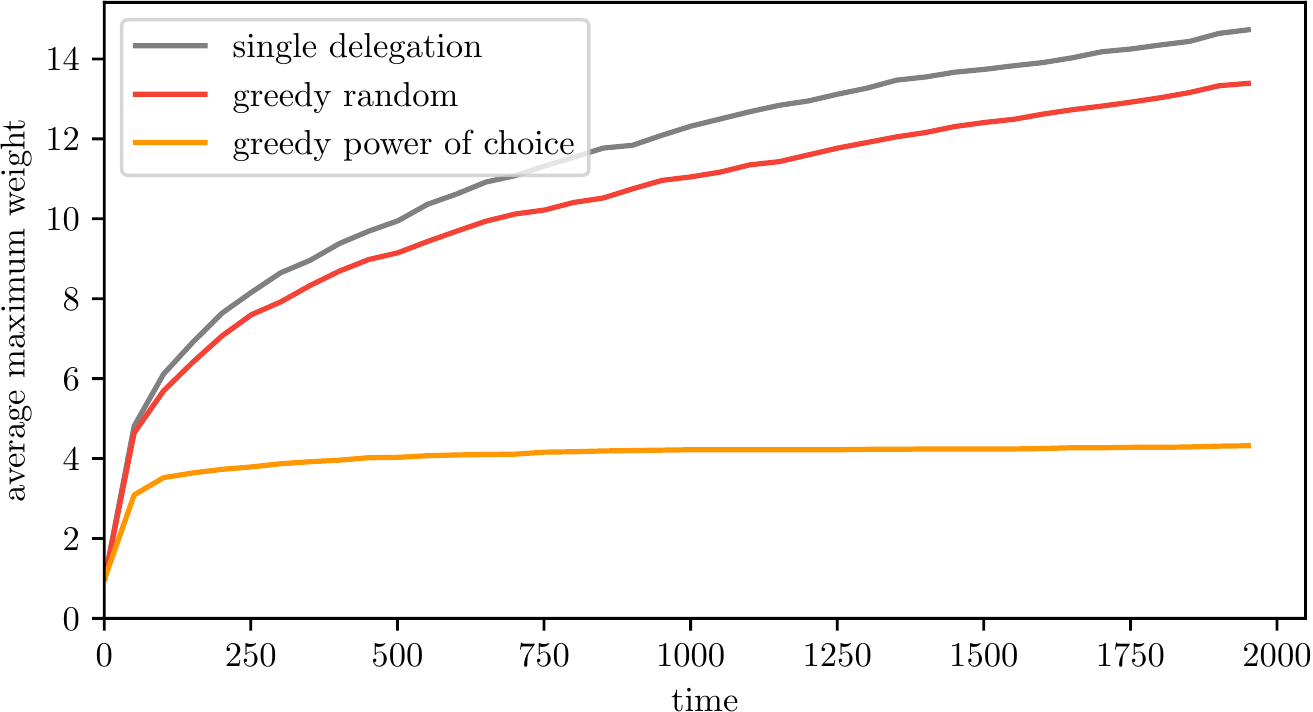}
		\caption{$\gamma = 0$, $d = 0.25$}
	\end{subfigure}
	\begin{subfigure}[h]{.32\textwidth}
		\includegraphics[width=\textwidth]{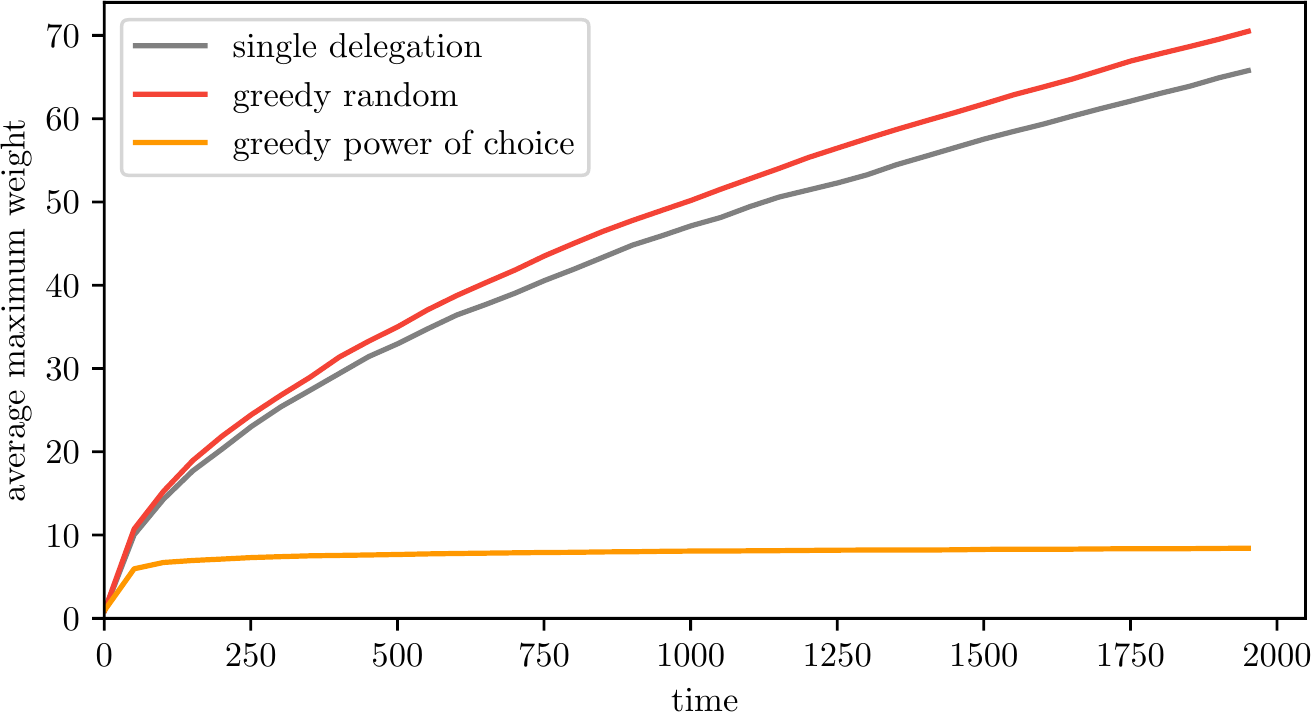}
		\caption{$\gamma = 0$, $d = 0.5$}
	\end{subfigure}
	\begin{subfigure}[h]{.32\textwidth}
		\includegraphics[width=\textwidth]{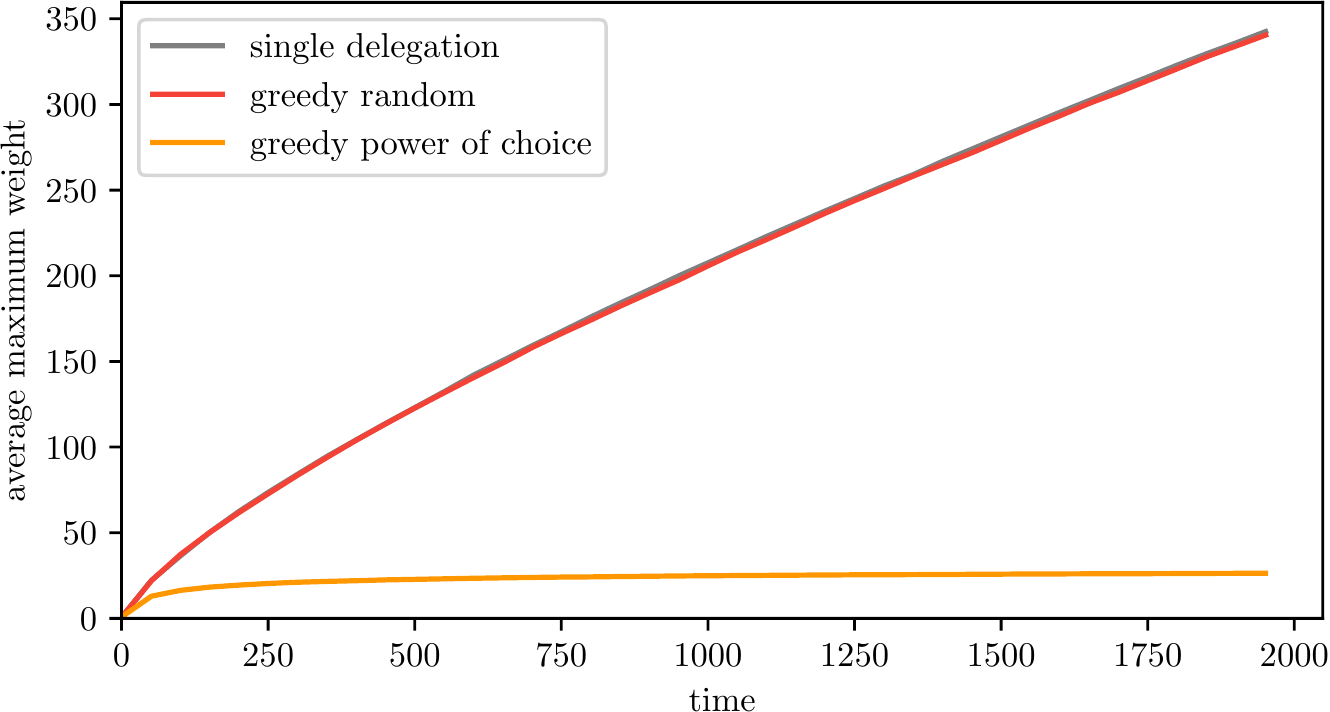}
		\caption{$\gamma = 0$, $d = 0.75$}
	\end{subfigure}	\\
	\begin{subfigure}[h]{.32\textwidth}
		\includegraphics[width=\textwidth]{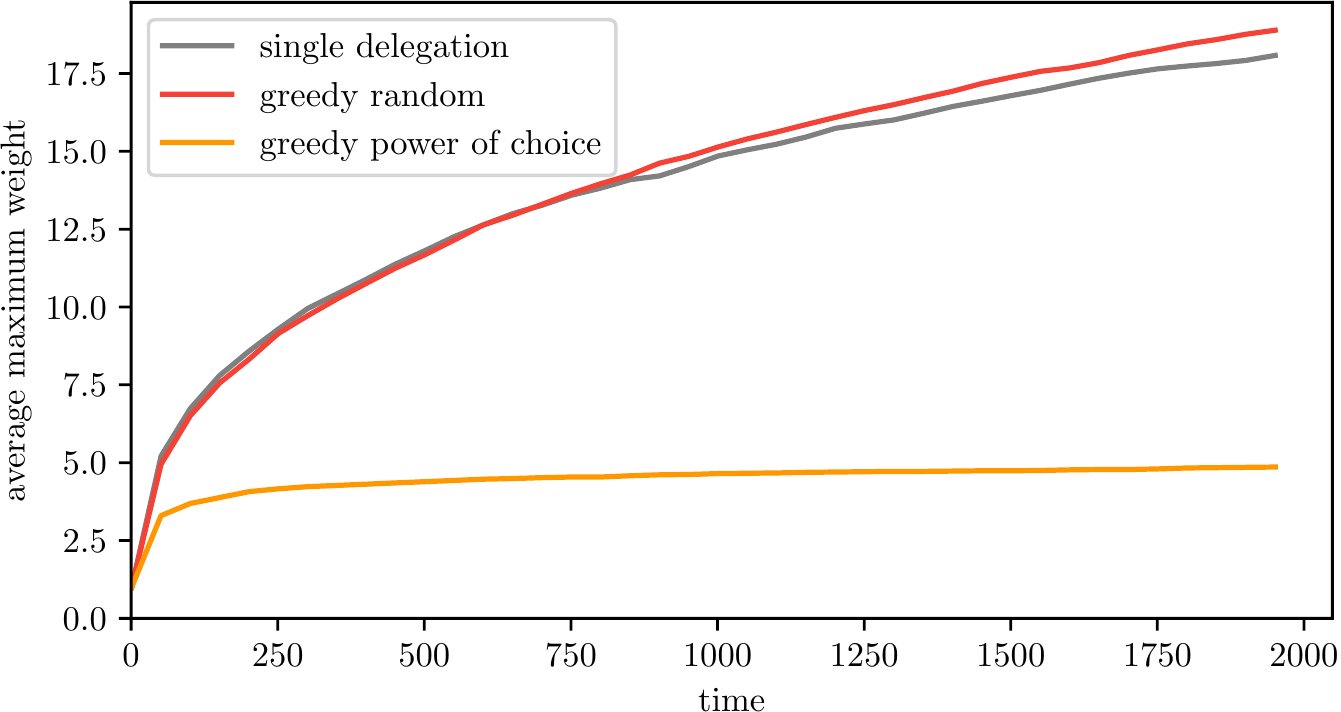}
		\caption{$\gamma = 0.5$, $d = 0.25$}
	\end{subfigure}
	\begin{subfigure}[h]{.32\textwidth}
		\includegraphics[width=\textwidth]{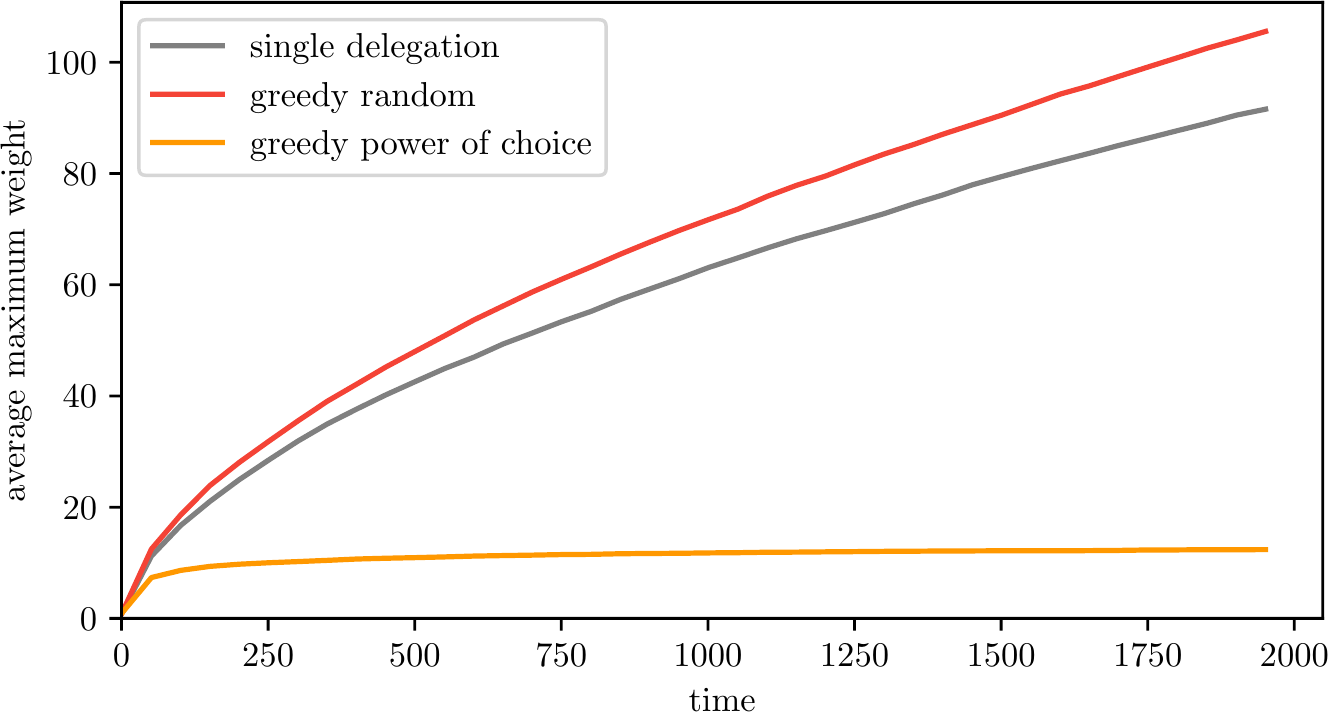}
		\caption{$\gamma = 0.5$, $d = 0.5$}
	\end{subfigure}
	\begin{subfigure}[h]{.32\textwidth}
		\includegraphics[width=\textwidth]{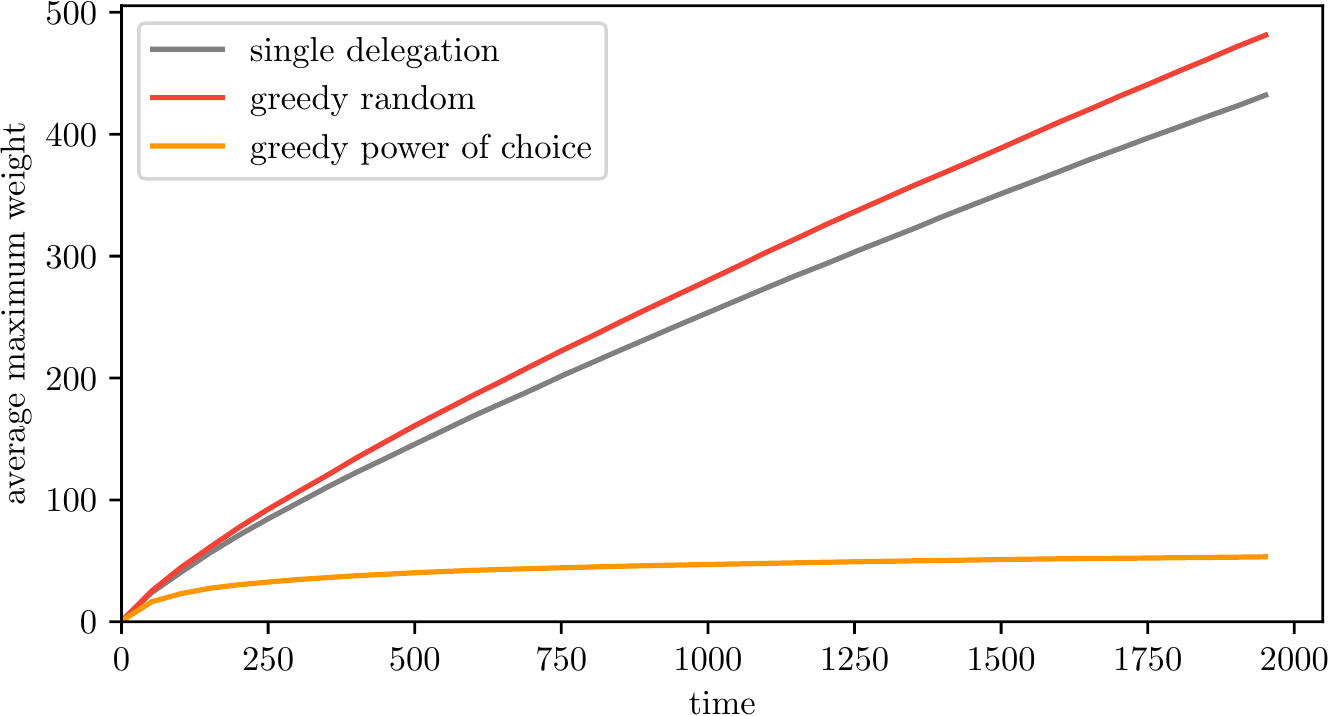}
		\caption{$\gamma = 0.5$, $d = 0.75$}
	\end{subfigure}	\\
	\begin{subfigure}[h]{.32\textwidth}
		\includegraphics[width=\textwidth]{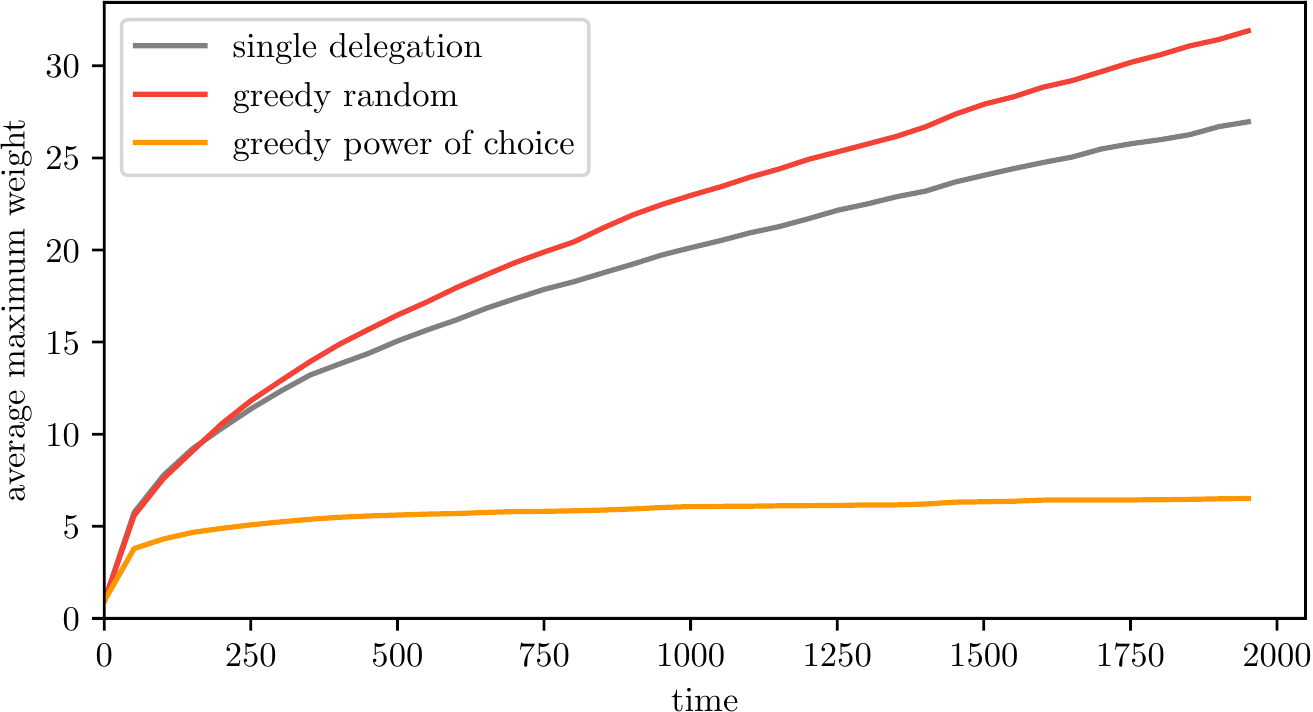}
		\caption{$\gamma = 1$, $d = 0.25$}
	\end{subfigure}
	\begin{subfigure}[h]{.32\textwidth}
		\includegraphics[width=\textwidth]{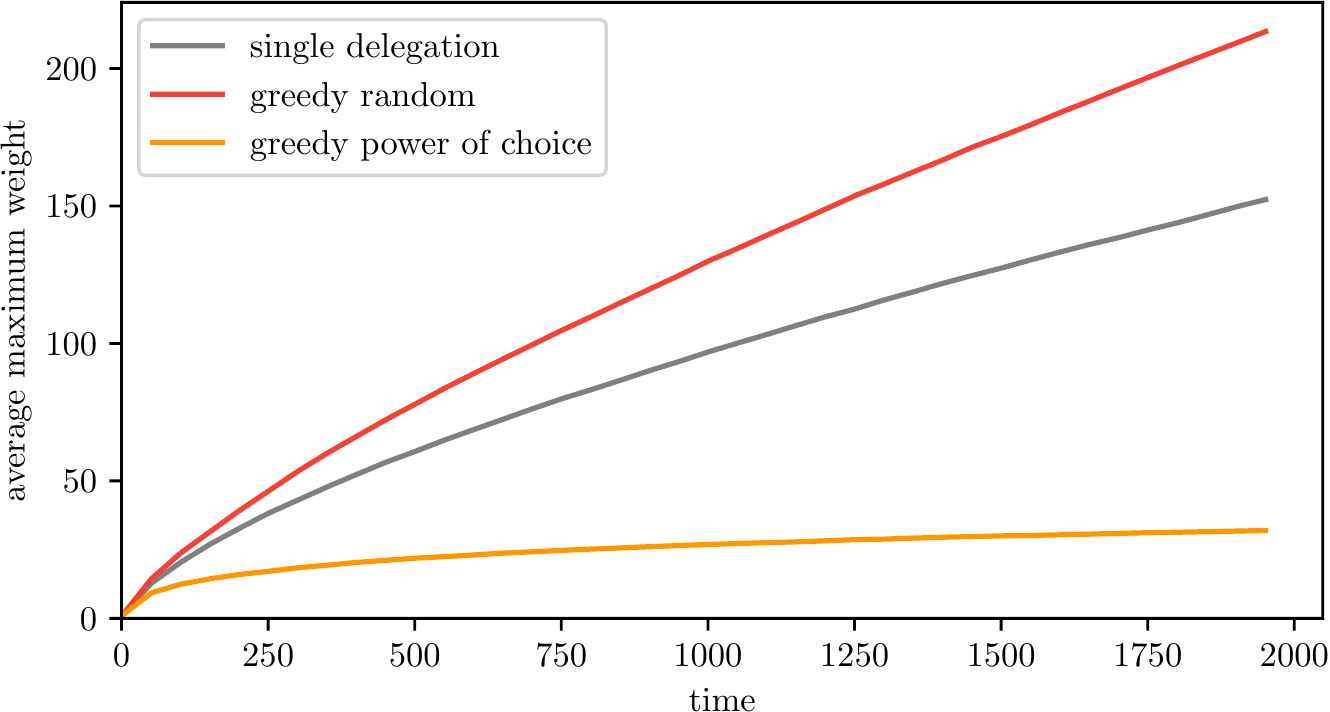}
		\caption{$\gamma = 1$, $d = 0.5$}
	\end{subfigure}
	\begin{subfigure}[h]{.32\textwidth}
		\includegraphics[width=\textwidth]{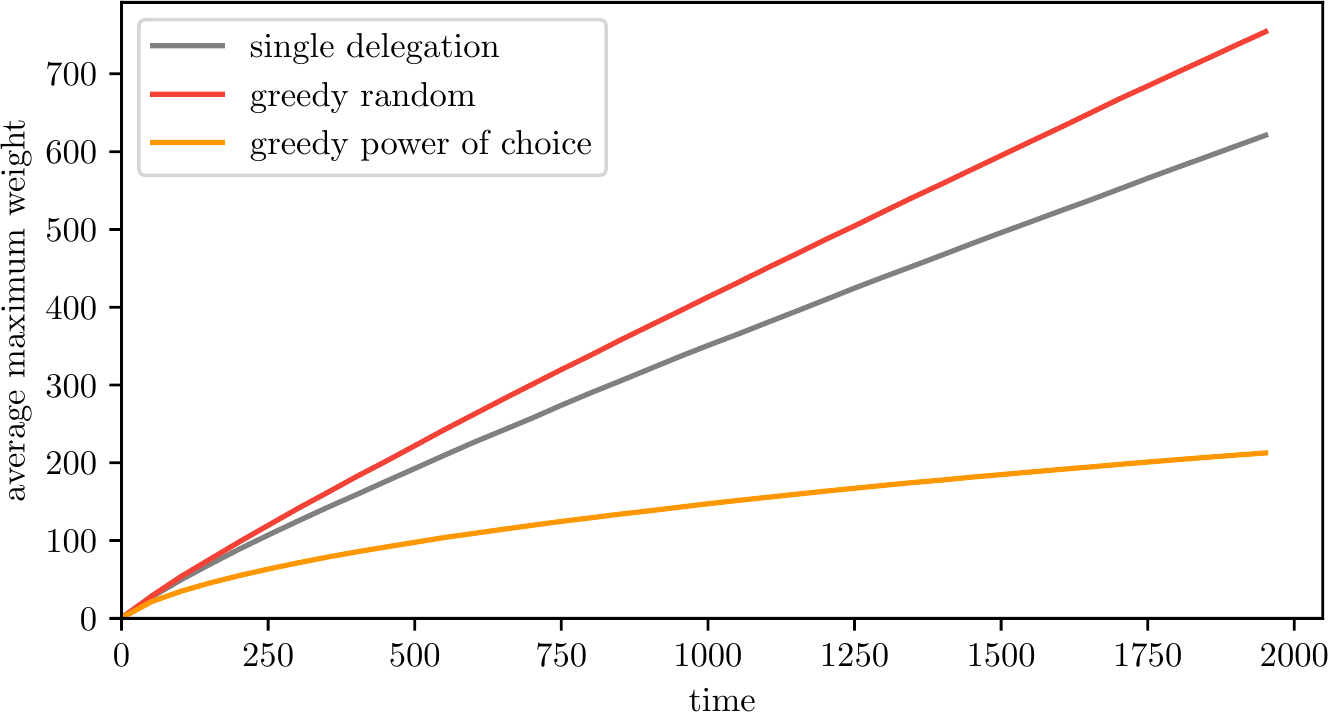}
		\caption{$\gamma = 1$, $d = 0.75$}
	\end{subfigure}
	\caption{Maximum weight averaged over 200 simulations of length 2\,000 time steps each. Maximum weight has been computed every 20 time steps.}
	\label{fig:randomvssingle}
\end{figure}

\subsection{Confluent vs.\ Splittable Flow}
\label{app:confsplit}

In order to determine whether we could gain significantly from relaxing the problem to that of splittable delegations, we evaluate the penalty incurred from enforcing confluence. 
That is, we examine the requirement that each delegator must delegate all of her weight to exactly one other agent instead of splitting her vote among multiple agents. 
This is equivalent to comparing the optimal solutions to the problems of confluent and splittable flow on the same graph.
We compute these solutions by solving an MILP and LP, respectively. 
	
As seen in \cref{appfig:confsplit}, the difference between the two solutions is negligible even for large values of $t$. 
\cref{appsubfig:confsplitsingle} plots a single run of the two solutions over time and suggests that the confluent solution is very close to the ceiling of the fractional LP solution. 
\cref{appsubfig:confsplitavg} averages the optimal confluent and splittable solutions over $100$ traces to demonstrate that, in our setting, the solution for confluent flow closely approximates the less constrained solution to splittable flow on average.

\begin{figure}[H]
	\centering
	\begin{subfigure}[h]{0.49\textwidth}
		\includegraphics[width=\linewidth]{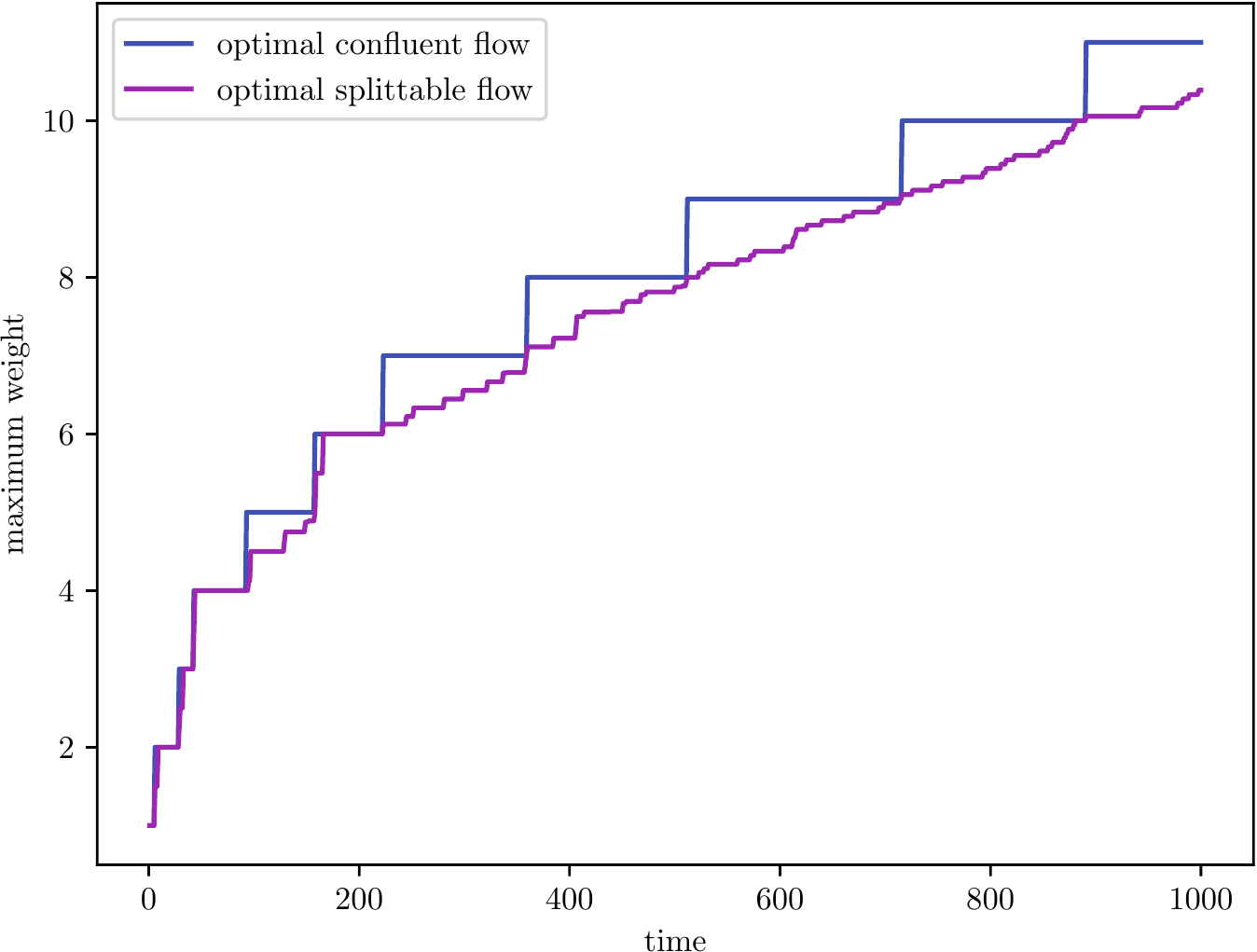}
		\caption{Single trace}
		\label{appsubfig:confsplitsingle}
	\end{subfigure}
	\begin{subfigure}[h]{0.49\textwidth}
		\includegraphics[width=\linewidth]{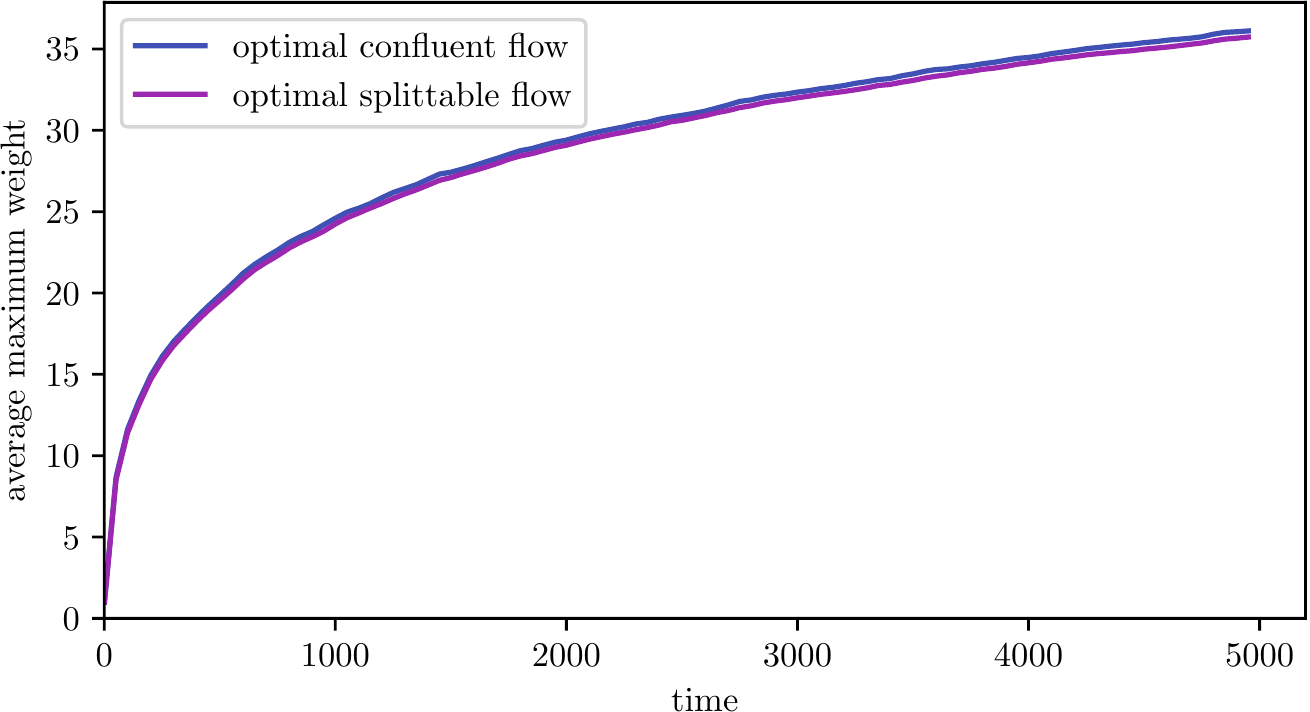}
		\caption{Averaged over $100$ traces}
		\label{appsubfig:confsplitavg}
	\end{subfigure}%
	\caption{Confluent vs.\ splittable flow: $\gamma = 1$, $d = 0.5$, $k = 2$.}
	\label{appfig:confsplit}
\end{figure}

 \subsection{Histograms}
 \label{app:histograms}

 \begin{figure}[H]
 	\centering
 	\begin{subfigure}[h]{0.49\textwidth}
 		\includegraphics[width=\linewidth]{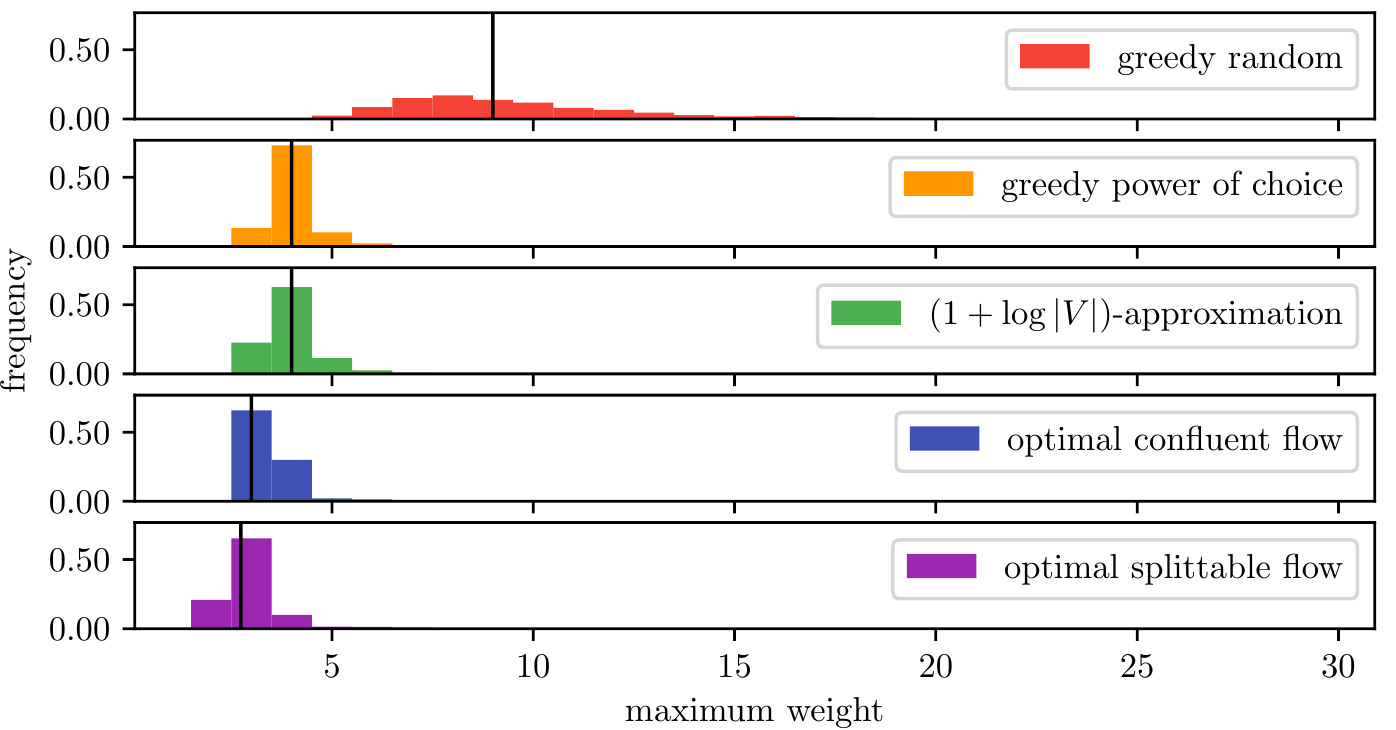}
 		\caption{All mechanisms and optimal splittable solution}
 	\end{subfigure}
 	\begin{subfigure}[h]{0.49\textwidth}
 		\includegraphics[width=\linewidth]{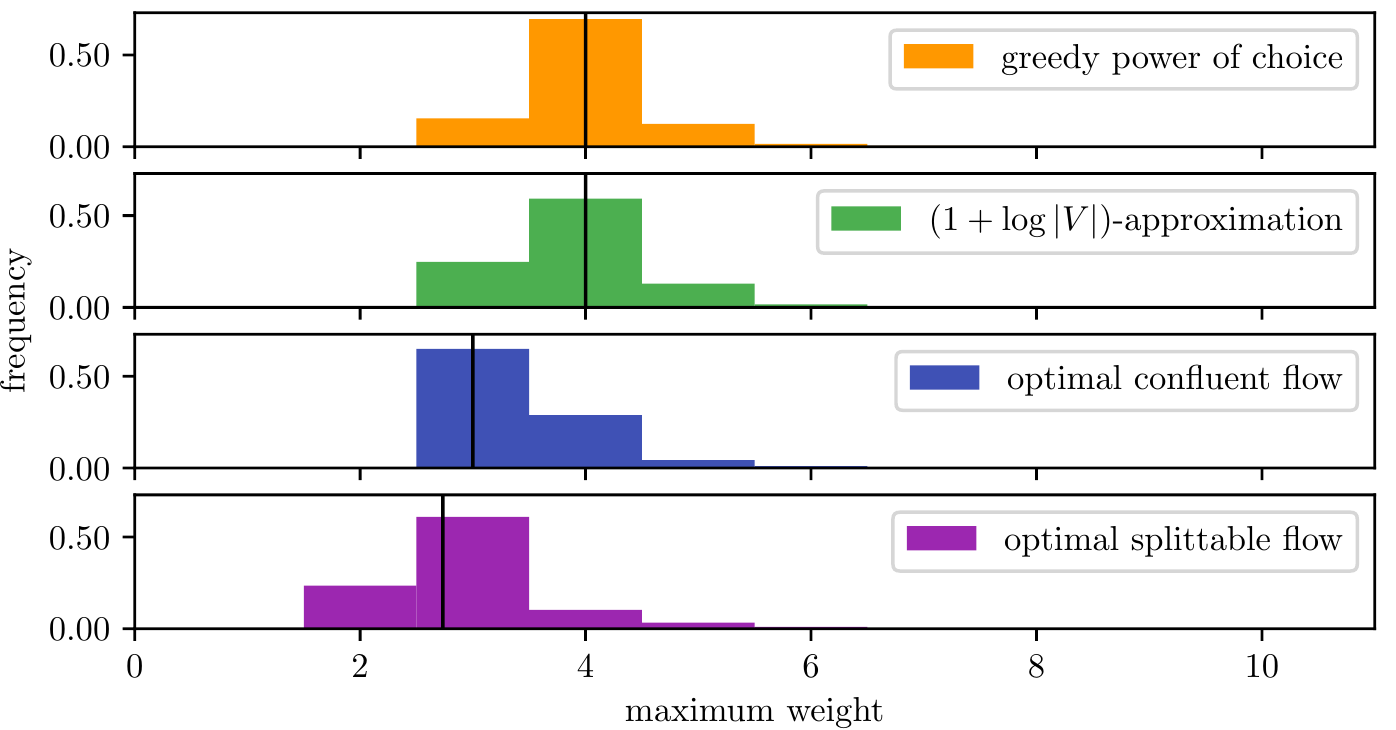}
 		\caption{Excluding random delegation}
 	\end{subfigure}%
 	\caption{Frequency of maximum weights at time $500$ over $1\,000$ runs: $\gamma = 0$, $d = 0.25$, $k = 2$.}
 	\label{fig:histg0d25}
 \end{figure}

 \begin{figure}[H]
 	\centering
 	\begin{subfigure}[h]{0.49\textwidth}
 		\includegraphics[width=\linewidth]{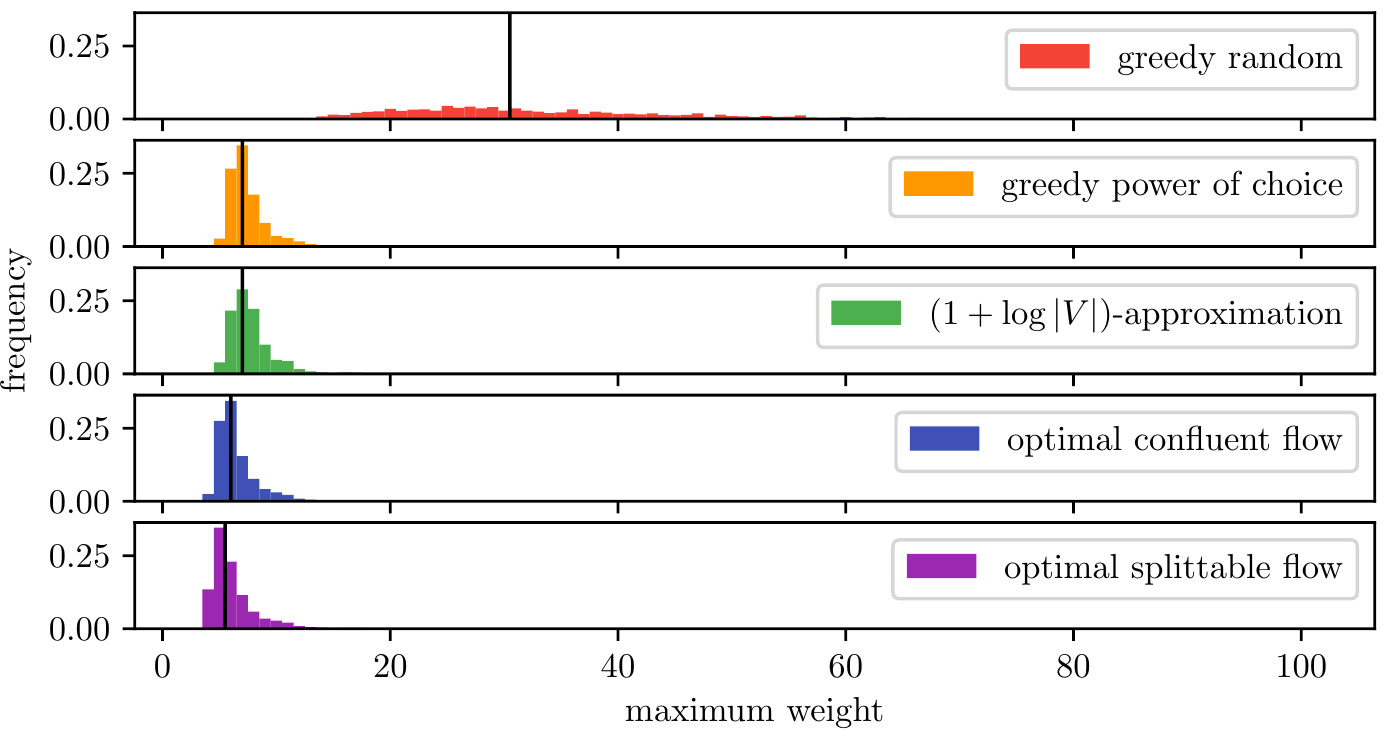}
 		\caption{All mechanisms and optimal splittable solution}
 	\end{subfigure}
 	\begin{subfigure}[h]{0.49\textwidth}
 		\includegraphics[width=\linewidth]{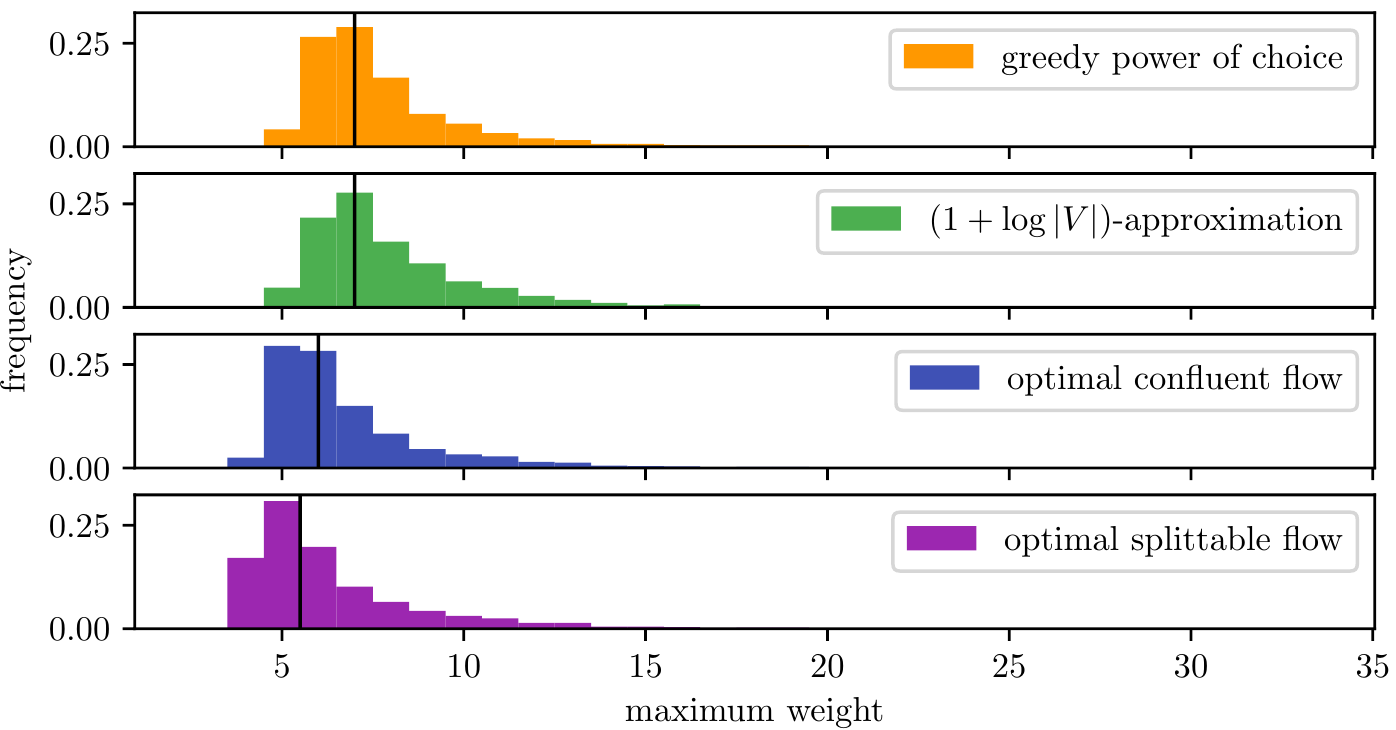}
 		\caption{Excluding random delegation}
 	\end{subfigure}%
 	\caption{Frequency of maximum weights at time $500$ over $1\,000$ runs: $\gamma = 0$, $d = 0.5$, $k = 2$.}
 	\label{fig:histg0d50}
 \end{figure}

 \begin{figure}[H]
 	\centering
 	\begin{subfigure}[h]{0.49\textwidth}
 		\includegraphics[width=\linewidth]{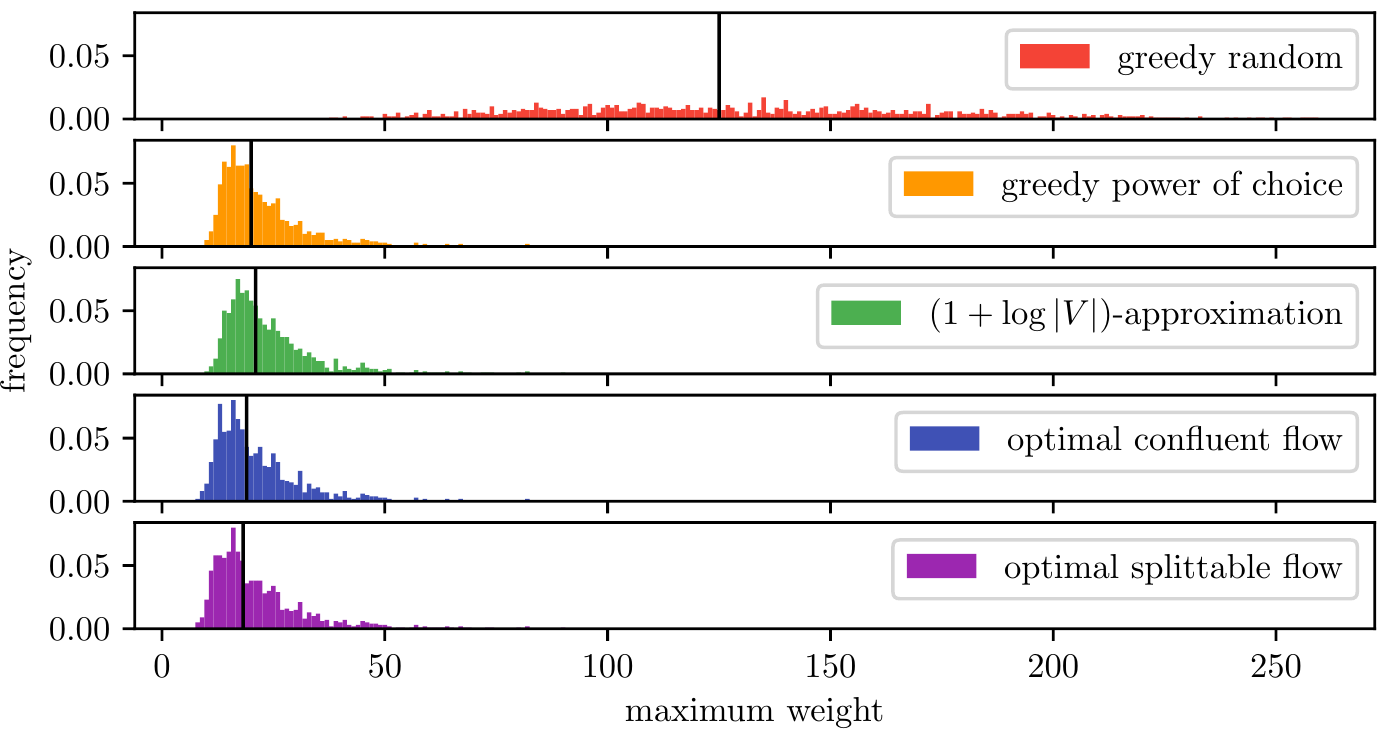}
 		\caption{All mechanisms and optimal splittable solution}
 	\end{subfigure}
 	\begin{subfigure}[h]{0.49\textwidth}
 		\includegraphics[width=\linewidth]{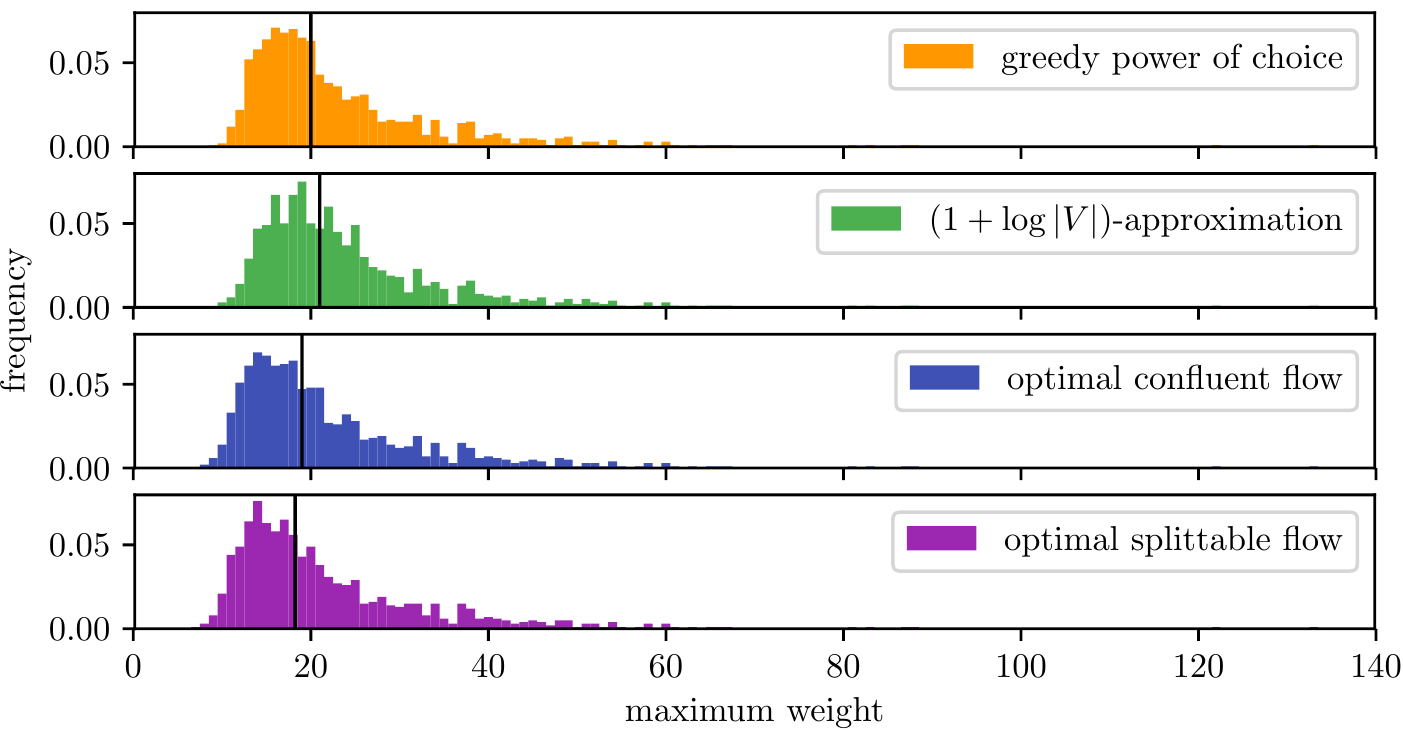}
 		\caption{Excluding random delegation}
 	\end{subfigure}%
 	\caption{Frequency of maximum weights at time $500$ over $1\,000$ runs: $\gamma = 0$, $d = 0.75$, $k = 2$.}
 	\label{fig:histg0d75}
 \end{figure}

 \begin{figure}[H]
 	\centering
 	\begin{subfigure}[h]{0.49\textwidth}
 		\includegraphics[width=\linewidth]{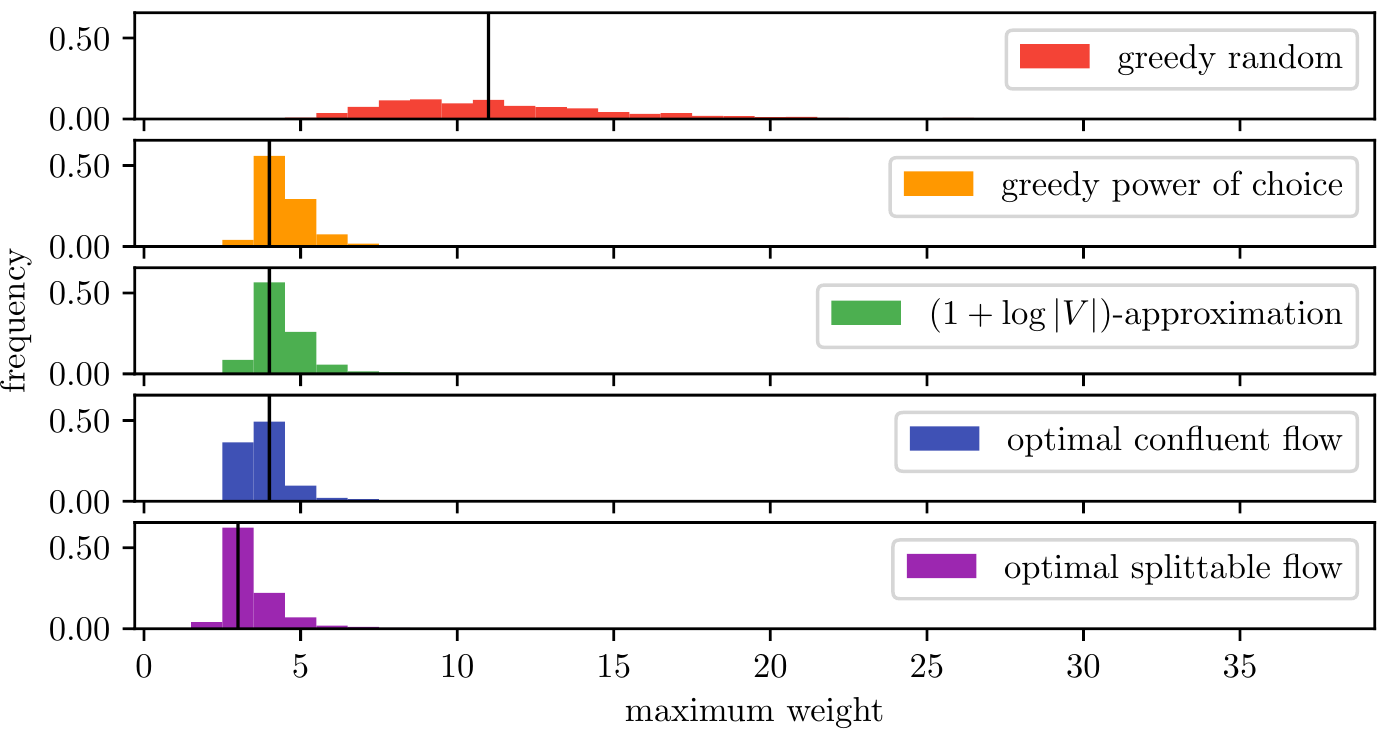}
 		\caption{All mechanisms and optimal splittable solution}
 	\end{subfigure}
 	\begin{subfigure}[h]{0.49\textwidth}
 		\includegraphics[width=\linewidth]{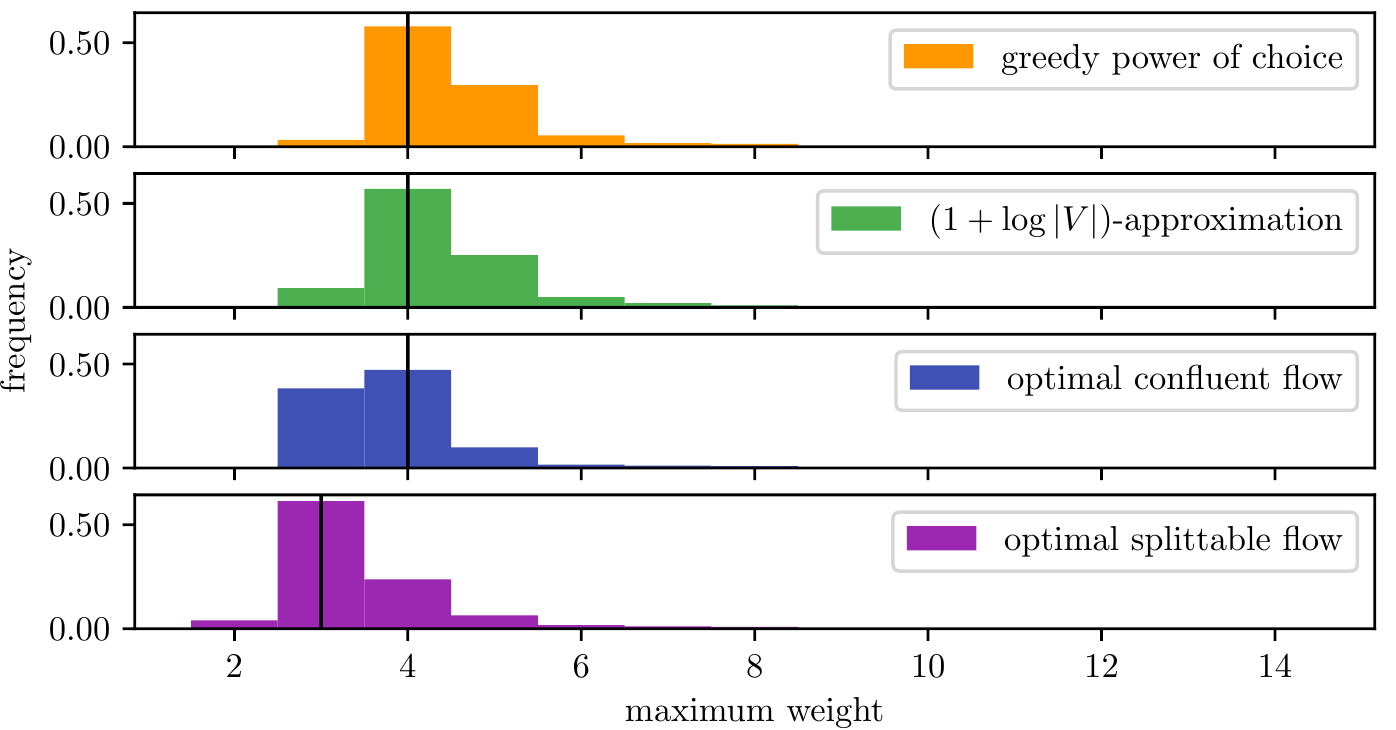}
 		\caption{Excluding random delegation}
 	\end{subfigure}%
 	\caption{Frequency of maximum weights at time $500$ over $1\,000$ runs: $\gamma = 0.5$, $d = 0.25$, $k = 2$.}
 	\label{fig:histg50d25}
 \end{figure}

 \begin{figure}[H]
 	\centering
 	\begin{subfigure}[h]{0.49\textwidth}
 		\includegraphics[width=\linewidth]{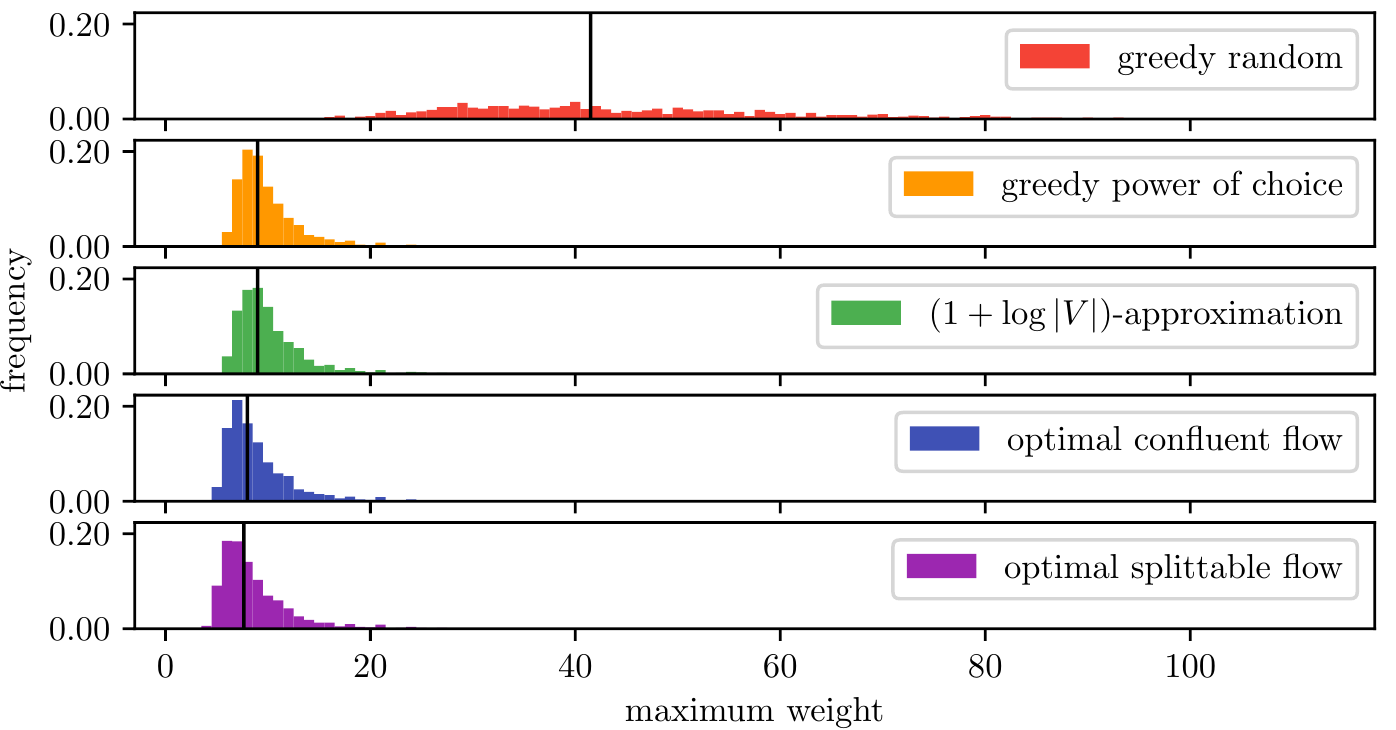}
 		\caption{All mechanisms and optimal splittable solution}
 	\end{subfigure}
 	\begin{subfigure}[h]{0.49\textwidth}
 		\includegraphics[width=\linewidth]{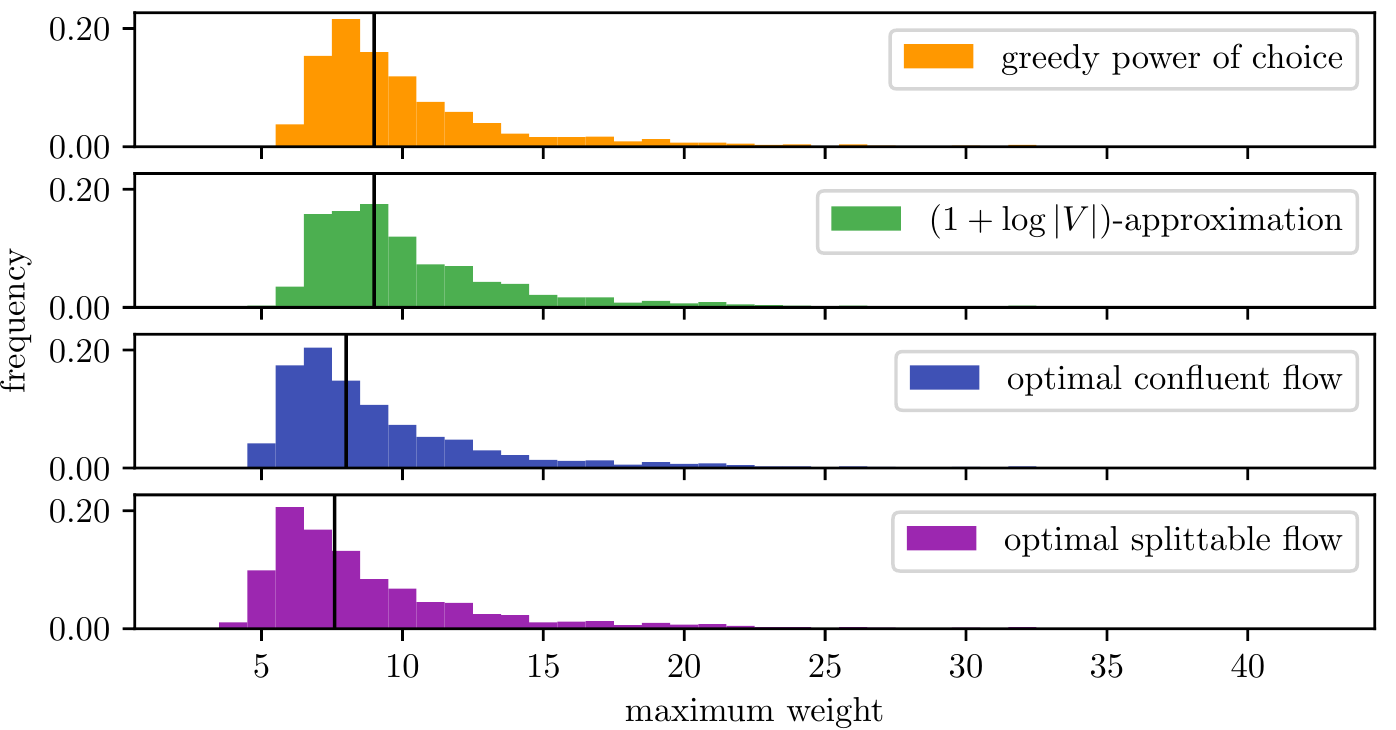}
 		\caption{Excluding random delegation}
 	\end{subfigure}%
 	\caption{Frequency of maximum weights at time $500$ over $1\,000$ runs: $\gamma = 0.5$, $d = 0.5$, $k = 2$.}
 	\label{fig:histg50d50}
 \end{figure}

 \begin{figure}[H]
 	\centering
 	\begin{subfigure}[h]{0.49\textwidth}
 		\includegraphics[width=\linewidth]{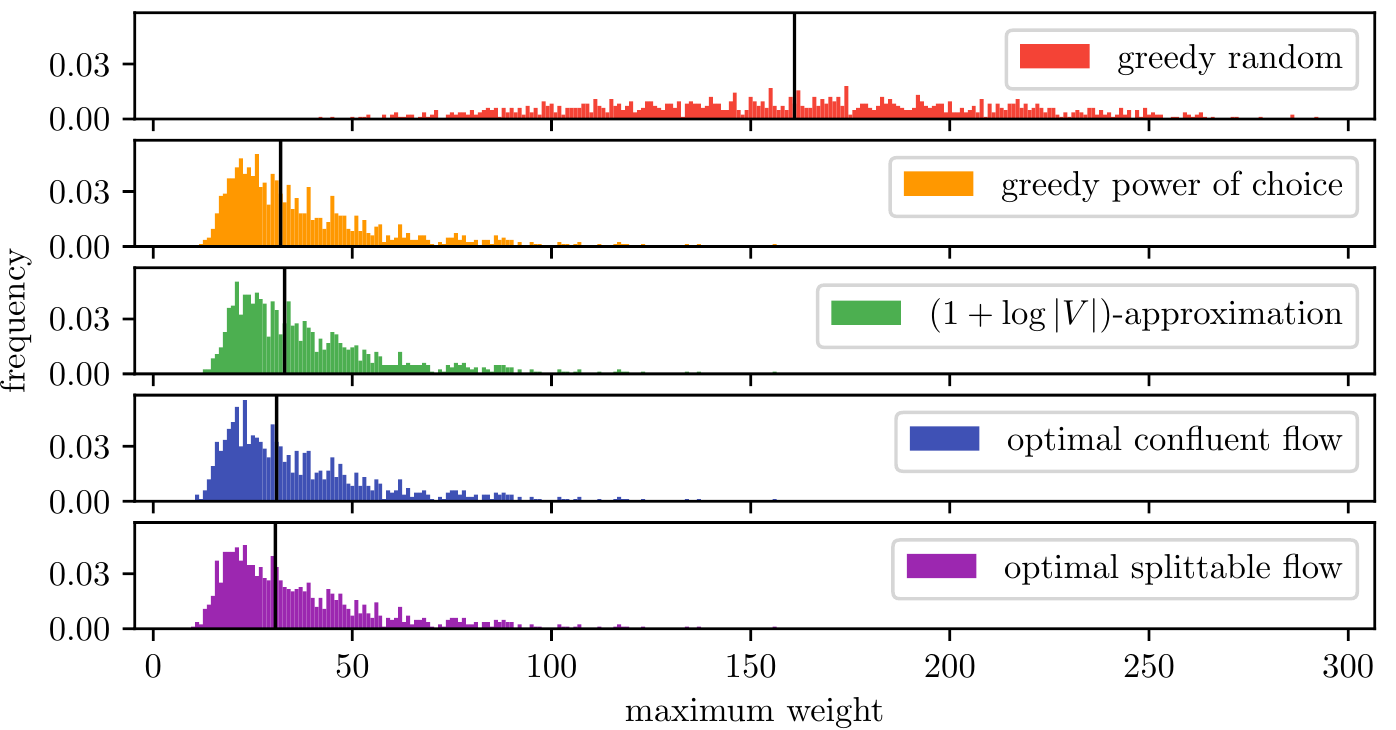}
 		\caption{All mechanisms and optimal splittable solution}
 	\end{subfigure}
 	\begin{subfigure}[h]{0.49\textwidth}
 		\includegraphics[width=\linewidth]{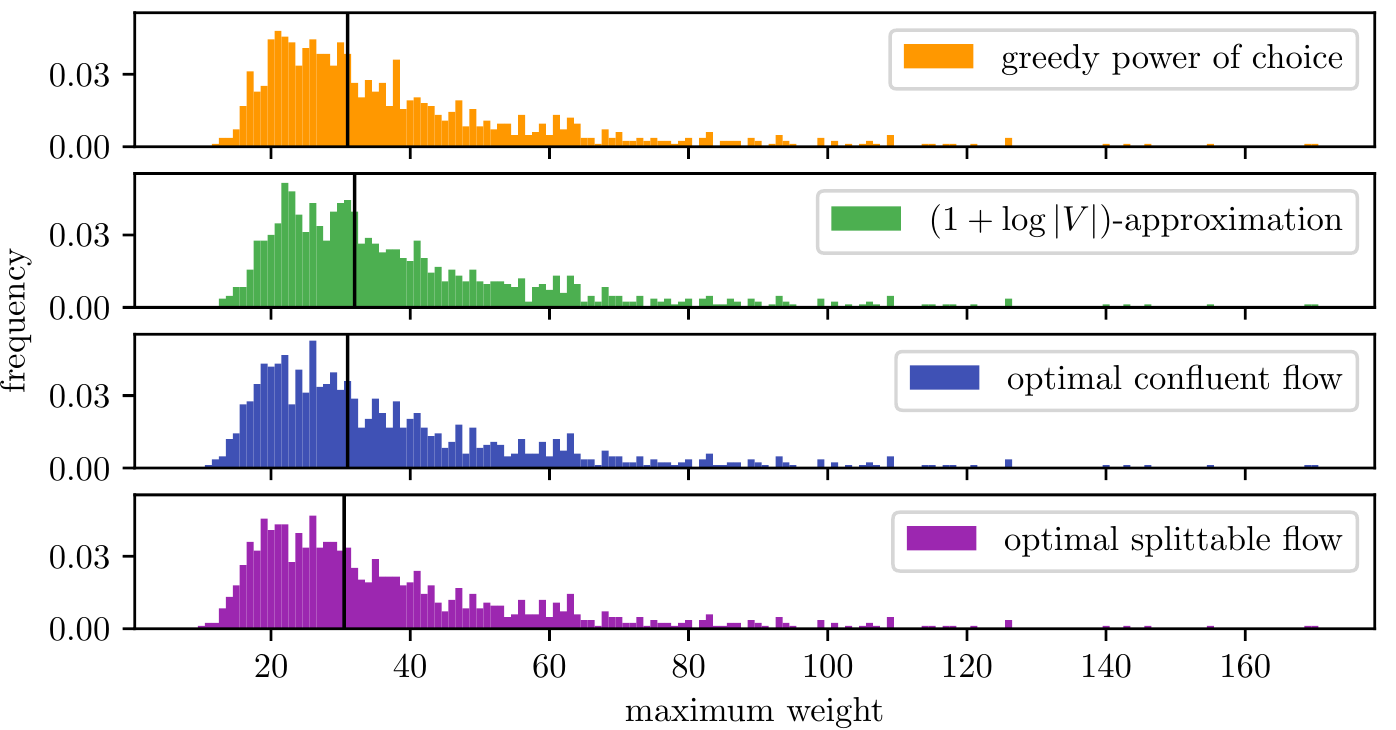}
 		\caption{Excluding random delegation}
 	\end{subfigure}%
 	\caption{Frequency of maximum weights at time $500$ over $1\,000$ runs: $\gamma = 0.5$, $d = 0.75$, $k = 2$.}
 	\label{fig:histg50d75}
 \end{figure}

 \begin{figure}[H]
 	\centering
 	\begin{subfigure}[h]{0.49\textwidth}
 		\includegraphics[width=\linewidth]{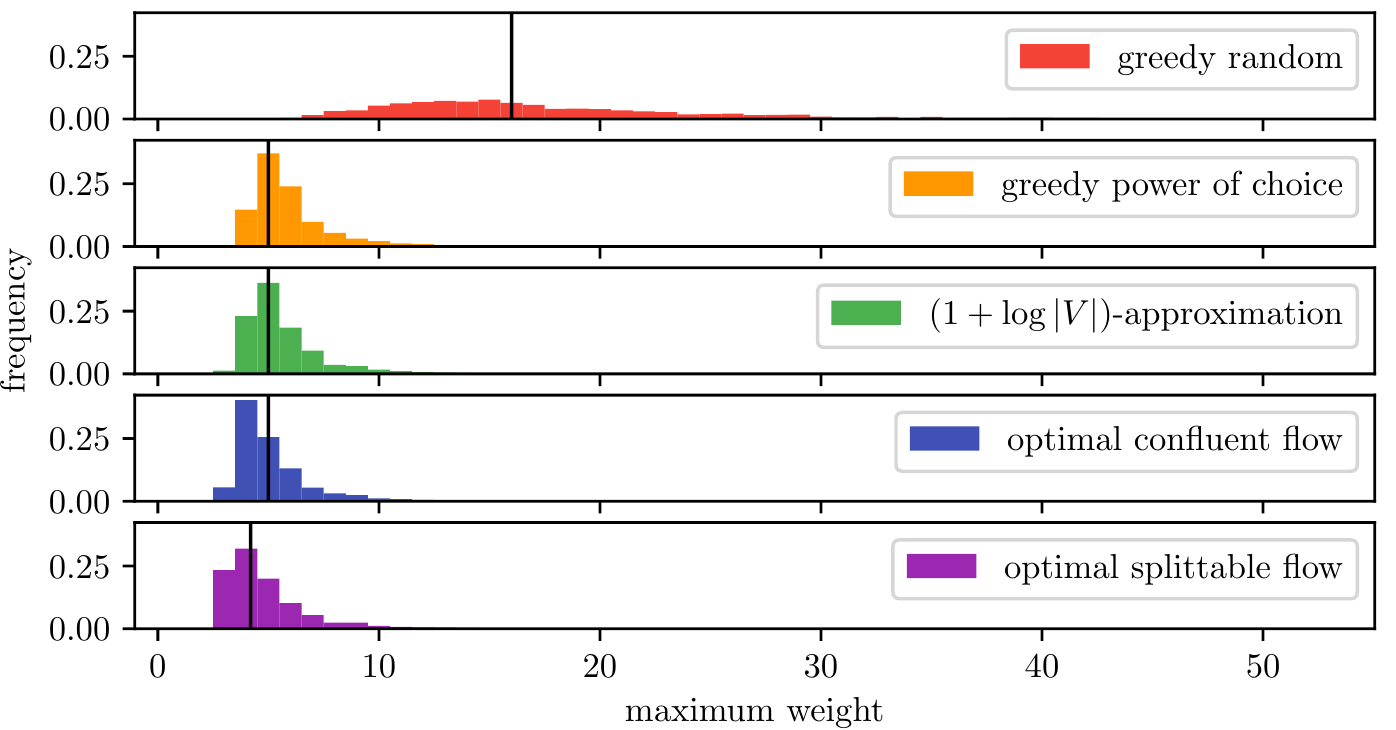}
 		\caption{All mechanisms and optimal splittable solution}
 	\end{subfigure}
 	\begin{subfigure}[h]{0.49\textwidth}
 		\includegraphics[width=\linewidth]{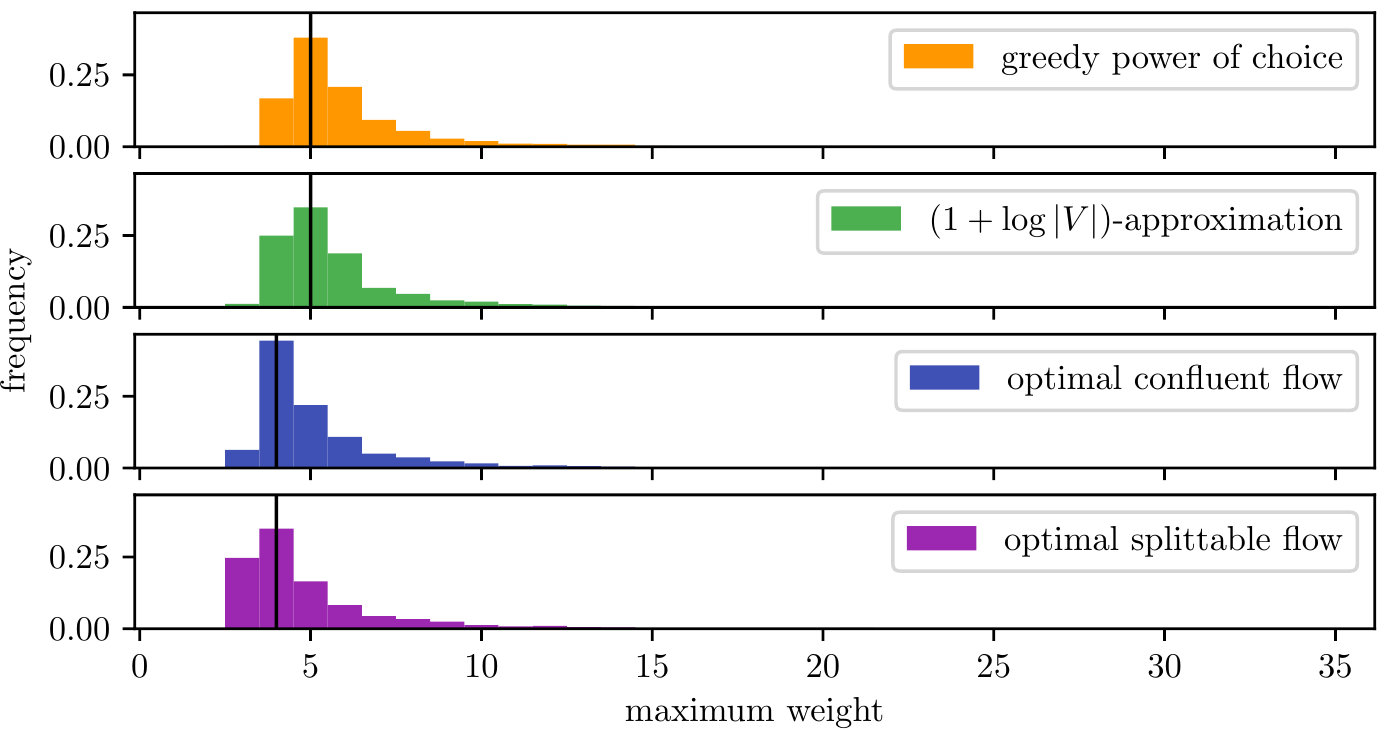}
 		\caption{Excluding random delegation}
 	\end{subfigure}%
 	\caption{Frequency of maximum weights at time $500$ over $1\,000$ runs: $\gamma = 1$, $d = 0.25$, $k = 2$.}
 	\label{fig:histg100d25}
 \end{figure}

 \begin{figure}[H]
 	\centering
 	\begin{subfigure}[h]{0.49\textwidth}
 		\includegraphics[width=\linewidth]{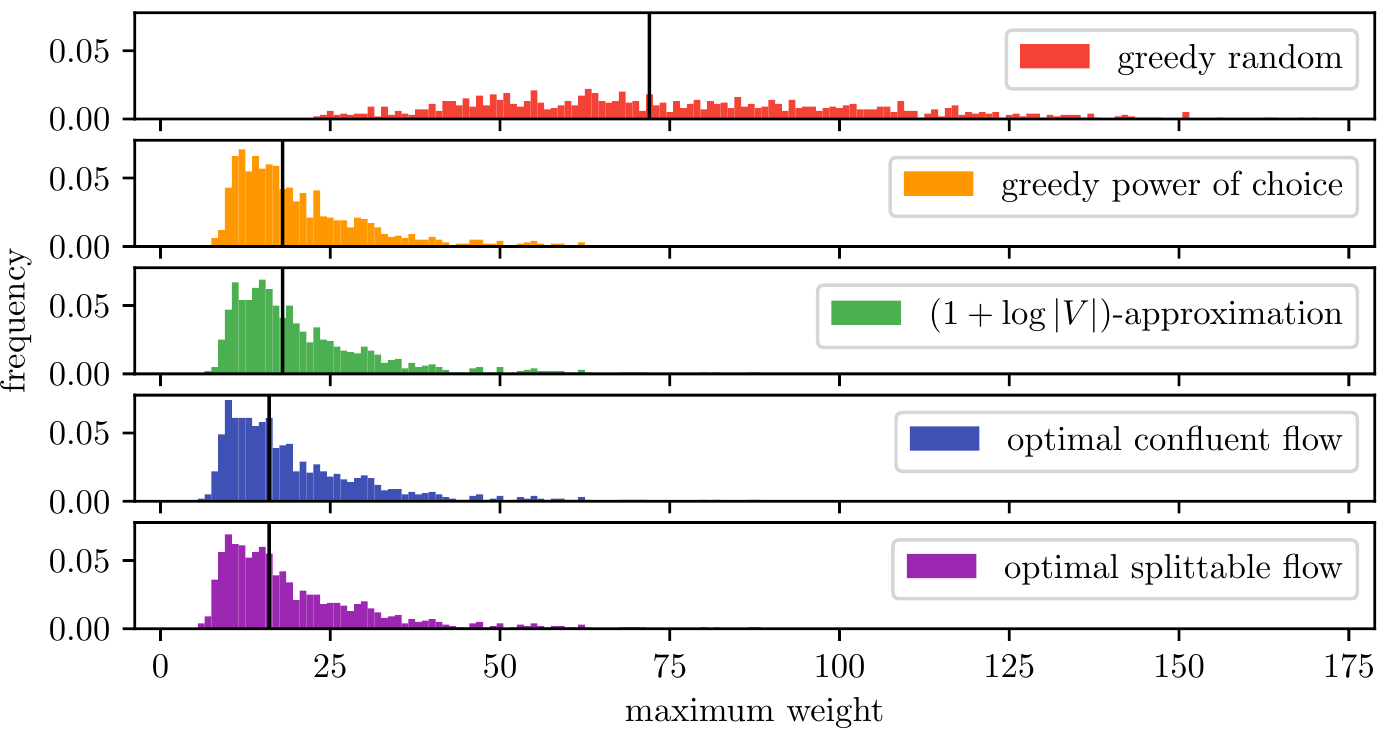}
 		\caption{All mechanisms and optimal splittable solution}
 		\label{subfig:histg100d50all}
 	\end{subfigure}
 	\begin{subfigure}[h]{0.49\textwidth}
 		\includegraphics[width=\linewidth]{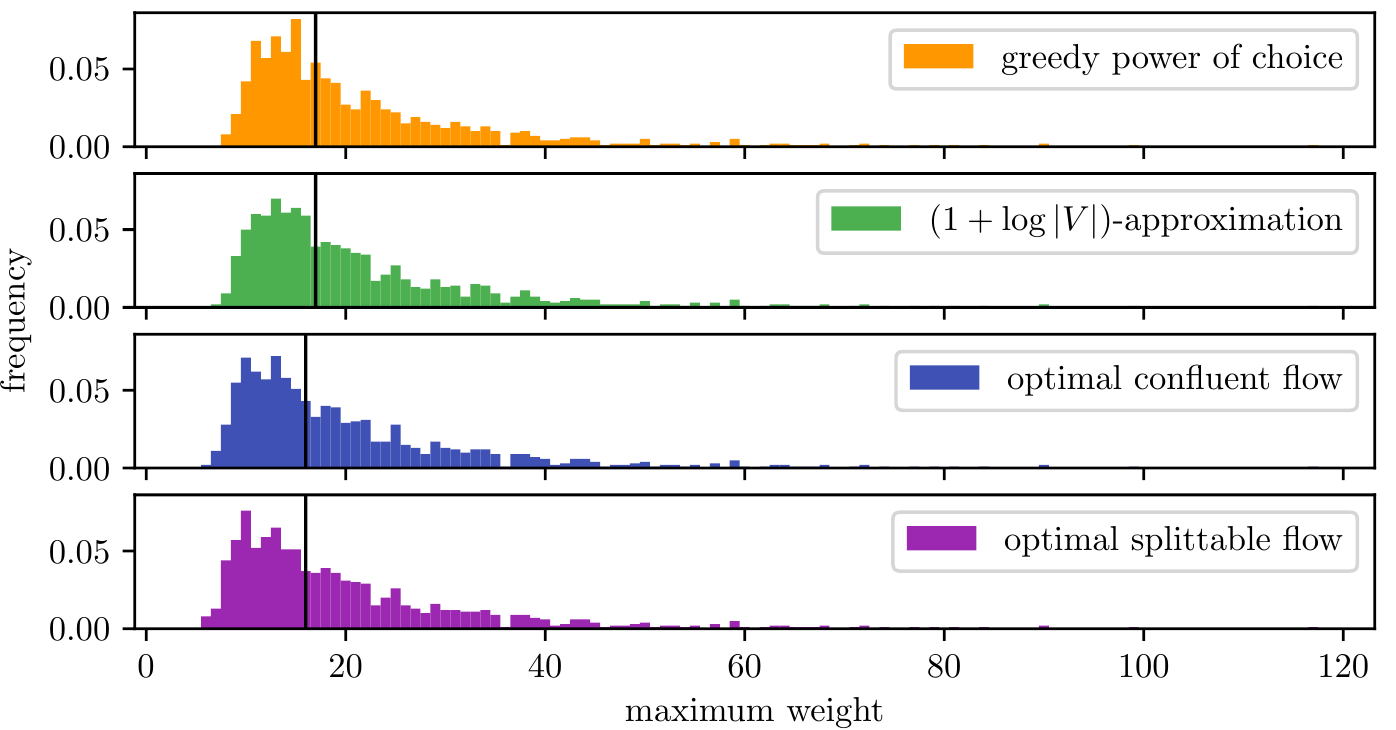}
 		\caption{Excluding random delegation}
 		\label{subfig:histg100d50norandom}
 	\end{subfigure}%
 	\caption{Frequency of maximum weights at time $500$ over $1\,000$ runs: $\gamma = 1$, $d = 0.5$, $k = 2$.}
 	\label{fig:histg100d50app}
 \end{figure}

 \begin{figure}[H]
 	\centering
 	\begin{subfigure}[h]{0.49\textwidth}
 		\includegraphics[width=\linewidth]{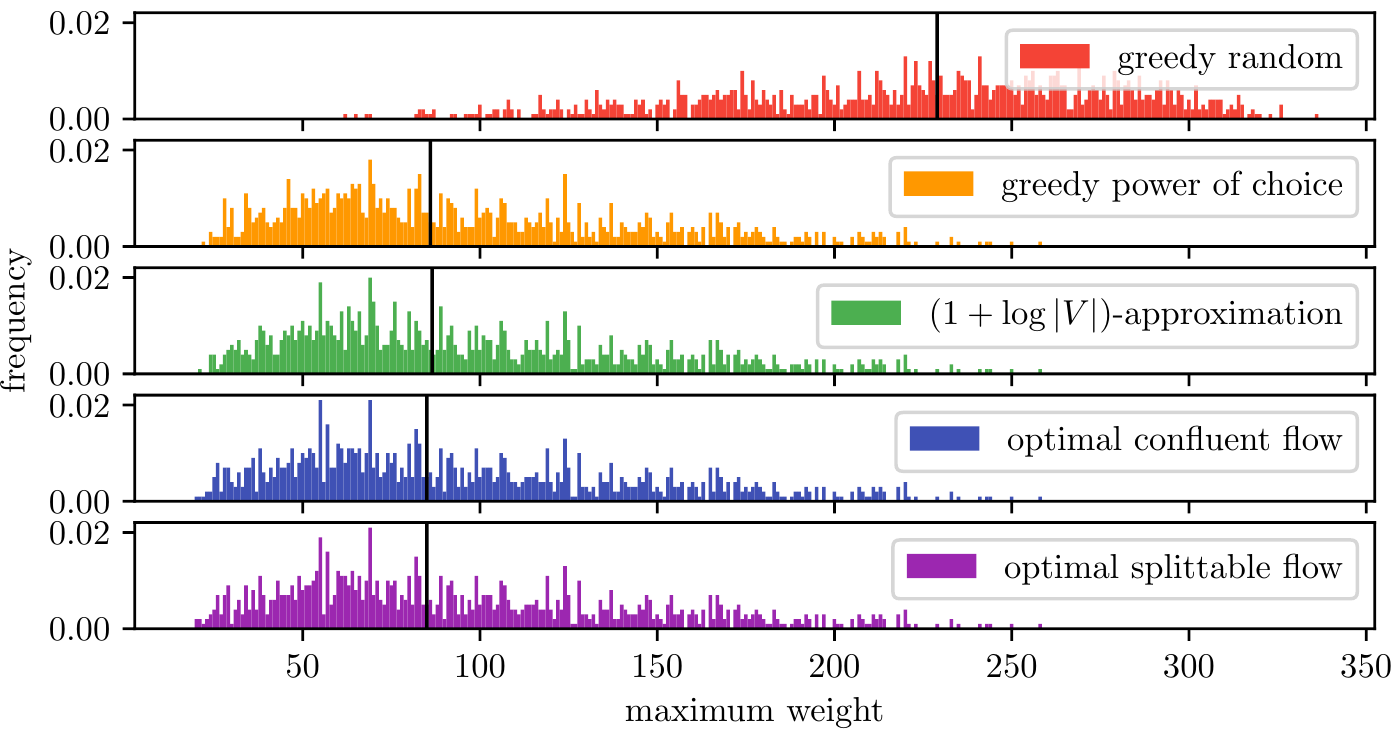}
 		\caption{All mechanisms and optimal splittable solution}
 	\end{subfigure}
 	\begin{subfigure}[h]{0.49\textwidth}
 		\includegraphics[width=\linewidth]{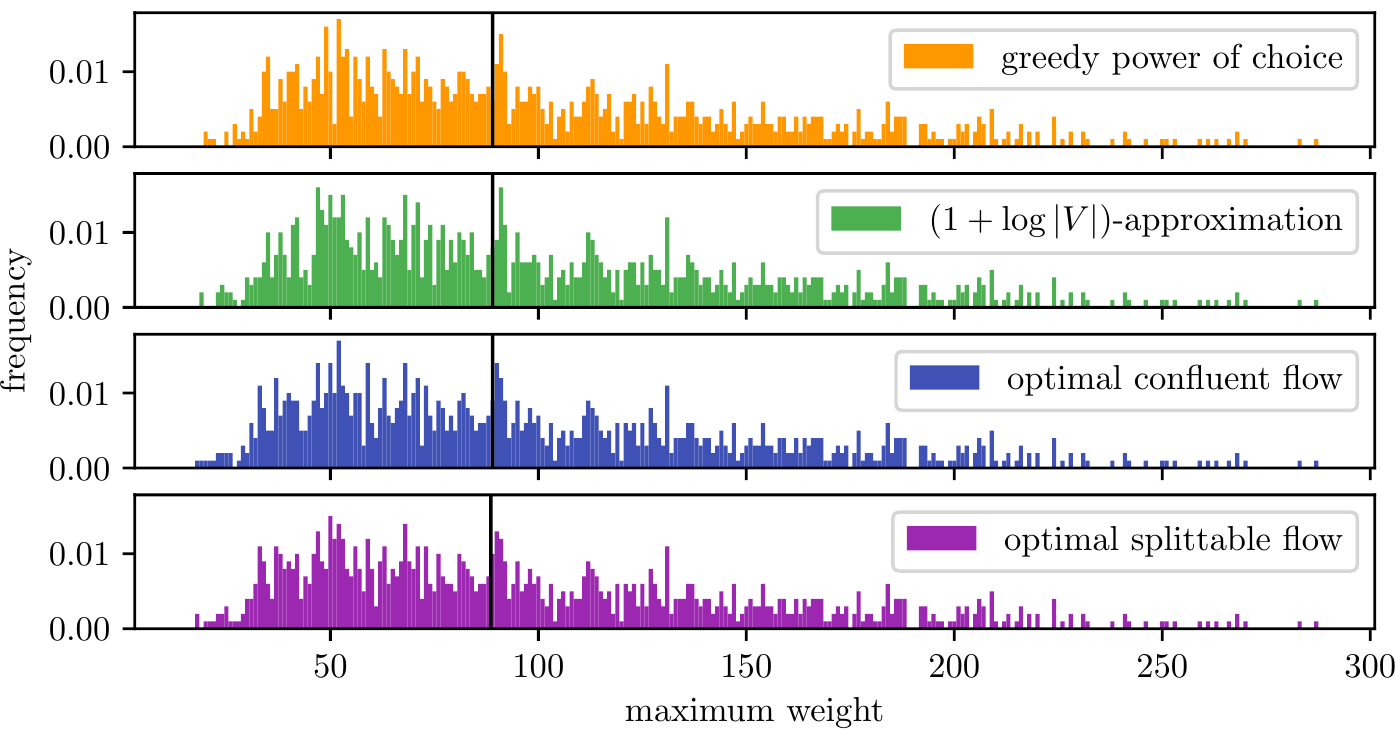}
 		\caption{Excluding random delegation}
 	\end{subfigure}%
 	\caption{Frequency of maximum weights at time $500$ over $1\,000$ runs: $\gamma = 1$, $d = 0.75$, $k = 2$.}
 	\label{fig:histg100d75}
 \end{figure}

\subsection{Running Times}
\label{sec:runtimegraphs}
\begin{figure}[H]
\centering
\begin{subfigure}[h]{.45\textwidth}
\includegraphics[width=\textwidth]{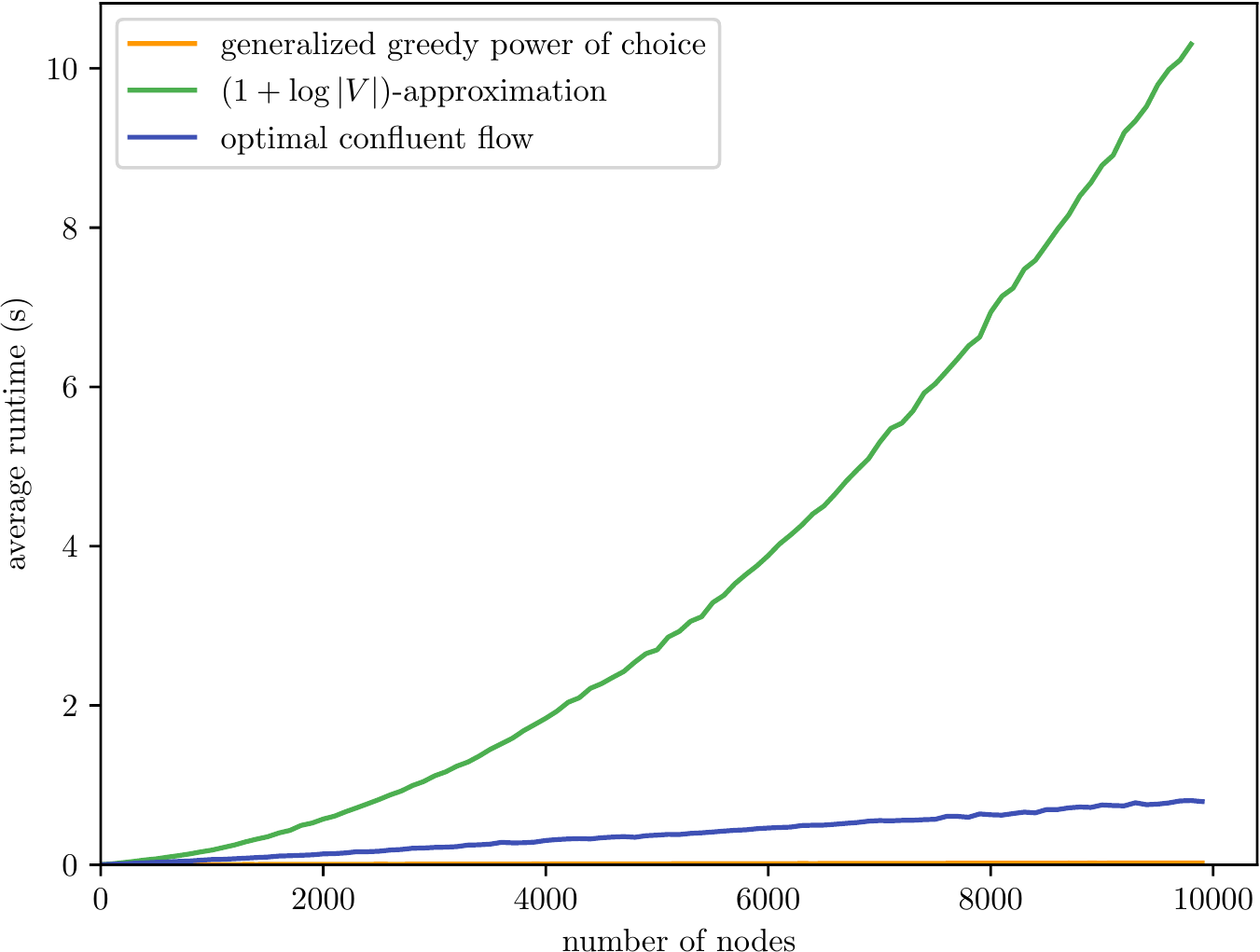}
\caption{$\gamma = 0$, $d = 0.25$}
\end{subfigure}
\begin{subfigure}[h]{.45\textwidth}
\includegraphics[width=\textwidth]{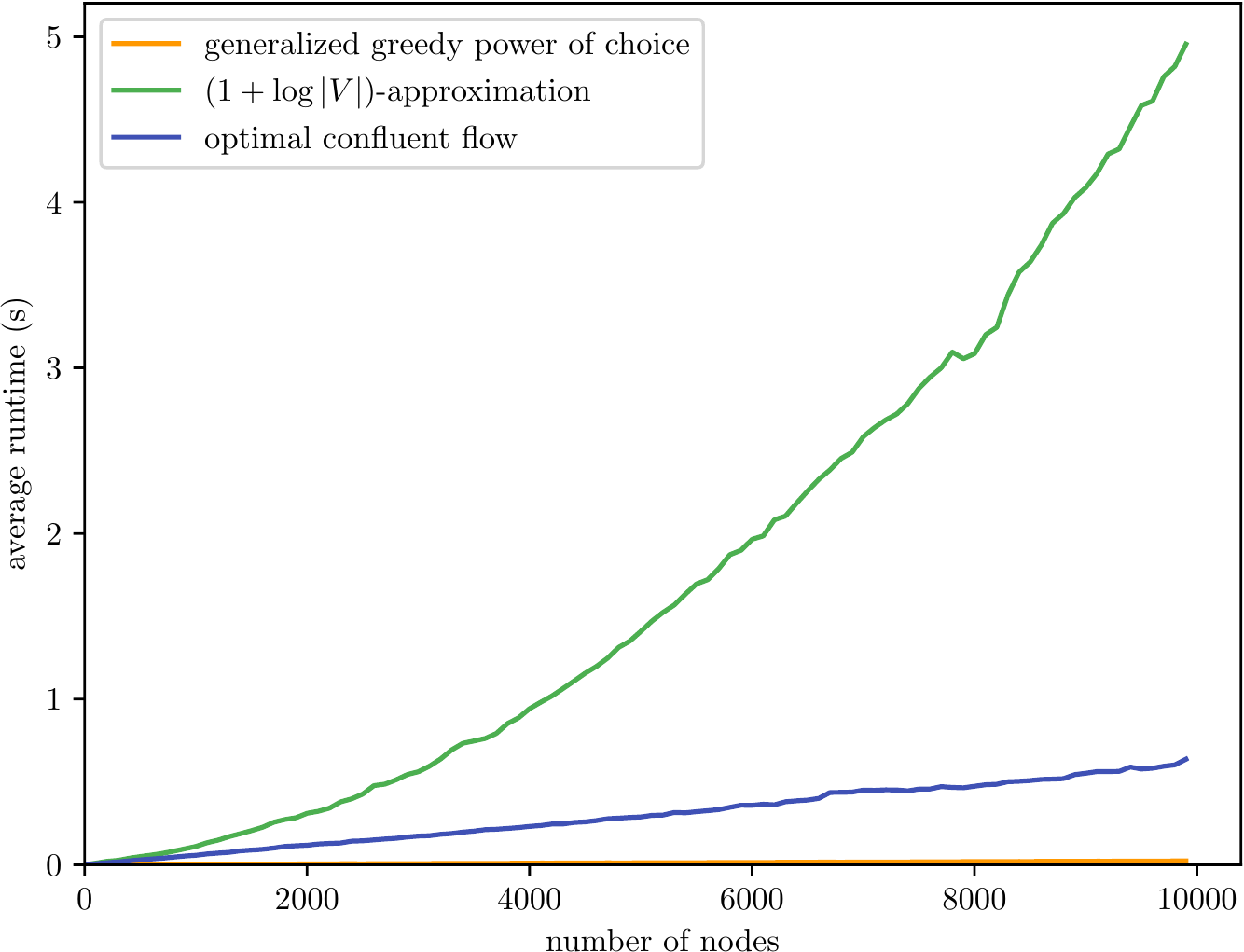}
\caption{$\gamma = 1$, $d = 0.25$}
\end{subfigure} \\
    \begin{subfigure}[h]{.45\textwidth}
\includegraphics[width=\textwidth]{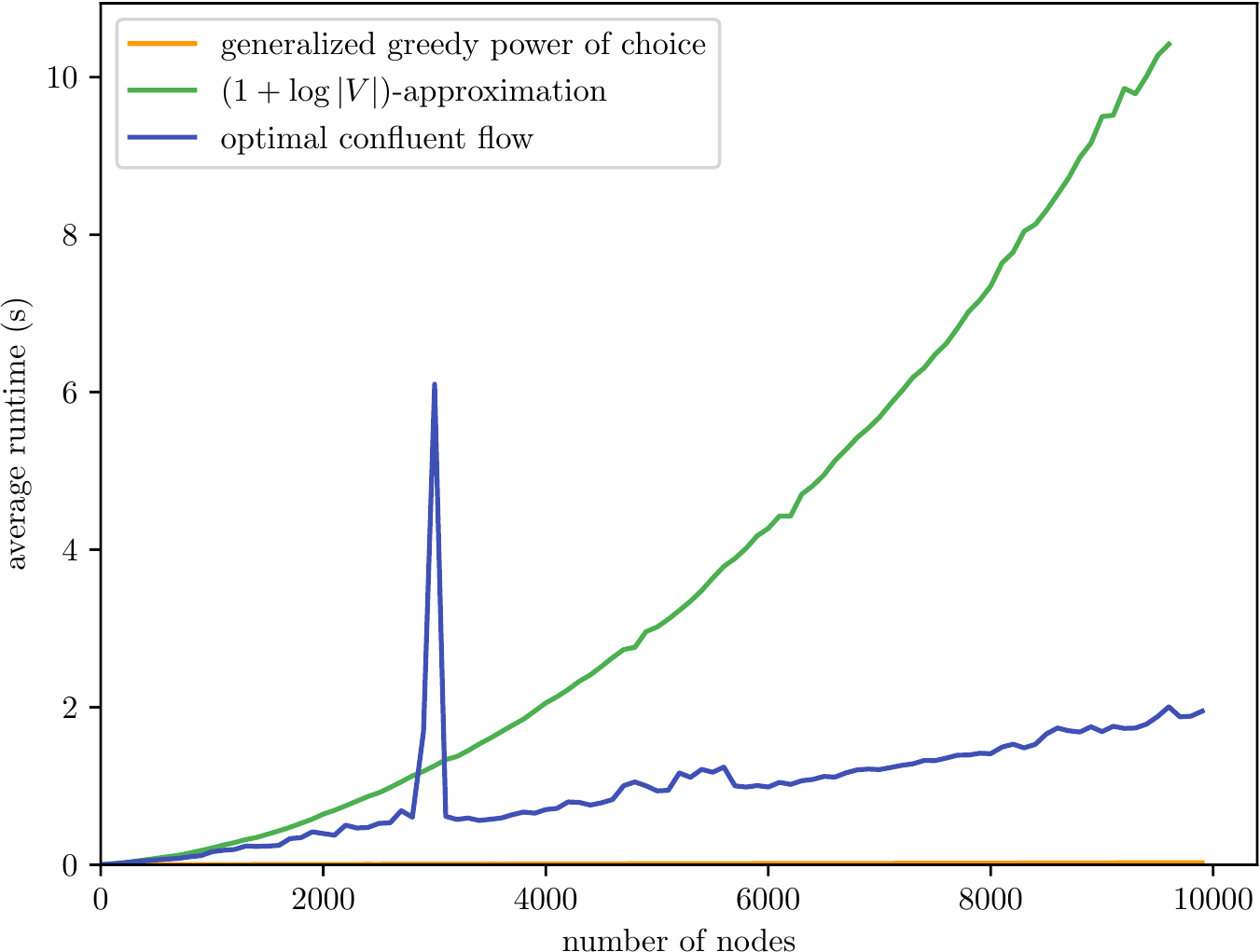}
\caption{$\gamma = 0$, $d = 0.5$}
\end{subfigure}
    \begin{subfigure}[h]{.45\textwidth}
\includegraphics[width=\textwidth]{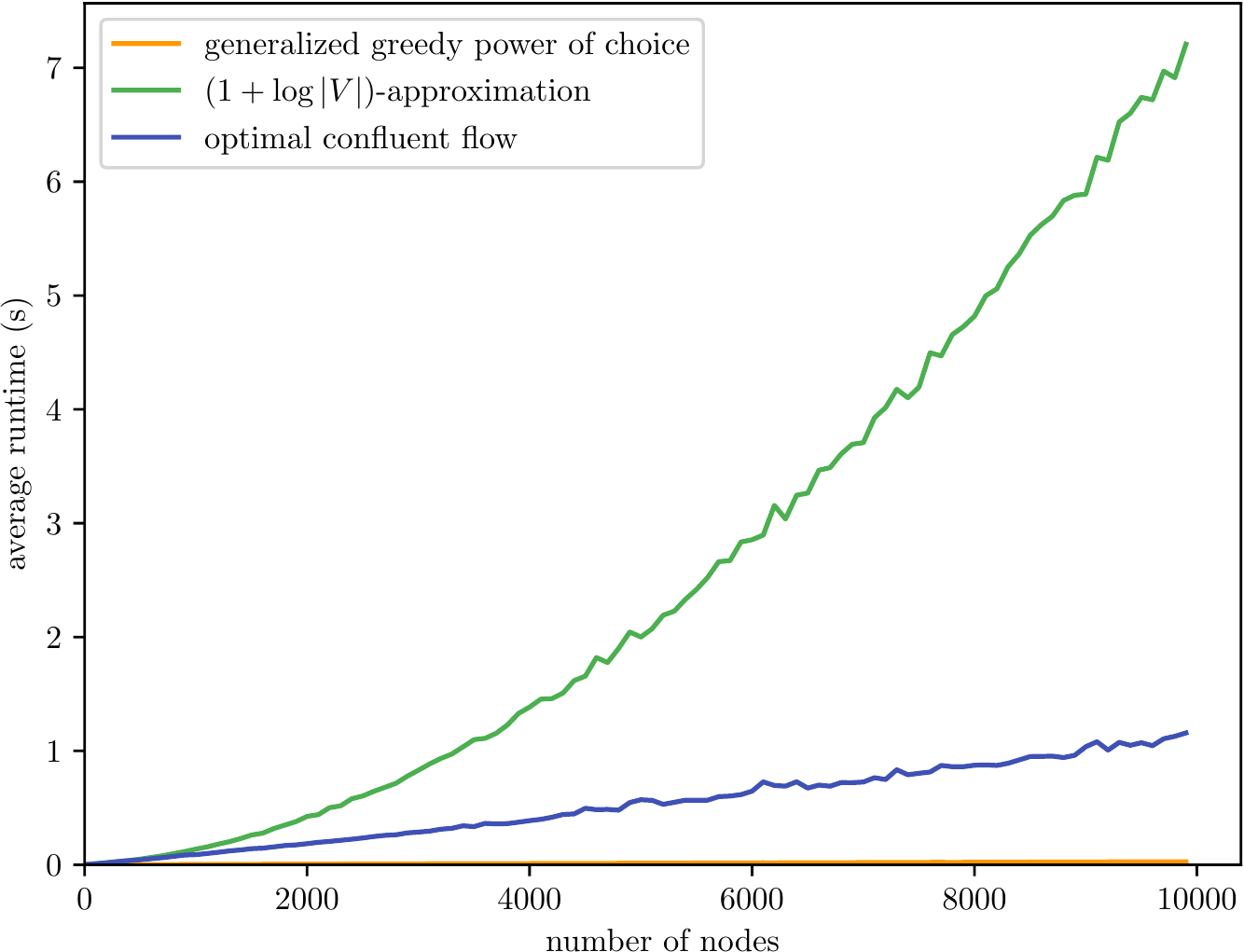}
\caption{$\gamma = 1$, $d = 0.5$}
\end{subfigure} \\
\begin{subfigure}[h]{.45\textwidth}
\includegraphics[width=\textwidth]{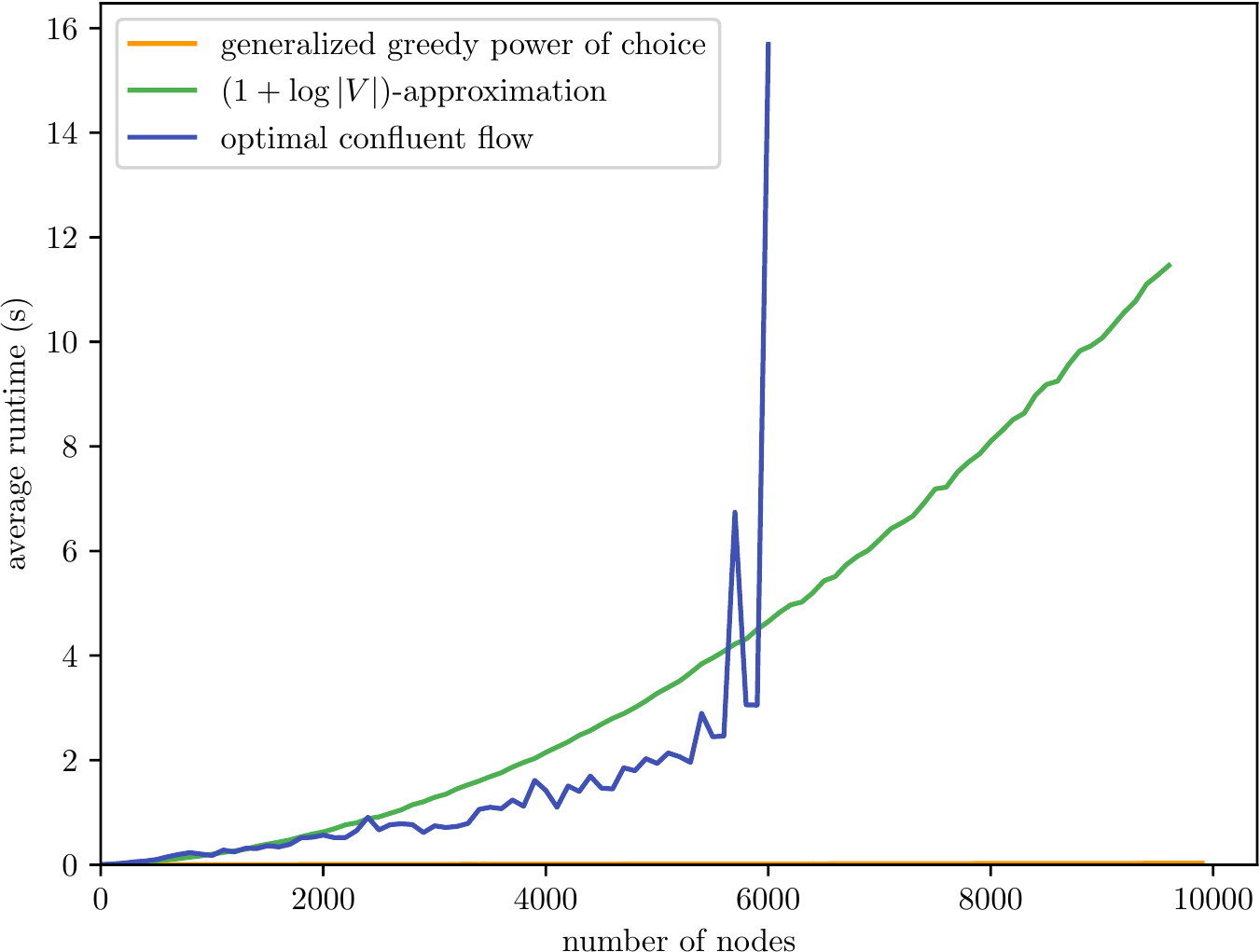}
\caption{$\gamma = 0$, $d = 0.75$}
\end{subfigure} 
\begin{subfigure}[h]{.45\textwidth}
\includegraphics[width=\textwidth]{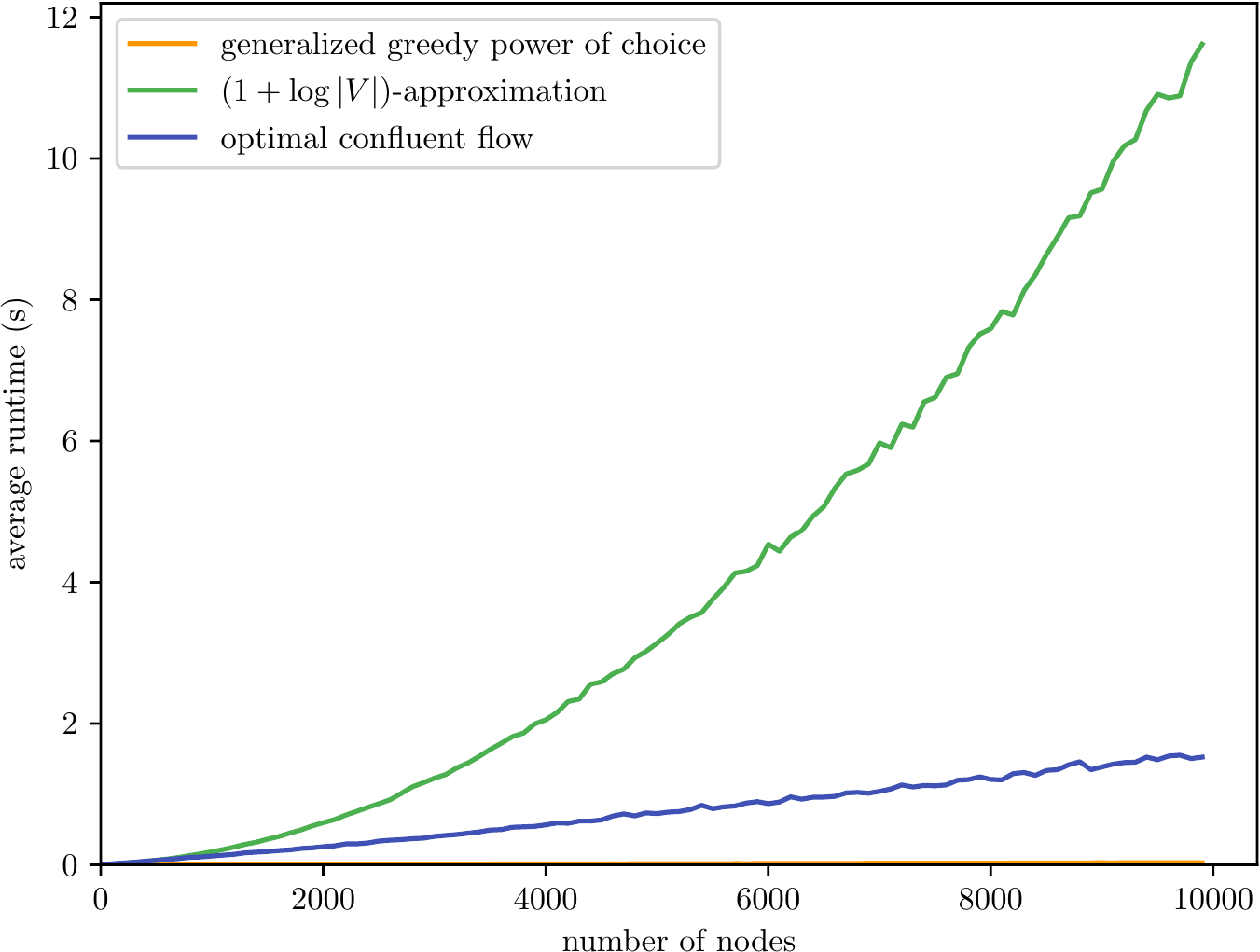}
\caption{$\gamma = 1$, $d = 0.75$}
\end{subfigure}
    \caption{Running time of mechanisms, averaged over 20 simulations. Running time computed every 20 steps of the simulation. Curve for a mechanism $m$ is discontinued at size $s$ if the mechanism needed more than 8~minutes total to resolve delegations of size $1, 21, 41, \dots, s$ in one simulation.}
    \label{fig:runtimefig}
\end{figure}
 
\end{document}